\documentclass[11pt]{article}
\usepackage{fullpage}
\usepackage{graphicx}
\usepackage{subfigure}
\usepackage{amsmath,amssymb,amsthm,amsfonts,float}
\usepackage{frame,color}
\usepackage{hyperref}
\usepackage{framed}
\usepackage{comment}
\usepackage{indentfirst}
\usepackage{enumerate}
\usepackage{tikz}
\tikzset{elegant/.style={smooth,thick,samples=50,cyan}}
\tikzset{eaxis/.style={->,>=stealth}}

\theoremstyle{plain}
\newtheorem{theorem}{Theorem}[section]
\newtheorem{lemma}[theorem]{Lemma}

\newtheorem{corollary}[theorem]{Corollary}

\newtheorem{assumption}{Assumption}
\newtheorem{conjecture}{Hypothesis}

\theoremstyle{definition}
\newtheorem{definition}{Definition}[section]

\newcommand{\ignore}[1]{}
\newcommand{\E}{\operatorname{E}}
\newcommand{\nnz}{\operatorname{nnz}}
\newcommand{\poly}{\operatorname{poly}}
\newcommand{\Var}{\operatorname{Var}}

\newcommand{\floor}[1]{\lfloor #1 \rfloor}
\newcommand{\polylog}{\operatorname{polylog}}

\newcommand{\norm}[1]{{\| #1 \|}}
\newcommand{\cE}{\mathcal{E}}
\DeclareMathOperator{\argmin}{argmin}

\newcommand\qm{q_{\max}}
\newcommand\hm{h_{\max}}
\newcommand\hW{{\hat W}}
\newcommand\bW{{\bar W}}
\newcommand\LL{{\Lambda}}

\newcommand\cT{{\cal T}}
\newcommand{\CP}{{\mathrm{CP}}}

\DeclareMathOperator{\colspan}{\mathbf{im}}
\newcommand\cS{{\cal S}}
\newcommand\xupper{U}
\newcommand\hM{{\hat M}}
\newcommand{\cN}{\mathcal{N}}
\newcommand\alphap{2}
\newcommand\vssum[1]{\;\sum_{\mathclap{\substack{#1}}}\;}
\DeclareMathOperator{\Vol}{Vol}
\newcommand{\eps}{\varepsilon}
\newcommand{\R}{{\mathbb R}}

\newcommand{\deltastruct}{\delta_{\mathsf{struct}}}
\newcommand{\deltalewis}{\delta_{\mathsf{lewis}}}
\newcommand{\deltase}{\delta_{\mathsf{subspace}}}
\newcommand{\deltaos}{\delta_{\mathsf{o}}}

\newcommand{\niceremark}[3]{}
\newcommand{\badremark}[3]{}

\newif\ifarxiv
\arxivtrue

\title{Dimensionality Reduction for Tukey Regression\thanks{Ruosong Wang and David P. Woodruff were supported in part by Office of Naval Research (ONR) grant N00014-18-1-2562. Part of this work was done while the authors were visiting the Simons Institute for the Theory of Computing.}}
\date{}
\author{Kenneth L. Clarkson \\ IBM Research - Almaden \\ \tt{klclarks@us.ibm.com} \\
\and
Ruosong Wang \\Carnegie Mellon University \\ \tt{ruosongw@andrew.cmu.edu}
\and 
David P. Woodruff \\Carnegie Mellon University \\ \tt{dwoodruf@cs.cmu.edu}}
\date{}
\begin{document}

\begin{titlepage}
\maketitle
\thispagestyle{empty}
\begin{abstract}
We give the first dimensionality reduction methods
for the overconstrained Tukey regression problem. 
The Tukey loss function $\norm{y}_M = \sum_i M(y_i)$
has $M(y_i)\approx |y_i|^p$
for residual errors $y_i$ smaller than a prescribed threshold $\tau$, but $M(y_i)$
becomes constant for errors $|y_i| > \tau$.

Our results depend on a new structural result, proven constructively, showing
that for any $d$-dimensional subspace  $L\subset \R^n$, there is a fixed bounded-size subset
of coordinates containing, for every $y\in L$, all the large coordinates, {\it with respect to the Tukey loss function}, of~$y$. We think of these as ``residual leverage scores'', since the coordinates in $y$ itself may have very different magnitude even though they both contribute the same value $\tau$ to the $M$-function. 

%This makes the Tukey measure
%more robust to outliers than $\ell_p$ norms, as well as other measures such
%as the Huber loss function which have at least linear growth.
Our methods reduce a given Tukey regression problem
to a smaller weighted version, whose solution is a provably
good approximate solution to the original problem.
Our reductions are simple and easy to implement,
and we give empirical results demonstrating their practicality, using existing heuristic solvers for the small versions.

One of our reductions uses row sampling, for
an instance $\min_{x \in \mathbb{R}^d} \|Ax-b\|_M$,
%of overconstrained Tukey regression,
where $A\in\R^{n\times d}$ and $b\in\R^n$, with $n \gg d$.
The algorithm takes $\widetilde{O}(\nnz(A) + \poly(d \log n /\varepsilon))$ time
to return a weight vector with $\poly(d \log n /\varepsilon)$ non-zero entries,
such that the solution of the resulting weighted Tukey regression problem
is a $(1 + \varepsilon)$-approximate solution.
Here  $\nnz(A)$ is the number of non-zero entries of~$A$.
Another reduction
% with the similar runtime
%and size reduction 
uses a sketching matrix $S$,
chosen independently of $A$ and $b$,
such that $SA$ and $Sb$ yield an $O(\log n)$ approximation.
so that the solution for a weighted version with inputs $SA,Sb$
is an $O(\log n)$-approximate solution.
Here $S$ has $\poly(d \log n)$ rows and $SA$ and $Sb$
are computable in $O(\nnz(A))$ time.
%By slightly modifying the estimator on the small problem,
%we achieve a $C$-approximation for a constant $C \geq 1$. 

We also give exponential-time algorithms giving provably good solutions,
and hardness results suggesting that a significant speedup in the worst case is unlikely. 
\end{abstract}

\end{titlepage}

% !TEX root = arxiv.tex
\section{Introduction}

A number of problems in numerical linear algebra have witnessed remarkable
speedups via the technique of linear sketching. Such speedups are made possible
typically by reductions in the dimension of the input
(here the number of rows of the input matrix),
whereby a large scale optimization problem
is replaced by a much smaller optimization problem, and then a slower algorithm
is run on the small problem. It is then argued that the solution to the smaller
problem provides an approximate solution to the original problem. We refer the
reader to several recent surveys on this topic \cite{kv09,m11,w14}.

This approach has led to optimal algorithms for approximate overconstrained
least squares regression: given an $n \times d$ matrix $A$, with $n \gg d$, and
an $n \times 1$ vector $b$, output a vector $x' \in \mathbb{R}^d$ for which
$\|Ax'-b\|_2 \leq (1+\eps) \min_x \|Ax-b\|_2$. For this problem, one first
samples a random matrix $S \in \mathbb{R}^{k \times n}$ with a small number $k$
of rows, and replaces $A$ with $S \cdot A$ and $b$ with $S \cdot b$. Then one
solves (or approximately solves) the small problem $\min_x \|SAx-Sb\|_2$. The
goal of this \emph{sketch and solve} approach
is to choose a distribution $S$ so that if $x'$ is the minimizer to this
latter problem, then one has that $\|Ax'-b\|_2 \leq (1+\eps)\min_x \|Ax-b\|_2$
with high probability. Note that $x' = (SA)^+Sb$, where $(SA)^+$ denotes the
pseudoinverse of $SA$
%(see, e.g., \cite{}, for a survey)
%\Ken{ref: give refs for pseudoinverse?}
and can be computed in $kd^2$ time, see, e.g., \cite{w14} for a survey and
further background. Consequently, the overall time to solve least squares
regression is $T + kd^2$, where $T$ is the time to compute $S \cdot A$ and $S
\cdot b$. Thus, the goal is to minimize both the time $T$ and the sketching
dimension $k$. Using this approach, S\'arlos showed \cite{s06} how to achieve
$O(nd \log n) + \poly(d/\eps)$ overall time, which was subsequently improved to
the optimal $\nnz(A) + \poly(d/\eps)$ time in \cite{CW13,MengMahoney,NN13}.

Recently, a number of works have looked at more robust versions of regression.
Sohler and Woodruff \cite{sw11}, building off of earlier work of Clarkson
\cite{c05} (see also \cite{ddhkm09}), showed how to use the sketch and solve
paradigm to obtain a $(1+\eps)$-approximation to $\ell_1$ regression, namely, to
output a vector $x' \in \mathbb{R}^d$ for which
$\|Ax'-b\|_1 \leq (1+\eps)\min_x \|Ax-b\|_1$.
This version of regression, also known as least absolute deviations,
is known to be less sensitive to outliers than least squares
regression, since one takes the absolute values of the errors in the residuals
rather than their squares, and so does not try to fit outliers as much. By now,
we also have optimal $\nnz(A) + \poly(d/\eps)$ time algorithms for $\ell_1$
regression \cite{li2013iterative,cfast,wz13,CW13,MengMahoney,wang2019tight}, for
the related quantile regression problem \cite{ymm13}, and for $\ell_p$
regression for every $p \geq 1$ \cite{ddhkm09,wz13,cp15}.

In this paper we consider more general
overconstrained regression problems: given an $n \times d$ matrix $A$ with
$n \gg d$, and an $n \times 1$ vector $b$, output a vector $x' \in \mathbb{R}^d$
for which $\|Ax'-b\|_M \leq (1+\eps)\min_x \|Ax-b\|_M$, where for an
$n$-dimensional vector $y$ and a function $M:\mathbb{R} \rightarrow
\mathbb{R}^{+}$, the notation $\|y\|_M$ denotes $\sum_{i=1}^n M(y_i)$. If
$M(x) = x^2$ then we have the least squares regression problem, while if
$M(x) = |x|$ we have the least absolute deviations problem.

Clarkson and Woodruff showed
\cite{cw15} that for any function~$M(\cdot)$ which has at least linear (with positive
slope) and at most quadratic growth, as well as some natural other properties,
there is a distribution on sketching matrices $S$ with $k = \poly(d \log n)$
rows and corresponding vector $w \in \mathbb{R}^k$, with the following
properties. The product $S \cdot A$ can be computed in $\nnz(A)$ time, and if
one solves a weighted version of the sketched problem,
$\min_x \sum_{i=1}^k w_i M((SAx-b)_i)$, then the minimizer is a constant-factor approximation to the
original problem. 
This gives an algorithm with overall $\nnz(A) + \poly(d)$ running time for the important
Huber loss function: given a parameter $\tau$, we have 
$$
M(x) = \begin{cases}
x^2/(2\tau) & |x| \leq \tau\\
|x|-\tau/2 & \text{otherwise}
\end{cases}.
$$
Unlike least absolute
deviations, the Huber loss function is differentiable at the origin, which is
important for optimization. However, like least absolute deviations, for large
values of $x$ the function is linear, and thus pays less attention to outliers.
Other loss functions similar to Huber, such as the $\ell_1-\ell_2$ and Fair
estimators, were also shown to have $\nnz(A) + \poly(d)$ time algorithms. These
results were extended to $(1+\eps)$-approximations via sampling-based techniques
in \cite{cw15b}. 

Despite the large body of $M$-functions handled by previous work, a notable
well-studied exception is the Tukey loss function \cite{fox2002robust}, with 
\begin{equation}\label{eq:tukey_loss}
M(x) = \begin{cases}
	\frac{\tau^2}{6} (1-[1-\left(\frac{x}{\tau}\right)^2]^3) & |x| \leq \tau \\
	\frac{\tau^2}{6} & \text{otherwise}
	\end{cases}.
\end{equation}
\begin{figure}
\centering
\ifarxiv
\begin{tikzpicture}
     % draw the axis
    \draw[eaxis] (-2,0) -- (2,0) node[below] {$x$};
    \draw[eaxis] (0,-0.5) -- (0,2) node[above] {$M(x)$};
     % draw the function (piecewise)
    \draw[elegant,domain=-2:-1] plot(\x,1);
    \draw[elegant,domain=1:2] plot(\x,1);
   \draw[elegant,domain=-1:1] plot(\x,{1-(1-(\x)^2)^3});
    %\draw[elegant,orange,domain=-\num:\num] plot(\x,{sin(\x r)});
\end{tikzpicture}
\else
\includegraphics[scale=0.7]{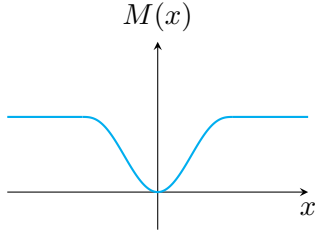}
\fi
\caption{The Tukey loss function.}
\label{fig:tukey}
\end{figure}
See Figure~\ref{fig:tukey} for a plot of the Tukey loss function.
By a simple Taylor expansion, it is easy to see that $M(x) = \Theta(x^2)$ for
$|x| \leq \tau$ and $M(x) = \Theta(\tau^2)$ otherwise. While similar to Huber in
the sense that it is quadratic near the origin, it is even less sensitive to
outliers than Huber since it is constant when one is far enough away from the
origin. Thus, it does not satisfy the linear growth (with positive slope)
requirement of \cite{cw15,cw15b}. An important consequence of this distinction
is that, while for $M$-functions with linear growth, a single outlier ``at
infinity'' can radically affect the regression output, this is not the case for
the Tukey loss function, due to the bound on its value.

Although the Tukey loss function is not convex, a local minimum of it can be
found via iteratively reweighted least squares \cite{fox2002robust}. Also, the
dimensionality reduction approach still makes sense for Tukey: here a large non-convex
problem is reduced to a smaller non-convex one. Our reduction of a non-convex
problem to a smaller one is arguably even more interesting than reducing the
size of a convex problem, since inefficient algorithms may now be efficient on the much smaller problem.                                                                    

\paragraph{Notation.}
For a matrix $A \in \mathbb{R}^{n \times d}$, we use $A_{i, *}$ to denote its
$i$-th row, $A_{*, j}$ to denote its $j$-th column, and $A_{i, j}$ to denote a
specific entry. For a set of indices $\Gamma \subseteq [n]$, we use
$A_{\Gamma, *}$ to denote the submatrix of $A$ formed by all rows in $\Gamma$. Similarly, we
use $A_{*, \Gamma}$ to denote the submatrix formed by all columns in $\Gamma$.
For a matrix $A \in \mathbb{R}^{n \times d}$ and a vector $b \in \mathbb{R}^n$, we use $[A~b] \in \mathbb{R}^{n \times (d + 1)}$ to denote the matrix whose first $d$ columns are $A$ and the last column is $b$. 
For a matrix $A \in \mathbb{R}^{n \times d}$, we use $\colspan(A) = \{Ax \mid x\in\R^d\}$ to denote the column span of $A$.

\subsection{Our Assumptions}

Before stating our results, we give the general assumptions our algorithms and analyses
need.
We need the following assumptions on the loss function.
\begin{assumption}\label{as M}
There exist real numbers $\tau \ge 0$, constants $p \ge 1$ and $0 < L_M \le 1 \le U_M$
such that the function $M: \R \rightarrow \R^+$ satisfies:
\begin{enumerate}
\item \label{as M sym} Symmetric: $M(a) = M(-a)$ for all $a$.
\item \label{as M inc} Nondecreasing: $M(a) \ge M(a')$ for $|a| \ge |a'|$.
\item \label{as M growth} Growth condition: for $|a| \ge |a'|$, 
$$
\left|\frac{a}{a'}\right|^p \ge \frac{M(a)}{M(a')}.
$$
\item \label{as M near quad} Nearly $p$ th power: for all $|a| \le \tau$, %)
$$
	L_M |a|^p \le M(a) \le U_M |a|^p.
$$
\item \label{as M flat} Mostly flat: $M(a) = \tau^p$ for $|a|\ge \tau$.
\end{enumerate}
\end{assumption}
The conditions in Assumption \ref{as M} state that our loss function essentially behaves as an $\ell_p$ loss function $M(a) = |a|^p$ for $a \leq \tau$, at which point $M(a) = \tau^p$. However, the conditions are more robust in that $M(a)$ just needs to agree with $|a|^p$ up to a fixed constant factor. This is a non-trivial extension of the $\ell_p$ loss function since we will obtain $(1+\varepsilon)$-approximations in our algorithms, and consequently cannot simply replace the $M$ loss function in our problem with an $\ell_p$ loss function, as this would result in a constant factor loss. Moreover, such an extension is essential for capturing common robust loss functions, such as the Tukey loss function (\ref{eq:tukey_loss}) above: $M(a) = a^2 (1-(1-(a/\tau)^2)^3)$ if $|a| < \tau$, and $M(a) = \tau^2$ for $|a| \geq \tau$ (we  note that sometimes this loss function is divided by the number $6$, but this plays no role from a minimization perspective). Note that the Tukey loss function $M(a)$ does not coincide with the $\ell_2$ loss function for $|a| \le \tau$, though it is within a constant factor of it so our conditions can handle it. For a recent use of this loss function for regression in applications to deep learning, see, e.g., \cite{b15}.

For $\tau = \infty$, we indeed have $M(a) = |a|^p$ and $\ell_p$ loss functions are widely studied for regression. The case $p = 2$ is just ordinary least squares regression.  For $1 \leq p < 2$, the $\ell_p$ loss function is less sensitive than least squares since one is not squaring the errors, and therefore for such $p$ the loss function is considered to be more robust than ordinary least squares. For $p > 2$, and in particular large values of $p$, the $\ell_p$ regression solution approaches the $\ell_{\infty}$ regressioin solution, which minimizes the maximum error. The $\ell_p$ loss functions for every $p \geq 1$ are well-studied in the context of regression, see, e.g., \cite{c05,ddhkm09,cfast,wz13,cp15}. Loss functions with the polynomial growth condition (Assumption \ref{as M}.\ref{as M growth}) are also well-studied \cite{cw15,cw15b}. We also note that already for $\ell_p$ loss functions, all known dimensionality reduction techniques reduce the original regression problem to a problem of size at least $d^{\Omega(p)}$, which is recently shown to be necessary~\cite{li2019tight}.
Consequently, it is impossible to obtain a considerable dimensionality reduction when $p$ is allowed to grow with $n$. Since this is a special case of our more general setting of arbitrary $\tau$, we restrict our attention to $p$ that does not grow with $n$.  

For general $\tau$ and $p$, we obtain the Tukey $\ell_p$ loss function $M(a) = |a|^p$ for $a \leq \tau$, and $M(a) = \tau^p$ for $a \geq \tau$. Note that for $p = 1$, the loss function is the maximum likelihood estimator in the presence of i.i.d. Laplacian noise, while for $p = 2$, the loss function is the maximum likelihood estimator in the presence of i.i.d. Gaussian noise. Thus we obtain the first dimensionality reduction for loss functions that correspond to maximum likelihood estimators of classical noise models where the loss function completely saturates beyond some threshold $\tau$. This is particularly useful for handling large outliers. 

We will also need the following assumptions on the input, which we justify below.
\begin{assumption}\label{as main}
We assume:
\begin{enumerate}
\item \label{as main A int} For given $C_1 \le \poly(n)$, 
there is $\xupper = n^{O(d^2)}$ such that
$\norm{\hat x }_2\le \xupper$ for any $C_1$-approximate solution $\hat x$
of $\min_x \norm{Ax-b}_M$.
\item \label{as main size} The columns of $A$ and $b$ have $\ell_2$ norms in $n^{O(d)}$.
\item \label{as tau size} The threshold $\tau = \Omega(1/n^{O(d)})$.
\end{enumerate}
\end{assumption}

As we will show, Assumption~\ref{as main}.\ref{as main A int}  holds
when \ref{as main}.\ref{as main size} and \ref{as main}.\ref{as tau size} hold,
and the entries of $A$ are integers.
Alternatively, Assumption~\ref{as main}.\ref{as main A int}  might hold due
to the particular input given, or to an
additional explicit problem constraint,
or as a consequence of regularization.

%(Our analysis goes through without Assumption~\ref{as main}
%for a variation of the problem where $Ax$ is constrained to have norm at most $R$, for some $R$,
%and with the condition $R/\tau = n^{O(d)}$. However, we find the above formulation more useful for our analysis.)
%Note that entries in floating point would allow norms larger than $n^{O(d)}$ even with $O(d\log n)$ bits.

%David: it seems okay here.
%\Ken{This should go in the right place, maybe not here}

We need such a bound on the magnitude of the entries of the Tukey regression optimum,
since as the following example shows, they can grow quite large, and behave quite differently from
$\ell_p$ regression solutions.
Here we use the case $p = 2$ as an example. 

Suppose $A$ is the $n\times 2$ matrix
\[
A = \left[\begin{matrix}
1 & 0\\
1 & 1\\
0 & \eps \\
0 & \eps \\
\vdots\\
0 & \eps
\end{matrix}
\right]
\]
with $n-2$ rows of $[0\,\, \eps]$ for $\eps >0$. Suppose $b\in\R^n$ is the vector of all ones.
Then the $\ell_2$ regression optimum
$x^* = \argmin_x \norm{Ax-b}_2^2$ has $\norm{Ax^*-b}_2^2$ at most $n$, since $x=0$ has that cost. So
$n \geq \norm{Ax^*-b}_2^2 >= (x^*_1 - 1)^2 + (x^*_1 + x^*_2 - 1)^2$,
from the first two rows, so $(x^*_1)^2 = O(n)$, implying that also $(x^*_2)^2 = O(n)$.

It can also be shown that when $\eps<1/n$, we have the entries of the $\ell_2$
regression optimum $x^*$ in $O(1)$: as $x_1$ and $x_2$ get larger, they
contribute to the cost, due to the first two rows, more rapidly than the
decrease in the $(n-2)(1-x_2/n)^2$ cost due to the $n-2$ last rows.

However, for the Tukey loss function $\norm{Ax-b}_M$ with parameter $\tau = 1/2$
and $p = 2$, the cost for $x=[1\,\, 1/\eps]$ is at most a constant, since the
contributions for all rows of $A$ but the second is zero, and contribution made
by the second row is at most a constant. However if $x_2 < 1/(2\eps)$, the cost
is $\Omega(n)$, since all but the first two rows contribute $\Omega(1)$ to the
cost. Thus the optimal $x$ for the Tukey regression has entries
$\Omega(1/\eps)$, and no $x$ with entries $o(1/\eps)$ is a constant-factor
approximation.

Indeed, given an upper bound on the entries of $x$, for any $n' < n-2$ there is
a large enough version of our example such that no $x$ satisfying that bound can
be within an $n'$ factor of optimal for the Tukey regression problem. 
This example is in
fact a near-optimal separation, as one can show that the $\ell_2$ regression
solution always provides an $O(n)$-approximate solution when $p = 2$.

\subsection{Our Contributions}
We show that under the assumptions mentioned above, it is possible 
to obtain dimensionality
reductions for Tukey regression. All of our results hold with
arbitrarily large constant probability, which can be amplified by independent
repetitions and taking the best solution found.

\paragraph{Row Sampling Algorithm.}
We first give a row sampling algorithm for Tukey
regression. 
\begin{theorem}
Given matrix $A \in \mathbb{R}^{n \times d}$ and
vector $b \in \mathbb{R}^n$, there is an algorithm that constructs a weight vector $w$ in
$\widetilde{O}(\nnz(A) + d^{p / 2} \cdot \poly(d \log n) / \varepsilon^2)$\footnote{Throughout the paper we use $\widetilde{O}(f)$ to denote $f \polylog f$, and $\widetilde{\Omega}(f)$ to denote $f / \polylog f$.} time
with $\|w\|_0 \le \widetilde{O}(d^{p / 2} \cdot \poly(d \log n) /
\varepsilon^2)$, for which if $x_{M,w}^*$ is the minimizer to $\min\sum_{i =1}^n w_i
M((Ax-b)_i)$, then 
\[
\|Ax_{M,w}^* -b\|_M \leq (1 + \varepsilon) \min_x \|Ax-b\|_M,
\]
where $M$ is the Tukey loss function. 
\end{theorem}
Since one can directly ignore those rows
$A_{i, *}$ with $w_i = 0$ when solving $\min\sum_{i =1}^n w_i M((Ax-b)_i)$, our
row sampling algorithm actually reduces a Tukey regression instance to a
weighted version of itself with
$\widetilde{O}(d^{p / 2} \cdot \poly(d \log n) / \varepsilon^2)$ rows.
Notably, the running time of the algorithm and the number of
rows in the reduced problem match those given by Lewis weights sampling
\cite{cp15}, up to $\poly(d \log n)$ factors. However, Lewis weights sampling is
designed specifically for the $\ell_p$ norm, which is a simple special case of
the Tukey loss function where $\tau = \infty$.

Our reduction is from the Tukey regression problem to a smaller,
weighted version of itself, and since known heuristics for Tukey regression can
also handle weights, we can apply them to the reduced problem as well. 

\paragraph{Oblivious Sketch.}
While the row sampling algorithm produces a $(1 + \varepsilon)$-approximate solution,
where $\varepsilon$ can be made
arbitrarily small,  the algorithm does have some properties that can be
a disadvantage in some settings: it makes $\polylog(n)$ passes over the matrix
$A$ and vector $b$, and the rows chosen depend on the input.
In the setting of streaming algorithms and distributed computation,
on the other hand, a \emph{sketch-and-solve} approach
can be more effective.
We give such an approach to Tukey regression.
\begin{theorem}
When $1 \le p \le 2$, there is a
distribution $S \in \mathbb{R}^{r \times n}$ over sketching matrices with
$r = \poly(d \log n)$, and a corresponding weight vector $w \in \mathbb{R}^r$, for
which $S \cdot A$ and $S \cdot b$ can be computed in $O(\nnz(A))$ time and for which if $x_{S,M,w}^*$
is the minimizer to $\min\sum_{i =1}^r w_i M((SAx-Sb)_i)$, then 
\[
\|Ax_{S,M,w}^* - b\|_M \leq O(\log n) \min_x \|Ax-b\|_M,
\]
where $M$ is the Tukey loss function. 
\end{theorem}
%Our sketching dimension $r$
%depends polynomially on the small dimension $d$, and only logarithmically on the
%large dimension~$n$. 

%Although our approximation factor is logarithmic, we note that a logarithmic approximation to the Tukey loss function may be much more useful than a $(1+\eps)$-approximation to other loss functions such as least squares, when outliers are an issue. Indeed, one arbitrarily large corrupted entry of $b$ will cause $x$ to fit that entry with respect to any currently known loss function for which the sketch and solve paradigm is possible; this is true even for $(1+\eps)$-approximations. In contrast, even an $O(\log n)$-approximation to Tukey will mostly ignore such an outlier. 
%First, we note that our reduction is such that the weights $w_i$ can be partitioned into a small number of groups $G$ of coordinates for which for any $i, i'$ in the same group $G$, one has $w_i = w_{i'}$. Let $w_G$ denote the common weight in group $G$ and let $(Ax-b)_G$ denote the vector obtained from $Ax-b$ by restricting to coordinates in $G$. Then, if instead of minimizing $\sum_{i=1}^k w_i M((SAx-b)_i)$, one minimizes $\sum_{\textrm{groups } G} w_G \|(Ax-b)_G\|_{KF, t}$, where for a vector $v$, the Ky-Fan norm $\|v\|_{KF}$ denotes the sum of the absolute values of its largest $t$ entries, for some parameter $t$, then the minimizer $x'$ to this problem satisfies $\|Ax'-b\|_M \leq C \min_x \|Ax-b\|_M$ for an absolute constant $C \geq 1$. 

Our sketching matrices $S$ are {\em oblivious},
meaning that their distribution {\em does not} depend on the data matrix
$A$ and the vector $b$. Furthermore, applying the sketching matrices requires
only one pass over these inputs, and thus can be readily
implemented in streaming and distributed settings.

We further show
that the same distribution $S$ on sketching matrices gives a fixed constant $C
\geq 1$ approximation factor if one slightly changes the regression problem
solved in the reduced space. 

We also remark that for oblivious sketches, the condition
that $p \le 2$ is necessary, as shown in \cite{braverman2010zero}.

\paragraph{Hardness Results and Provable Algorithms.}
We give a reduction from
$\mathsf{MAX}$-$\mathsf{3SAT}$ to Tukey regression, which implies the
\textsf{NP}-Hardness of Tukey regression. Under the Exponential
Time Hypothesis \cite{impagliazzo2001complexity}, using Dinur's PCP Theorem~\cite{dinur2007pcp}, we can strengthen the hardness result and show that
even solving Tukey regression with approximation ratio $1 + \eta$ requires
$2^{\widetilde{\Omega}(d)}$ time for some fixed constant $\eta > 0$.

We complement our hardness results by giving an exponential time algorithm for Tukey regression, using the polynomial system verifier \cite{renegar1992computational,basu1996combinatorial}. 
This technique has been used to solve a number of numerical linear algebra problems in previous works \cite{song2017low, razenshteyn2016weighted, clarkson2015input, arora2016computing, moitra2016almost}.
For the loss function defined in \eqref{eq:tukey_loss}, the algorithm runs in $2^{O(n \log n)} \cdot \log(1 / \varepsilon)$ time to find a $(1 + \varepsilon)$-approximate solution of an instance of size $n \times d$. 
By further applying our dimensionality reduction methods, the running time can be reduced to $2^{\poly(d \log n)}$, which is significantly faster when $n \gg d$ and comes close to the $2^{\widetilde{\Omega}(d)}$ running time lower bound. 

\paragraph{Empirical Evaluation.}
We test our dimensionality reduction methods on both synthetic datasets and real datasets.
Our empirical results quite clearly demonstrate the practicality of our methods. 

% !TEX root = arxiv.tex

\section{Technical Overview}
In this section, we give an overview of our technical contributions. 
For convenience, we state our high-level ideas in terms of the loss function 
$$
M(x) = \begin{cases}
x^2 & |x| \le 1 \\
1 & |x| \ge 1
\end{cases}
$$
where $p = 2$ and $\tau = 1$.
We show how to generalize our ideas to other loss functions that satisfy Assumption \ref{as M} later.

\subsection{Structural Theorem and Algorithms for Finding Heavy Coordinates}
Our first main technical contribution is the following structural theorem, which is crucial for each of our 
algorithms.
%oblivious sketch and the row sampling algorithm. 
\begin{theorem}[Informal]\label{thm:struct_informal}
For a given matrix $A \in \mathbb{R}^{n \times d}$ and $\alpha \ge 1$,
there exists a set of indices $I\subseteq [n]$ with size $|I| \le \widetilde{O}(d \alpha)$, 
such that for all $y \in \colspan(A)$, 
if $y$ satisfies $\|y\|_M \le \alpha$
then $\{i \in n \mid |y_i| > 1\}\subseteq I$.
\end{theorem}

Intuitively, Theorem \ref{thm:struct_informal} states that for all vectors $y$
in the column space of $A$ with small $\|y\|_M$, the heavy coordinates of $y$
(coordinates with $|y_i| \ge 1$) must lie in a set $I$ with small cardinality.
To prove Theorem \ref{thm:struct_informal}, in Figure \ref{alg:informal_struct} we give an informal description of our algorithm for finding the set $I$.
\ifarxiv
The formal description of the algorithm can be found in Section~\ref{sec:struct1}. 
\else
The formal description of the algorithm can be found in the supplementary material.
\fi

\begin{figure}[H]
\begin{framed}
\begin{enumerate}
\item Let $I = \emptyset$.
\item Repeat the following for $\alpha$ times: 
\begin{enumerate}
\item Calculate the leverage scores $\{u_i\}_{i \in [n] \setminus I, *}$ of the matrix $A_{ [n] \setminus I, *}$.
\item For each $i \in  [n] \setminus I$, if $u_i \ge \Omega(1 / \alpha)$, then add $i$ into $I$.
\end{enumerate}
\item Return $I$.
\end{enumerate}
\end{framed}
\caption{Polynomial time algorithm for finding heavy coordinates.}
\label{alg:informal_struct}
\end{figure}

For correctness, we first notice that for
a vector $y$ with $\|y\|_M \le \alpha$, the number of heavy coordinates is at
most $\alpha$, since $M(y_i) = 1$ for all $|y_i| > 1$. Now consider the
coordinate $i$ with largest $|y_i|$ and $|y_i| > 1$. We must have
$\|y\|_2^2 \le \alpha + \alpha y_i^2$, since the contribution of coordinates with $|y_i| \le 1$
to $\|y\|_2^2$ is upper bounded by $\|y\|_M \le \alpha$, and there are at most
$\alpha$ coordinates with $|y_i| > 1$, each contributing at most $y_i^2$ to
$\|y\|_2^2$. Now we claim that we must add the coordinate $i$ with largest
$|y_i|$ into the set $I$, which simply follows from
\begin{equation}\label{eq:leverage_large}
\frac{y_i^2}{\|y\|_2^2} \ge \frac{y_i^2}{\alpha + \alpha y_i^2} \ge \Omega(1 / \alpha)
\end{equation}
and thus the leverage score of the row $A_{i, *}$ is at least $\Omega(1 / \alpha)$.
(Here we use that the $i$-th leverage score is at least as large as $y_i^2/\|y\|_2^2$ for all $y\in\colspan(A)$.)
After adding $i$ into $I$, we consider the second largest $|y_i|$ with $|y_i| \ge 1$.
A similar argument shows that we will also add $i$ into $I$ in the second repetition. 
After repeating $\alpha$ times we will add all coordinates $i$ with $|y_i| > 1$ into $I$,
and all coordinates added to $I$ have leverage score $\Omega(1 / \alpha)$.

The above algorithm has two main drawbacks. First of all, it returns a set with
size $|I| \le O(d \alpha^2)$ as opposed to $\widetilde{O}(d\alpha)$. Moreover,
the algorithm runs in $O(\nnz(A) \cdot \alpha)$ time since we need to calculate
the leverage scores of $A_{[n] \setminus I, *}$ a total of $\alpha$ times. When $\alpha
= \poly(d)$, such an algorithm does not run in input-sparsity time. An
input-sparsity time algorithm for finding such a set $I$, on the other hand, is
an important subroutine for our input-sparsity time row sampling algorithm. 
\ifarxiv
In Section \ref{sec:struct2}, 
\else
In the supplementary material,
\fi
we give a randomized algorithm for finding a set $I$
with size $|I| \le \widetilde{O}(d\alpha)$ that runs in input-sparsity time,
and we give an informal description of the algorithm in Figure \ref{alg:informal_struct2}.
Notice that calculating leverage scores of the matrices $A_{\Gamma_j, *}$ can be done in $\widetilde{O}(\nnz(A) + \poly(d))$ time using existing approaches \cite{CW13, NN13}.

\begin{figure}[H]
\begin{framed}
\begin{enumerate}
\item Let $I = \emptyset$.
\item Repeat the following for $O(\log(d \alpha))$ times: 
\begin{enumerate}
\item Randomly partition $[n]$ into $\Gamma_1, \Gamma_2, \ldots, \Gamma_{\alpha}$.
\item For each $j \in [\alpha]$, calculate the leverage scores $\{u_i\}_{i \in \Gamma_j}$ of the matrix $A_{\Gamma_j, *}$.
\item For each $j \in [\alpha]$, for each $i \in \Gamma_j$, if
$
\hat{u}_i \ge \Omega(1)
$,
then add $i$ to $I$.
\end{enumerate}
\item Return $I$.
\end{enumerate}
\end{framed}
\caption{Input-sparsity time algorithm for finding heavy coordinates.}
\label{alg:informal_struct2}
\end{figure}

For correctness, recall that we only need to find those coordinates $i$
for which there exists a vector $y \in \colspan(A)$ with $\|y\|_M \le \alpha$ and $|y_i| \ge 1$.
Since $\|y\|_M \le \alpha$, there are most $\alpha$ coordinates in $y$ with absolute value at least $1$.
Thus, with constant probability, the coordinate $i$ is in a set $\Gamma_j$ such that it is the only coordinate with $|y_i| \ge 1$ in $\Gamma_j$.
Moreover, by Markov's inequality, with constant probability the squared $\ell_2$ norm of coordinates in $\Gamma_j \setminus \{i\}$ is at most a constant. 
Conditioned on these events, using an argument similar to \eqref{eq:leverage_large}, the leverage score of the row $A_{i, *}$ in $A_{\Gamma_j, *}$ is at least a constant, in which case we will add $i$ into $I$.
In order to show that we will add all such $i$ into $I$ with good probability, we repeat the whole procedure for $O(\log(d \alpha))$ times and apply a union bound over all such $i$.
$O(\log(d \alpha))$ repetitions suffice since there are at most $O(\poly(d \alpha))$ different such $i$, as implied by the existential result mentioned above. 

The above algorithm also implies the existence of a set $I$ with better upper bounds on $|I|$, by the probabilisitic method. These algorithms can be readily
generalized to general $\tau > 0$, and any $p \ge 1$ using {\em $\ell_p$ Lewis
weights} in place of leverage scores. 
We also give a brief overview of Lewis weights and related properties 
\ifarxiv
in Section \ref{sec:lewis} 
\else
in the supplementary material
\fi
for readers unfamiliar with these topics.

\subsection{The Net Argument}
Our second technical contribution is a net argument for Tukey loss functions.
Due to the lack of scale-invariance, the net size for the Tukey loss functions
need not be $n^{\poly(d)}$. While the $M$-functions in \cite{cw15} also do not
satisfy scale-invariance, the $M$-functions in \cite{cw15} have at least linear
growth and so for any value $c$, and for an orthonormal basis $U$ of $A$, the
set of $x$ for which $\|Ux\|_M = c$ satisfy $c/\poly(n) \leq \|x\|_2 \leq c
\cdot \poly(n)$, and so one could use $O(\log n)$ nested nets for the $\ell_2$
norm to obtain a net for the $M$-functions. This does not hold for the Tukey
loss function $M$, e.g., if $c = \tau$, and if the first column of $U$ is $(1,
0, \ldots, 0)^T$, then if $x_1 = \infty$ and $x_2 = x_3 = \cdots = x_d = 0$, one
has $\|Ux\|_M = c$. This motivates Assumption~\ref{as main} above.

Using Assumption \ref{as main}, we construct a net
$\mathcal{N}_\varepsilon$ with size
$|\mathcal{N}_\varepsilon| \le (n / \varepsilon)^{\poly(d)}$, such that for
any $y = Ax-b$ with $\|x\|_2 \le \xupper = n^{\poly(d)}$, there exists
$y' \in \mathcal{N}_\varepsilon$ with $\|y' - y\|_M \le \varepsilon$.
The construction is based on a standard volume argument. 
Notice that such a net only gives an {\em additive error}
guarantee. To give a {\em relative error} guarantee, we notice that for a vector
$y = Ax - b$ with sufficiently small $\|y\|_M$, we must have $\|y\|_{\infty} \le
1$, in which case the Tukey loss function $\|\cdot\|_M$ behaves similarly to the
squared $\ell_2$ norm $\|\cdot\|_2^2$, and thus we can instead use the net
construction for the $\ell_2$ norm. This idea can be easily generalized to
general $p \ge 1$ and $\tau > 0$ if the loss function satisfies $M(x) = |x|^p$
when $|x| \le \tau$. 

To cope with other loss functions that satisfy Assumption
\ref{as M} for which $M(x)$ can only be {\em approximated} by $|x|^p$ when
$|x| \le \tau$, we use the nested net construction in \cite{cw15} when $\|y\|_M$
is sufficiently small. 
%The formal construction and analysis is given in Lemma \ref{lem:net2}.
Our final net argument for Tukey loss functions is a careful combination of the two net constructions mentioned above. 
\ifarxiv
The full details are given in Section \ref{sec:net_arg}.
\else
The full details are given in the supplementary material.
\fi

%Our bound on $\norm{x}_2$, and net analysis, in \S\ref{sec:net} is thus entirely novel relative to the prior work \cite{cw15}.

\subsection{The Row Sampling Algorithm}
Our row sampling algorithm proceeds in a recursive manner, and employs a combination
of {\em uniform sampling} and {\em leverage score sampling}, together with the
procedure for finding heavy coordinates. 
We give an informal description in Figure \ref{alg:informal_sample}.
\ifarxiv
See Section \ref{sec:row_sample} for the formal description and analysis.
\else
See the supplementary material for the formal description and analysis.
\fi

\begin{figure}
\begin{framed}
\begin{enumerate}
\item Use the algorithm in Figure \ref{alg:informal_struct2} to find a set $I$ with $\alpha = \poly(d \log n / \varepsilon)$.
\item Calculate the leverage scores $\{u_i\}$ of the matrix $A_{[n] \setminus I, *}$.
\item For each row $A_{i, *}$, we define its sampling probability $p_i$ to be
$$
p_i =  \begin{cases}
1 & i \in I\\
\min\{1, 1 / 2 +  u_i  \poly(d / \varepsilon)\} & i \notin I
\end{cases}.
$$
\item Sample each row with probability $p_i$.
\item Recursively call the algorithm on the resulting matrix until the number of remaining rows is at most $\poly(d \log n / \varepsilon)$.
\end{enumerate}
\end{framed}
\caption{The row sampling algorithm. }
\label{alg:informal_sample}
\end{figure}

For a vector $y = Ax - b$, we conceptually split
coordinates of $y$ into two parts: heavy coordinates (those with $|y_i| > 1$)
and light coordinates (those with $|y_i| \le 1$). Intuitively, we need to apply
uniform sampling to heavy coordinates,
 since all heavy coordinates contribute the same to $\|y\|_M$,
and leverage score sampling to light coordinates, since the Tukey loss function
behaves similarly to the squared $\ell_2$ norm for light coordinates. 

\ifarxiv
In the formal analysis given in Section \ref{sec:recursive_one_step}, we show that if
\else
In the formal analysis given in supplementary material, we show that if
\fi
either the contribution from heavy coordinates to $\|y\|_M$ or the contribution
from light coordinates to $\|y\|_M$ is at least $\Omega(\poly(d \log n / \varepsilon))$,
then with high probability, uniform sampling with sampling probability $1 / 2$
will preserve $\|y\|_M$ up to $\varepsilon$ relative error, for all vectors $y$ in the net.
The proof is based on standard concentration inequalities. 

If both the contribution from heavy coordinates and the contribution from light coordinates
is $O(\poly(d \log n / \varepsilon))$, uniform sampling will no longer be
sufficient, and we resort to the structural theorem in such cases. 
By setting
$\alpha = \poly(d \log n / \varepsilon)$ in the algorithm for finding heavy coordinates,
we can identify a set $I$ with size
$|I| = \poly(d \log n / \varepsilon)$, which includes the indices of all heavy
coordinates.
We simply keep all coordinates in $I$ by setting $p_i = 1$.
The remaining coordinates must be light, and hence behave very similarly to the
squared $\ell_2$ norm. Thus, we can use leverage score sampling to deal with the
remaining light coordinates. 
This also explains why we need to use a combination
of uniform sampling and leverage score sampling.

Our algorithm will eliminate roughly half of the
coordinates in each round, and after $O(\log n)$ rounds there are at most
$O(\poly(d \log n / \varepsilon))$ remaining coordinates, in which case we stop the sampling process and return our reduced version of the problem.
In each round we calculate the leverage scores and call the algorithm in Figure \ref{alg:informal_struct2} to find heavy coordinates.
Since both subroutines can be implemented to run in $\widetilde{O}(\nnz(A) + \poly(d \log n / \varepsilon))$ time, the overall running time of our row sampling algorithm is also $\widetilde{O}(\nnz(A) + \poly(d \log n / \varepsilon))$.

The above algorithm can be readily
generalized to any loss function $M$ that satisfies Assumption \ref{as M}.
\ifarxiv
Our formal analysis in Section \ref{sec:row_sample} is a careful
\else
Our formal analysis in supplementary material is a careful
\fi
combination of all ideas mentioned above.

\subsection{The Oblivious Sketch} 
%At the crux of the argument in \cite{cw15},
%and earlier arguments for least squares
%regression, with probability $1-n^{-\poly(d)}$, for any fixed vector $x$,
%$\|S(Ax-b)\|_M \geq \frac{\|Ax-b\|_M}{2}$. On the other hand, with constant probability, for a fixed $x$,
%$\|S(Ax-b)\|_M \leq 2 \|Ax-b\|_M$. One then places a net on all vectors $x$, which under the assumptions on the
%$M$-functions in \cite{cw15}, has size $n^{\poly(d)}$, and by a union bound, for every net vector $x$,
%$\|S(Ax-b)\|_M \geq \frac{\|Ax-b\|_M}{2}$, which since the net is fine enough, implies 
%for all vectors $x \in \mathbb{R}^{d}$, $\|S(Ax-b)\|_M \geq C\frac{\|Ax-b\|_M}{2}$.
%On the other hand, by a Markov bound, with constant
%probability the minimizer $x^*$ to $\min_x \|Ax-b\|_M$ satisfies $\|S(Ax^*-b)\|_M \leq 2 \|Ax^*-b\|_M$. These
%conditions suffice for regression, since it is a minimization problem.
From an algorithmic standpoint, our oblivious sketch $S$ is similar to that in \cite{cw15}. 
The distribution on matrices $S$ can be viewed roughly as a stack of $\hm=O(\log n)$
matrices, where the $i$-th such matrix is the product of a {\sf CountSketch} matrix with $\poly(d \log n)$ rows
with a diagonal matrix $D$ which samples roughly
$1/(d \log n)^i$ uniformly random coordinates of an $n$-dimensional vector.
%({\sf CountSketch} matrices were called \emph{sparse embedding matrices} in \cite{cw09}.)
 Thus, $S$ can be
viewed as applying {\sf CountSketch} to a subsampled set of coordinates of a vector, where
the subsampling is more aggressive as $i$ increases. The weight vector $w$ is such that
$w_j = (d \log n)^i$ for all coordinates $j$ corresponding to the $i$-th matrix in the stack. 
Our main technical contribution here is showing that this simple sketch actually works for Tukey loss functions.
%which requires an overhaul to the analysis in \cite{cw15}. 
%The major difference with \cite{cw15} is how different groups of similar coordinates in a vector $y = Ax-b$ are estimated. 
One of the main ideas in \cite{cw15} is that if there is a subset of at least $\poly(d) \log n$ coordinates of a vector $y$
of similar absolute value, then in one of the levels of subsampling of $S$, with probability $1-1/ n^{\poly(d)}$
there will be $\Theta(\poly(d) \log n)$ coordinates
in this group which survive the subsampling and are \emph{isolated}, that is,
they hash to separate {\sf CountSketch} buckets. Using that the $M$-function
does not grow too quickly, which holds for Tukey loss functions as well if $p \le 2$, this suffices for estimating the contribution to $\|y\|_M$
from all large subsets of similar coordinates.

The main difference in this work is how estimates for {\it small subsets of
coordinates} of $y$ are made. In \cite{cw15} an argument based on leverage
scores sufficed, since, as one ranges over all unit vectors of the form $y =
Ax-b$, there is only a small subset of coordinates which could ever be large,
which follows from the condition that the column span of $A$ is a
low-dimensional subspace. At first glance, for Tukey loss functions this might
not be true. One may think that for any $t \leq \poly(d) \log n $, it could be
that for a vector $y = Ax-b$, any subset $T$ of $t$ of its $n$ coordinates could
have the property that $M(y_i) = 1$ for $i \in T$, and $M(y_i) < 1$ otherwise.
However, our structural theorem in fact precludes such a possibility. The
structural theorem implies that there are only $\poly(d \log n)$ coordinates for
which $M(y_i)$ could be $1$. For those coordinates with $M(y_i) < 1$, the Tukey
loss function behaves very similarly to the squared $\ell_2$ norm, and thus we
can again use the argument based on leverage scores. After considering these two
different types of coordinates, we can now apply the perfect hashing argument as
in \cite{cw15}. 

These ideas can be readily generalized to general $\tau > 0$,
and any $1 \le p \le 2$, again using $\ell_p$ Lewis weights in place of leverage
scores. 
\ifarxiv
We formalize these ideas in Lemma \ref{lem Q< good}.
\else
We formalize these ideas in the supplementary material.
\fi

\section{Preliminaries}\label{sec:pre}
For two real numbers $a$ and $b$, we use the notation $a = (1 \pm \varepsilon) b$ if $a \in [(1 - \varepsilon)b, (1 + \varepsilon)b]$.

We use $\|\cdot \|_p$ to denote the $\ell_p$ norm of a vector, and $\|\cdot\|_{p, w}$ to denote the weighted $\ell_p$ norm, i.e., 
$$
\|y\|_{p, w} = \left(\sum_{i=1}^n w_i |y_i|^p \right)^{1 / p}.
$$

For a vector $y \in \mathbb{R}^n$, a weight vector $w \in \mathbb{R}^n$ whose
entries are all non-negative and a loss function $M : \R \to \R^+$ that
satisfies Assumption \ref{as M}, $\|y\|_{M, w}$ is defined to be
$$
\|y\|_{M, w} = \sum_{i=1}^n w_i \cdot M(y_i).
$$
We also define $\|y\|_M$ to be
$$
\|y\|_{M} = \sum_{i=1}^n  M(y_i).
$$
For a vector $y \in \R^n$ and a real number $\tau \ge 0$, we define $H_y$ to be
the set $H_y = \{i \in [n] \mid |y_i| > \tau\}$, and $L_y$ to be the set
$L_y = \{i \in [n] \mid |y_i| \le \tau\}$.

\subsection{Tail Inequalities} 

\begin{lemma}[Bernstein's inequality] \label{lem:bernstein}
Suppose $X_1, X_2, \ldots, X_n$ are independent random variables taking values in $[-b, b]$. 
Let $X = \sum_{i=1}^n X_i$
and 
$\Var[X]= \sum_{i=1}^n \Var[X_i]$ be the variance of $X$. For any $t > 0$ we have
$$
\Pr[|X - \E[X]|> t] \le 2\exp\left(-\frac{t^2}{2 \Var[X] + 2bt/3}\right).
$$
\end{lemma}

\subsection{Facts Regarding the Loss Function}
\begin{lemma}\label{lem:entry_perturb}
Under Assumption \ref{as M}, there is a constant $C > 0$ that depends only on
$p$, for which for any $a, b$ with $|b| \le \varepsilon |a|$, we have
$M(a+b) = (1 \pm C \varepsilon) M(a)$.
\end{lemma}
\begin{proof}
Without loss of generality we assume $a > 0$.
When $b \ge 0$, by Assumption \ref{as M}.\ref{as M growth}, we have
$$
M(a) \le M(a + b) \le (1 + \varepsilon)^p \cdot M(a)  \le (1 + C\varepsilon) M(a).
$$
When $b < 0$, we have
$$
M(a) \ge M(a + b) \ge \left(\frac{a}{a + b}\right)^p M(a) \ge (1 - C \varepsilon) M(a).
$$
\end{proof}
\begin{lemma}\label{lem:perturb}
Under Assumption \ref{as M}, there is a constant $C' > 0$ that depends only on
$p$, for which for any $e, y \in \R^n$ and any weight vector $w$ with
$\|e\|_{M, w} \le \varepsilon^{2p + 1} \|y\|_{M, w}$,
$$\|y + e\|_{M, w} = (1 \pm C' \varepsilon)\|y\|_{M, w}.$$
\end{lemma}
\begin{proof}
Clearly, by Assumption \ref{as M}.\ref{as M growth}, 
$$
\|e / \varepsilon^2\|_{M, w} \le \varepsilon^{-2p} \|e\|_{M, w} \le \varepsilon \|y\|_{M, w}.
$$
Let $S = \{i \in n \mid |e_i| \le \varepsilon |y_i|\}$.
By Lemma \ref{lem:entry_perturb}, for all $i \in S$ we have $M(y_i + e_i) = (1 \pm C \varepsilon) M(y_i)$.
For all $i \in [n] \setminus S$, we have $|e_i| > \varepsilon |y_i|$.
For sufficiently small $\varepsilon$, by Assumption \ref{as M}.\ref{as M inc} and Lemma \ref{lem:entry_perturb}, 
$$
M(e_i + y_i) \le M(e_i / \varepsilon^2 + y_i) \le (1 + C \varepsilon)M(e_i / \varepsilon^2),
$$
which implies
$$
\sum_{i \in [n] \setminus S} w_i M(y_i + e_i) \le (1 + C \varepsilon)
    \|e / \varepsilon^2\|_{M, w} \le (1 + C \varepsilon) \varepsilon \|y\|_{M, w}.
$$
Furthermore, 
$$
\sum_{i \in [n] \setminus S} w_i M(y_i)
    \le \sum_{i \in [n] \setminus S} w_i M(e_i / \varepsilon)
    \le \|e / \varepsilon^2\|_{M, w} \le \varepsilon \|y\|_{M, w}.
$$
Thus, 
\begin{align*}
&\|y + e\|_{M, w}\\
= & \sum_{i \in S} w_i M(y_i + e_i) + \sum_{i \in [n] \setminus S} w_iM(y_i + e_i)\\
= & (1 \pm C\varepsilon) \sum_{i \in S} w_i M(y_i) \pm (1 + C\varepsilon)\varepsilon \|y\|_{M, w}\\
= & (1 \pm C' \varepsilon) \|y\|_{M, w}.
\end{align*}
\end{proof}

% !TEX root = main.tex
\subsection{Facts Regarding Lewis Weights}\label{sec:lewis}
In this section we recall some facts regarding leverage scores and Lewis weights.
\begin{definition}\label{def:leverage_score}
Given a matrix $A \in \mathbb{R}^{n \times d}$. The {\em leverage score} of a row $A_{i, *}$ is defined to be
\[
\tau_i(A) =  A_{i, *} (A^TA)^{\dagger}  (A_{i, *})^T.
\]
\end{definition}
%\Ken{Changed above because $A_{i, *}$ is a row vector (isn't it?), and above used it as column vector before}

\begin{definition}[\cite{cp15}]\label{def:lewis}
For a matrix $A \in \mathbb{R}^{n \times d}$, its {\em $\ell_p$ Lewis weights}
$\{u_i\}_{i=1}^n$ are the {\em unique weights} such that for each $i \in [n]$ we
have
$$
u_i = \tau_i(U^{1/2 - 1/p} A).
$$
Here $\tau_i$ is the leverage score of the $i$-th row of a matrix
and
$U$ is the diagonal matrix formed by putting the elements of $u$ on the diagonal.
\end{definition}
\begin{theorem}[\cite{cp15}]\label{thm:alg_lewis}
There is an algorithm that receives a matrix $A \in \mathbb{R}^{n \times d}$ and outputs $\{\hat{u}\}_{i = 1}^n$ such that 
$$
u_i \le \hat{u}_i \le 2u_i,
$$
where $\{u_i\}_{i=1}^n$ are the $\ell_p$ Lewis weights of $A$.
Furthermore, the algorithm runs in $\widetilde{O}(\nnz(A) + d^{p / 2 + O(1)})$ time. 
\end{theorem}

\begin{theorem}[Lewis's change of density \cite{Lewis1978}, see also {\cite[p.~113]{wojtaszczyk1996banach}}]\label{thm:fa_lewis}
Given a matrix $A \in \mathbb{R}^{n \times d}$ and $p \ge 1$, there exists a
basis matrix $H \in \mathbb{R}^{n \times d}$ of the column space of $A$, such
that if we define a weight vector $\overline{u} \in \mathbb{R}^n$ where
$\overline{u}_i = \|H_{i,*}\|_2$, then the following hold:
\begin{enumerate}
\item $\|\overline{u}\|_p^p \le d$;
\item $\overline{U}^{p / 2 - 1}H$ is an orthonormal matrix.
\end{enumerate}
Here $\overline{U}$ is the diagonal matrix formed by putting the elements of $\overline{u}$ on the diagonal.
\end{theorem}

\begin{lemma}[See, e.g., {\cite[p.~115]{wojtaszczyk1996banach}}]\label{lem:l2_bound}
Given a matrix $A \in \mathbb{R}^{n \times d}$, for the basis matrix $H$ and the
weight vector $\overline{u}$ defined in Theorem \ref{thm:fa_lewis}, for all $x
\in \mathbb{R}^d$ we have
$$
\|\overline{U}^{p / 2 - 1}Hx\|_2 \le \|Hx\|_p \le d^{1/p - 1/2} \|\overline{U}^{p / 2 - 1}Hx\|_2
$$
when $1 \le p \le 2$, and
$$
\|Hx\|_p \le \|\overline{U}^{p / 2 - 1}Hx\|_2 \le  d^{1/2 - 1/p}\|Hx\|_p
$$
when $p \ge 2$.

Since $\overline{U}^{p / 2 - 1}H$ is an orthonormal matrix, for all $x \in \mathbb{R}^d$ we have
$$
\|x\|_2 \le \|Hx\|_p \le d^{1/p - 1/2} \|x\|_2
$$
when $1 \le p \le 2$, and
$$
\|Hx\|_p \le \|x\|_2 \le  d^{1/2 - 1/p}\|Hx\|_p
$$
when $p \ge 2$.
\end{lemma}

\begin{lemma}\label{lem:equivalence}
Given a matrix $A \in \mathbb{R}^{n \times d}$ and $p \ge 1$, the weight vector
$u$ defined in Definition \ref{def:lewis} and the weight vector $\overline{u}$
defined in Theorem \ref{thm:fa_lewis} satisfies
$$
u_i = \overline{u}_i^p.
$$
\end{lemma}
\begin{proof}
We show that substituting $u_i = \overline{u}_i^p$ will satisfy 
$$
u_i = \tau_i(U^{1/2 - 1/p} A),
$$
and thus the theorem follows by the uniqueness of Lewis weights.

Since leverage scores are invariant under change of basis (see, e.g., {\cite[p.~30]{w14}}), we have 
$$
\tau_i(U^{1/2 - 1/p}A) = \tau_i(U^{1/2 - 1/p}H),
$$
where $H$ is the basis matrix defined in Theorem \ref{thm:fa_lewis}.
Substituting $u_i = \overline{u}_i^p$ we have
$$
\tau_i(U^{1/2 - 1/p}A) = \tau_i(\overline{U}^{p / 2 - 1} H).
$$
However, since $\overline{U}^{p / 2 - 1} H$ is an orthonormal matrix, and the 
leverage scores of an orthonormal matrix are just squared $\ell_2$ norm of rows
(see, e.g., {\cite[p.~29]{w14}}), we have
$$
\tau_i(U^{1/2 - 1/p}A) = \left( \overline{u}_i^{p / 2 - 1} \|H_{i,*}\|_2 \right)^2= \overline{u}_i^p.
$$
\end{proof}
\begin{lemma}\label{lem:bound_entry}
Given a matrix $A \in \mathbb{R}^{n \times d}$ and $p \ge 1$, for all $y \in \colspan(A)$ and $i \in [n]$, we have
$$
|y_i|^p  \le d^{\max\{0, p / 2 - 1\}}  u_i \cdot \|y\|_p^p.
$$
Here $\{u_i\}_{i = 1}^n$ are the $\ell_p$ Lewis weights defined in Definition \ref{def:lewis}.
\end{lemma}
\begin{proof}
For all $y \in \colspan(A)$, we can write $y = Hx$ for some vector $x \in
\mathbb{R^d}$ and the basis matrix $H$ in Theorem \ref{thm:fa_lewis}. By the 
Cauchy-Schwarz inequality, 
$$
|y_i|^p = |\langle x, H_{i,*}\rangle|^p \le \|x\|_2^p \cdot \|H_{i,*}\|_2^p,
$$
which implies
$$
|y_i|^p  \le d^{\max\{0, p / 2 - 1\}} \cdot \|y\|_p^p \cdot \|H_{i,*}\|_2^p
$$
by Lemma \ref{lem:l2_bound}, which again implies 
$$
|y_i|^p  \le d^{\max\{0, p / 2 - 1\}} u_i \cdot \|y\|_p^p
$$
since $\overline{u}_i = \|H_{i,*}\|_2$ and $u_i = \overline{u}_i^p$ by Lemma \ref{lem:equivalence}.
\end{proof}
\begin{lemma}\label{lem:single_lp}
Under Assumption \ref{as M}, given a matrix $A \in \mathbb{R}^{n \times d}$, $\deltalewis \in (0, 1)$,
and a weight vector $w \in \mathbb{R}^n$ such that (i) $w_i \ge 1$ for all
$i \in [n]$ and (ii) $\max_{i \in [n]} w_i \le 2 \min_{i \in [n]} w_i$. Let
$w' \in \mathbb{R}^n$ be another weight vector which is defined to be
$$
w'_i = \begin{cases}
w_i / p_i & \text{with probability }p_i \\
0 & \text{with probability }1 - p_i
\end{cases}
$$
and $p_i$ satisfies
$$
p_i \ge \min\{1, \Theta(U_M / L_M \cdot d^{\max\{0, p / 2 - 1\}}  u_i \cdot \log(1 / \deltalewis) / \varepsilon^2)\},
$$
then for any fixed vectors $x \in \mathbb{R}^d$ such that $\|Ax\|_{\infty} \le \tau$, with probability at least $1 - \deltalewis$ we have
$$
\|Ax\|_{M, w} = (1 \pm \varepsilon) \|Ax\|_{M, w'}.
$$
\end{lemma}
\begin{proof}
Without loss of generality we assume $1 \le w_i \le 2$ for all $i \in [n]$.
Let $y = Ax$.
We use the random variable $Z_i$ to denote 
$$
Z_i = w_i' M(y_i).
$$
Clearly $\E[Z_i] = w_i M(y_i)$, which implies
$$
\E[\|y\|_{M, w'}] = \|y\|_{M, w}.
$$
Furthermore, $Z_i \le 2M(y_i) / p_i$. Since $\|y\|_{\infty} \le \tau$ and $L_M
|y_i|^p \le M(y_i) \le U_M |y_i|^p$ when $|y_i| \le \tau$, by
Lemma~\ref{lem:bound_entry} we have
$$
Z_i \le 2 U_M |y_i|^p / p_i \le \Theta(L_M \cdot \|y\|_p^p \cdot \varepsilon^2 / \log(1 / \deltalewis))
\le \Theta(\|y\|_{M, w} \cdot \varepsilon^2 / \log(1 / \deltalewis)).
$$
%\Ken{Changed to $L_M$ from $1/L_M$}
Moreover, $\E[Z_i^2] \le O((M(y_i))^2 / p_i)$, which implies
$$
\sum_{i=1}^n \E[Z_i^2] \le O \left( \sum_{i=1}^n (M(y_i))^2 / p_i \right).
$$
By H\"older's inequality, 
$$
\sum_{i=1}^n \E[Z_i^2] \le O(\|y\|_{M}) \cdot \max_{i \in [n]} M(y_i) / p_i \le O(\|y\|_{M, w}^2 \cdot \varepsilon^2 / \log(1 / \deltalewis)).
$$

Furthermore, since $$\Var\left[\sum_{i=1}^n Z_i\right] = \sum_{i=1}^n \Var[Z_i]
\le \sum_{i=1}^n \E[Z_i^2],$$ Bernstein's inequality in Lemma
\ref{lem:bernstein} implies
$$
\Pr \left[ \left|\|y\|_{M, w'} -\|y\|_{M, w} \right|  > t\right]
    \le \exp  \left( - \Theta\left(\frac{t^2}{\|y\|_{M, w} \cdot \varepsilon^2 / \log(1 / \deltalewis)\cdot t + \|y\|_{M, w}^2 \cdot \varepsilon^2 / \log(1 / \deltalewis)}\right)\right).
$$
Taking $t = \varepsilon \cdot \|y\|_{M, w}$ implies the desired result. 
\end{proof}

\begin{theorem}\label{thm:se}
Given a matrix $A \in \mathbb{R}^{n \times d}$,  $\deltase\in (0,1)$,
and a weight vector
$w \in \mathbb{R}^n$ such that (i) $w_i \ge 1$ for all $i \in [n]$ and (ii)
$\max_{i\in [n]} w_i \le 2 \min_{i \in [n]} w_i$. 
Let $w' \in \mathbb{R}^n$ be another weight vector which is defined to be
$$
w'_i = \begin{cases}
w_i / p_i & \text{with probability }p_i \\
0 & \text{with probability }1 - p_i
\end{cases}
$$
and $p_i$ satisfies
$$
p_i \ge \min\{1, \Theta( d^{\max\{0, p / 2 - 1\}}  u_i \cdot (d \log (1 / \varepsilon) + \log(1 / \deltase)) / \varepsilon^2)\},
$$
then with probability at least $1 - \deltase$, for all vectors $x \in \mathbb{R}^d$, we have
$$
\|Ax\|_{p, w}^p = (1 \pm \varepsilon) \|Ax\|_{p, w'}^p.
$$
\end{theorem}
\begin{proof}
Let $\mathcal{N}$ be an $\varepsilon$-net for $\{Ax \mid \|Ax\|_{p, w} = 1\}$. 
Standard facts (see, e.g., {\cite[p.~48]{w14}}) imply that
$\log |\mathcal{N}| \le O(d \log(1 / \varepsilon))$.
Now we invoke Lemma \ref{lem:single_lp} with $\deltalewis = \deltase / |\mathcal{N}|$.
Notice that $f(x) = |x|^p$ is also a loss function that satisfies
Assumption~\ref{as M}, with $L_M = U_M = 1$ and $\tau = \infty$.
Thus, if $p_i$ satisfies
$$
p_i \ge \Theta( d^{\max\{0, p / 2 - 1\}}  u_i \cdot (d \log (1 / \varepsilon) + \log(1 / \deltase)) / \varepsilon^2),
$$
then with probability $1 - \deltase$, simultaneously for all $x \in \mathcal{N}$ we have
$$
\|Ax\|_{p, w}^p = (1 \pm \varepsilon) \|Ax\|_{p, w'}^p.
$$
Now we can invoke the standard successive approximation argument (see, e.g.,
{\cite[p.~47]{w14}}) to show that with probability $1 - \deltase$, simultaneously
for all $x \in \mathbb{R}^d$ we have
$$
\|Ax\|_{p, w}^p = (1 \pm O(\varepsilon))\|Ax\|_{p, w'}^p.
$$
Adjusting constants implies the desired result. 
\end{proof}

% !TEX root = arxiv.tex
\section{Finding Heavy Coordinates}\label{sec:heavy}
\subsection{A Polynomial Time Algorithm}\label{sec:struct1}

\begin{figure}[H]
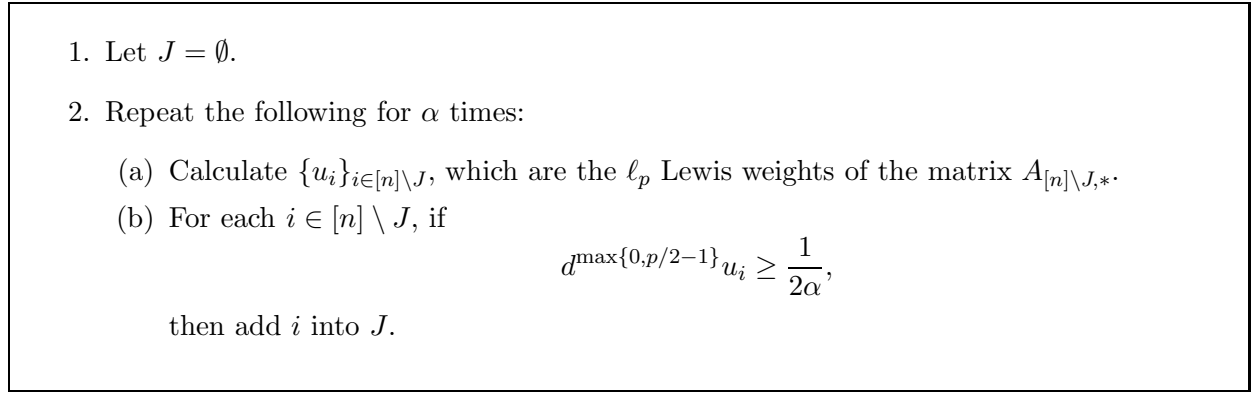

\begin{framed}
\begin{enumerate}
\item Let $J = \emptyset$.
\item \label{alg:main_loop}Repeat the following for $\alpha$ times: 
\begin{enumerate}
\item Calculate $\{u_i\}_{i \in [n] \setminus J}$, which are the $\ell_p$ Lewis weights of the matrix $A_{ [n] \setminus J, *}$.
\item For each $i \in  [n] \setminus J$, if
$$
d^{\max\{0, p / 2 - 1\}} u_i \ge \frac{1}{2\alpha},
$$
then add $i$ into $J$.
\end{enumerate}
\end{enumerate}
\end{framed}
\caption{Algorithm for finding the set $J$.}
\label{alg:struct}
\end{figure}

\begin{theorem}\label{thm:struct}
For a given matrix $A \in \mathbb{R}^{n \times d}$, $\tau \ge 0$ and $p \ge 1$, 
the algorithm in Figure \ref{alg:struct} returns a set of indices $J \subseteq [n]$ with size $|J| \le O(d^{\max\{p / 2, 1\}} \cdot \alpha^2)$, 
such that for all $y \in \colspan(A)$, 
if $y$ satisfies 
	(i) $\|y_{L_y}\|_p^p \le \alpha \cdot \tau^p$ and 
	(ii) $|H_y| \le \alpha$, 
then $H_y \subseteq J$.
\end{theorem}

\begin{proof}
Consider a fixed vector $y \in \colspan(A)$ that satisfies
	(i) $\|y_{L_y}\|_p^p \le \alpha \cdot \tau^p$ and 
	(ii) $|H_y| \le \alpha$.
For ease of notation, we assume $|y_1| \ge |y_2| \ge \cdots \ge |y_n|$.
Of course, this order is unknown and is not used by our algorithm.
Under this assumption, $H_y = \{1, 2, \ldots, |H_y|\}$.

We prove $H_y \subseteq J$ by induction. 
For any $i < |H_y|$, suppose $[i] \subseteq J$ and $i + 1 \notin J$ after the
$i$-th repetition of Step \ref{alg:main_loop}, we show that we will add $i + 1$
into $J$ in the $(i + 1)$-th repetition of Step \ref{alg:main_loop}. 
Since, $[i] \subseteq J$ and $|y_1| \ge |y_2| \ge \cdots \ge |y_n|$, 
$$
\|y_{[n] \setminus J}\|_p^p \le \|y_{L_y}\|_p^p + \alpha |y_{i + 1}|^p \le \alpha \tau^p + \alpha |y_{i + 1}|^p.
$$
Since $i + 1 \in H_y$, we must have $|y_{i + 1}| \ge \tau$, which implies
$$
\frac{|y_{i + 1}|^p}{\|y_{[n] \setminus J}\|_p^p} \ge \frac{1}{2\alpha}.
$$
By Lemma \ref{lem:bound_entry}, this implies
$$
d^{\max\{0, p / 2 - 1\}} u_{i + 1}  \ge \frac{1}{2\alpha},
$$
where $u_{i + 1}$ is the $\ell_p$ Lewis weight of the row $A_{i + 1, *}$ in $A_{[n] \setminus J, *}$, in which case we will add $i + 1$ into~$J$.
Thus, $H_y \subseteq J$ since $|H_y| \le \alpha$.

Now we analyze the size of $J$. 
For the algorithm in Figure \ref{alg:struct}, we repeat the whole procedure $\alpha$ times. 
Each time, an index $i$ will be added into $I$ if and only if 
$$
d^{\max\{0, p / 2 - 1\}} u_i \ge  \frac{1}{2\alpha}.
$$
However, since 
$$ 
\sum_{i \in [n] \setminus J}  u_i =  \sum_{i \in [n] \setminus J} \overline{u}_i^p \le d
$$ by Theorem \ref{thm:fa_lewis}, there are at most $O(d^{\max\{p / 2, 1\}}  \cdot \alpha)$ such indices~$i$.
Thus, the total size of $J$ is upper bounded by $O(d^{\max\{p / 2, 1\}} \cdot \alpha^2)$.
\end{proof}

The above algorithm also implies the following existential result. 
\begin{corollary} \label{cor:existence_structural}
For a given matrix $A \in \mathbb{R}^{n \times d}$, $\tau \ge 0$ and $p \ge 1$, 
there exists a set of indices $J \subseteq [n]$ with size $|J| \le O(d^{\max\{p / 2, 1\}} \cdot \alpha^2)$, 
such that for all $y \in \colspan(A)$, 
if $y$ satisfies 
	(i) $\|y_{L_y}\|_p^p \le \alpha \cdot \tau^p$ and 
	(ii) $|H_y| \le \alpha$, 
then $H_y \subseteq J$.
\end{corollary}
\subsection{An Input-sparsity Time Algorithm}\label{sec:struct2}
To find a set of heavy coordinates, the algorithm in Theorem \ref{thm:struct} runs in polynomial time. 
In this section we present an algorithm for finding heavy coordinates that runs in input-sparsity time. 
The algorithm is described in Figure \ref{alg:struct2}.
\begin{figure}
\begin{framed}
\begin{enumerate}
\item Let $|J| =  O(d^{\max\{p / 2, 1\}} \cdot \alpha^2)$ as in Corollary \ref{cor:existence_structural}.
\item Repeat the following for $O(\log( |J| / \deltastruct))$ times: 
\begin{enumerate}
\item Randomly partition $[n]$ into $\Gamma_1, \Gamma_2, \ldots, \Gamma_{\alpha}$.
\item For each $j \in [\alpha]$, use the algorithm in Theorem \ref{def:lewis} to obtain
weights $\{\hat{u}_i\}_{i \in \Gamma_j}$ such that $u_i \le \hat{u}_i \le 2u_i$,
where $\{u_i\}_{i \in \Gamma_j}$ are the $\ell_p$ Lewis weights of the matrix $A_{\Gamma_j, *}$.
\item For each $j \in [\alpha]$, for each $i \in \Gamma_j$, if
$$
d^{\max\{0, p / 2 - 1\}} \hat{u}_i \ge \frac{1}{6},
$$
then add $i$ to $I$.
\end{enumerate}
\end{enumerate}
\end{framed}
\caption{Algorithm for finding the set $I$.}
\label{alg:struct2}
\end{figure}

\begin{theorem}\label{thm:struct2}
For a given matrix $A \in \mathbb{R}^{n \times d}$, $\tau \ge 0$,  $\deltastruct\in (0,1)$, and $p \ge 1$, 
the algorithm in Figure \ref{alg:struct2} returns a set of indices
$I\subseteq [n]$ with size $|I| \le \widetilde{O}(d^{\max\{p / 2, 1\}}  \alpha  \cdot \log(1 / \deltastruct))$, 
such that with probability at least $1 - \deltastruct$, 
simultaneously for all $y \in \colspan(A)$, 
if $y$ satisfies 
	(i) $\|y_{L_y}\|_p^p \le \alpha \cdot \tau^p$ and 
	(ii) $|H_y| \le \alpha$, 
then $H_y \subseteq I$.
Furthermore, the algorithm runs in $\widetilde{O}\left(\left(\nnz(A) + d^{p / 2 + O(1)}  \cdot \alpha  \right)\cdot \log(1 / \deltastruct)\right)$ time.
\end{theorem}
\begin{proof}
Let $J$ be the set with size $|J| \le O(d^{\max\{p / 2, 1\}} \cdot \alpha^2)$ whose existence is proved in Corollary \ref{cor:existence_structural}.
For all $y \in \colspan(A)$, 
if $y$ satisfies 
	(i) $\|y_{L_y}\|_p^p \le \alpha \cdot \tau^p$ and 
	(ii) $|H_y| \le \alpha$, 
then $H_y \subseteq J$.
We only consider those $c \in J$ for which there exists $y \in \colspan(A)$
such that  
	(i) $\|y_{L_y}\|_p^p \le \alpha \cdot \tau^p$, 
	(ii) $|H_y| \le \alpha$ and
	(iii) $c \in H_y$,
since we can remove other $c$ from $J$ and the properties of $J$ still hold. For
such $c \in H_y$ and the corresponding $y \in \colspan(A)$, suppose for some
$j \in [\alpha]$ we have $c \in \Gamma_j$. Since $|H_y| \le \alpha$, with
probability $(1 - 1 / \alpha)^{|H_y| - 1} \ge 1 / e$, we have
$\Gamma_j \cap H_y = \{c\}$. Furthermore,
$\E[\|y_{L_y \cap \Gamma_j}\|_p^p] = \|y_{L_y}\|_p^p / \alpha \le \tau^p$.
By Markov's inequality, with probability at least $0.8$, we
have $\|y_{L_y \cap \Gamma_j}\|_p^p \le 5\tau^p$. Thus, by a union bound, with
probability at least $1 / e - 0.2 > 0.1$, we have $\|y_{L_y \cap \Gamma_j}\|_p^p
\le 5\tau^p$ and $\Gamma_j \cap H_y = \{c\}$. By repeating $O(\log( |J| /
\deltastruct))$ times, the success probability is at least
$1 - \deltastruct / |J|$. Applying a union bound over all $c \in J$, with probability $1 -
\deltastruct$, the stated conditions hold for all $c \in J$.
We condition on this event in the rest of the proof. 

Consider any $c \in J$ and $y \in \colspan(A)$ with the properties
stated above. Since $|y_c| \ge \tau$, we have 
$$
\frac{|y_c|^p}{\|y_{\Gamma_j}\|_p^p}  \ge \frac{|y_c|^p}{\|y_{\Gamma_j \cap L_y}\|_p^p + |y_c|^p} \ge \frac{1}{6}.
$$
By Lemma \ref{lem:bound_entry}, we must have
$$
d^{\max\{0, p / 2 - 1\}} u_c \ge  \frac{1}{6},
$$
where $u_c$ is the $\ell_p$ Lewis weight of the row $A_{c, *}$ in the matrix $A_{\Gamma_j, *}$, 
which also implies
$$
d^{\max\{0, p / 2 - 1\}} \hat{u}_c \ge   \frac{1}{6}
$$
since $\hat{u}_c \ge u_c$, in which case we will add $c$ to $I$.

Now we analyze the size of $I$.
For each $j \in [\alpha]$, we have
$$
\sum_{i \in \Gamma_j} \hat{u}_i \le 2 \sum_{i \in \Gamma_j} u_i = 2  \sum_{i \in \Gamma_j} \overline{u}_i^p \le 2d
$$ by Theorem \ref{thm:fa_lewis}.
For each $j \in [\alpha]$, there are at most $O(d^{\max\{p / 2, 1\}})$ indices $i$ which satisfy 
$$
d^{\max\{0, p / 2 - 1\}} \hat{u}_i \ge \frac{1}{6},
$$
which implies we will add at most
$O\left(\alpha \cdot d^{\max\{p / 2, 1\}} \right)$ elements into $I$
during each repetition. The bound on the size of $I$ follows
since there are only $O(\log( |J| / \deltastruct)) = O(\log d + \log \alpha + \log(1 / \deltastruct))$ repetitions. 

For the running time of the algorithm, since we invoke
the algorithm in Theorem \ref{thm:alg_lewis} for $O(\log( |J| / \deltastruct))$ times, 
and each time we estimate the $\ell_p$ Lewis weights
of $A_{\Gamma_1, *}, A_{\Gamma_2, *}, \ldots, A_{\Gamma_{|\alpha|}, *}$,
which implies the running time for each repetition is upper bounded by
$$
\sum_{j = 1}^{|\alpha|} \widetilde{O}\left(\nnz(A_{\Gamma_j, *}) + d^{p / 2 + O(1)}\right) = \widetilde{O}\left(\nnz(A) + d^{p / 2 + O(1)}  \cdot \alpha \right).
$$
The bound on the running time follows since we repeat for $O(\log( |J| / \deltastruct))$ times. 
\end{proof}

The above algorithm and the probabilisitic method also imply the following existential result. 
\begin{corollary} \label{cor:existence}
For a given matrix $A \in \mathbb{R}^{n \times d}$, $\tau \ge 0$ and $p \ge 1$, 
there exists a set of indices $I \subseteq [n]$ with size $|I| \le \widetilde{O}(d^{\max\{p / 2, 1\}} \cdot \alpha)$, 
such that for all $y \in \colspan(A)$, 
if $y$ satisfies 
	(i) $\|y_{L_y}\|_p^p \le \alpha \cdot \tau^p$ and 
	(ii) $|H_y| \le \alpha$, 
then $H_y \subseteq I$.
\end{corollary}
\section{The Net Argument}\label{sec:net}
\subsection{Bounding the Norm}\label{sec:norm_x}

We will generally assume that for product $Ax$, the $x$ involved
is in  $\colspan(A^\top)$, which is the orthogonal complement
of the nullspace of $A$; any nullspace component of $x$ would not affect $Ax$ or $SAx$,
and so can be neglected for our purposes.

\begin{lemma}\label{lem A+b size}
When the entries of $A$ are integral,
for any nonempty $\cS\subset [n]$,
$\norm{A_{\cS,*}^+}_2 \le \norm{A}_2^d\CP(A)\sqrt{d} $,
and under also Assumption~\ref{as main}.\ref{as main size},
$\norm{A_{\cS,*}^+}_2 \le n^{O(d^2)}$.
\end{lemma}

\begin{proof}
When $\cS$ is a nonempty proper subset of $[n]$,  then
since $\norm{A_{\cS,*}}_2 \le \norm{A}_2$
and $\CP(A_{\cS,*})\le \CP(A)$, we have that if
$\norm{A_{\cS,*}^+}_2 \le \norm{A_{\cS,*}}_2^d\CP(A_{\cS,*})\sqrt{d}$,
then the lemma follows. So we can assume $S=[n]$.

First suppose $A$ has full column rank, so that
$A^\top A$ is invertible. For any $y\in\R^n$,
$A^+y$ is the unique solution $x^*$ of $A^\top A x = A^\top y$.
Applying Cramer's rule, the entries of $x^*$ have the form
$x_i = \frac{\det B_i}{\det A^\top A}$, where $B_i$ is the same as $A^\top A$,
except that the $i$'th column of $B_i$ is $A^\top y$.
The integrality of $A$ implies $|\det A^\top A| \ge 1$; using that together with Hadamard's determinant inequality
and the definition of the spectral norm, we have
$\norm{x^*}_2\le \norm{A}_2^d\CP(A)\norm{y}_2\sqrt{d}$. Since this holds for any $y$,
we have $\norm{A^+}_2 \le \norm{A}_2^d\CP(A)\sqrt{d}$ as claimed.

Now suppose $A$ has rank $k < d$.
Then there is $\cT\subset [d]$ of size $k$
whose members are indices of a set of $k$ linearly independent
columns of $A$. Moreover, if $x^* = A^+y$ is a solution to $\min_x\norm{Ax-y}_2$,
then there is another solution where the entries with indices
in $[d]\setminus \cT$ are zero, since a given column not in $\cT$
is a linear combination of columns in $\cT$. That is,
the solution to $\min_{x\in\R^k} \norm{A_{*,\cT}x - y}_2$
can be mapped directly to a solution $x^*$ in $\R^k$
with the same Euclidean norm. Since $A_{*,\cT}$
has full column rank, the analysis above implies
that
\[
\norm{x^*}_2 \le  \norm{A_{*,\cT}}_2^k \CP(A_{*,\cT}) \norm{y}_2\sqrt{k} \le \norm{A}_2^d\CP(A)\norm{y}_2\sqrt{d},
\]
so the bound on $\norm{A^+}_2$ holds also when $A$ has less than full rank.

The last statement of the lemma follows directly, using the definitions of $\norm{A}_2$, $\CP(A)$,
and Assumption~\ref{as main}.\ref{as main size}.
\end{proof}

\begin{lemma}\label{lem small x}
If $A$ has integral entries, and if Assumptions~\ref{as M}, \ref{as main}.\ref{as main size},
\ref{as main}.\ref{as tau size} hold, then Assumption~\ref{as main}.\ref{as main A int} holds.
%Our returned solution $\tx$ also will have $\norm{\tx} \le \xupper$.
\end{lemma}

\begin{proof}
Let $x_{M}^{C_1}$ be a $C_1$-approximate solution of $\min_x \norm{Ax-b}_M$,
which Assumption~\ref{as main}.\ref{as main A int} requires to have bounded Euclidean norm.
Let $\hM(a) \equiv \min\{\tau^p, |a|^p\}$, so that
Assumptions~\ref{as M}.\ref{as M near quad} and \ref{as M}.\ref{as M flat} imply that
$L_M \hM(a) \le M(a) \le U_M \hM(a)$ for all $a$.
Letting $x^*_M\equiv \argmin_x \norm{Ax-b}_M$, and similarly
defining~$x^*_\hM$, this condition implies that 
\begin{align}
\norm{Ax_{M}^{C_1}-b}_\hM
	  & \le \frac{1}{L_M}\norm{A x_{M}^{C_1} - b}_M \nonumber
	\\ & \le \frac{C_1}{L_M}\norm{A x^*_M - b}_M \nonumber
	\\ & \le C_2 \norm{A x^*_M - b}_\hM \nonumber
	\\ & \le C_2\norm{A x^*_\hM - b}_\hM, \qquad \label{eq hx upper}
\end{align}
where $C_2 \equiv C_1 U_M / L_M$.

Let $\cS$ denote the set of indices at which $|A_{i,*}x_{M}^{C_1} - b_i| \le \tau$.
If $\cS$ is empty, then $x_{M}^{C_1}$ can be assumed to be zero.

Similarly to our general assumption that $x_{M}^{C_1}\in\colspan(A^\top)$,
we can assume that $x_{M}^{C_1}\in\colspan(A_{\cS,*}^\top)$,
since any component of $x_{M}^{C_1}$ in the nullspace of $A_{\cS,*}$ can be removed
without changing $A_{\cS,*}x_{M}^{C_1}$, and without increasing the $n-|S|$ contributions
of $\tau^p$ from the remaining summands in $\norm{Ax_{M}^{C_1} - b}_M$.
(Here we used Assumption~\ref{as M}.\ref{as M flat} that $M(a) = \tau^p$ for $|a| \ge \tau$.)

From $x_{M}^{C_1}\in\colspan(A^\top)$ it follows that
$\norm{x_{M}^{C_1}}_2 = \norm{A_{\cS,*}^+ A_{\cS,*}x_{M}^{C_1}}_2 \le  \norm{A_{\cS,*}^+}_2 \norm{A_{\cS,*}x_{M}^{C_1}}_2$,
and since
\begin{align*}
\norm{A_{\cS,*}x_{M}^{C_1}}_2 & \le \sqrt{n} \norm{A_{\cS,*}x_{M}^{C_1}}_p\\
	     & \le \sqrt{n} (\norm{A_{\cS,*}x_{M}^{C_1} - b_S}_p + \norm{b_S}_p)
	   \\ & \le C_2 \sqrt{n} (\norm{A x^*_\hM - b}_\hM^{1/p}+ \norm{b_S}_p) \qquad (\mathrm{by\ }\eqref{eq hx upper})
	   \\ & \le 2C_2 \sqrt{n}\norm{b}_p,
\end{align*}
we have
$\norm{x_{M}^{C_1}}_2 \le \norm{A_{\cS,*}^+}_2 \norm{A_{\cS,*}x_{M}^{C_1}}_2 \le \norm{A_{\cS,*}^+}_2 2C_2 \sqrt{n}\norm{b}_p$,
and so from Lemma~\ref{lem A+b size} and Assumption~\ref{as main}.\ref{as main size}, the bound
on $\norm{x_{M}^{C_1}}_2$ of Assumption~\ref{as main}.\ref{as main A int} follows.
\end{proof}

\subsection{Net Constructions}

%We consider a subset of $\R^d$ under the semi-norm $\norm{x}_A = \norm{Ax}$.
%Our assumption that any $x$ or $y$ under consideration is in $\colspan(A^\top)$ implies that for such $x$ and $y$,
%$\norm{x-y}_A = 0$ if and only if $x=y$.

\begin{lemma}\label{lem:net1}
Under the given assumptions, for $\xupper$ as in Assumption~\ref{as main}.\ref{as main A int},
there exists a set $\mathcal{N}_{\varepsilon} \subseteq \colspan([A~b])$ 
with size $|\mathcal{N}_{\varepsilon}| \le n^{O(d^3)} \cdot (1 / \varepsilon)^{O(d)}$, 
such that for any $x$ satisfying $\|x\|_2 \le U$, 
there exists $y' \in \mathcal{N}_{\varepsilon}$ such that 
\[
\|(Ax - b) - y'\|_M \le \eps^p.
\]
\end{lemma}

\begin{proof}
Let $\hM(a) \equiv \min\{\tau^p, |a|^p\}$.
Assume for now that $\eps\le \tau / 2$, so that if $\norm{Ax}_\hM\le\eps^p$, then every entry of $Ax$ is no more than $\tau$
in magnitude, and so $\norm{Ax}_\hM = \norm{Ax}_p^p$.

Let 
\[
B_{\eps} \equiv \{Ax-b \mid \norm{Ax-b}_\hM \le \eps^p \} = \{Ax-b \mid \norm{Ax-b}_p \le \eps \}
\]
and
\[
B_{U} \equiv \{Ax-b \mid \|x\|_2 \le U\} \subseteq \{Ax-b \mid \|Ax - b\|_p \le \sqrt{n} \cdot ( \norm{A}_2\xupper+\norm{b}_2)\} .
\]

From the scale invariance of the $\ell_p$ norm, and the volume in at-most $d$
dimensions, $\Vol(B_{\eps} ) \ge (\eps/(\sqrt{n} \cdot (\norm{A}_2\xupper+\norm{b}_2)))^d \Vol(B_\xupper)$,
so that at most $(\sqrt{n} \cdot (\norm{A}_2\xupper+\norm{b}_2)/\eps)^d$ translates of $B_{\eps} $ can
be packed into $B_\xupper$ without intersecting. Thus the set $\cN_{\varepsilon}$ of centers
of such a maximal packing of translates is an $\eps^p$-cover of
$B_\xupper$, that is, for any point $y\in B_\xupper$,
there is some $y'\in\cN$ such that $\norm{y'-y}_p \le \eps$, so that
$\norm{y'-y}_\hM\le\eps^p$.
%That is $\cN$ is an $\eps^2$-net for points $Ax$ with respect to
%$\norm{\cdot}_\hM$, for relevant $x$.

If $\eps > \tau / 2$, we just note that a $(\tau / 2)^p$-cover is also an $\eps^p$-cover, and so there is an $\eps^p$-cover
of size $(\sqrt{n} \cdot (\norm{A}_2\xupper+\norm{b}_2)/\min\{\tau / 2, \eps\})^d$.

Plugging in the bounds for $\xupper$ from Assumption~\ref{as main}.\ref{as main A int},
and for $\tau$, $\norm{b}_2$,
and $\norm{A}_2\le \max_{i\in [d]} \norm{A_{*,i}}_2$
from Assumptions~\ref{as main}.\ref{as main size} and \ref{as main}.\ref{as tau size},
the cardinality bound of the lemma follows.

This argument is readily adapted to more general $\norm{\cdot}_M$, 
by noticing that $\|y - y'\|_{M} \le U_M \cdot \|y - y'\|_{\hat{M}}$
using Assumption~\ref{as M}.\ref{as M near quad}
and adjusting constants. 
\end{proof}

\begin{lemma}\label{lem:net2}
Under the given assumptions, there exists a set
$\mathcal{M}_{\varepsilon}^{\alpha, \beta}\subseteq \colspan([A~b])$ with size 
$|\mathcal{M}_{\varepsilon}^{\alpha, \beta}| \le  O\left( \frac{\beta / \alpha}{\varepsilon} \right) \cdot n^{O(d^2)} \cdot (1 / \varepsilon)^{O(d)}$,
such that for any $x$ satisfying $\alpha \le \|Ax - b\|_p \le \beta \le \tau$,
there exists $y' \in \mathcal{M}_{\varepsilon}^{\alpha, \beta}$ such that
$$
\|(Ax - b) - y'\|_M \le \varepsilon^p \cdot \|Ax - b\|_M.
$$ 
\end{lemma}
\begin{proof}
We assume $\varepsilon \le \tau$, since otherwise we can take $\varepsilon$ to be $\tau$.
By standard constructions (see, e.g., {\cite[p.~48]{w14}}), there exists a set
$\mathcal{M}_{\gamma} \subseteq \colspan([A~b])$ with size $|\mathcal{M}_{\gamma}| \le (1 / \varepsilon)^{O(d)}$,
such that for any $y = Ax - b$ with $\|y\|_p = \gamma$, there exists
$y' \in \mathcal{M}_{\gamma}$ such that $\|y- y'\|_p \le \gamma \cdot \varepsilon$.

Let $\mathcal{M}_{\varepsilon}^{\alpha, \beta}$ be 
$$
\mathcal{M}_{\varepsilon}^{\alpha, \beta}= \mathcal{M}_{\alpha} \cup \mathcal{M}_{(1 + \varepsilon) \alpha} \cup \mathcal{M}_{(1 + \varepsilon)^2 \alpha} \cup \cdots \cup \mathcal{M}_{\beta}.
$$
Clearly, by Assumption \ref{as main},
$$
|\mathcal{M}_{\varepsilon}^{\alpha, \beta}| \le \log_{1 + \varepsilon}(\beta / \alpha)  \cdot n^{O(d^2)} \cdot (1 / \varepsilon)^{O(d)} \le  O\left( \frac{\beta / \alpha}{\varepsilon} \right) \cdot n^{O(d^2)} \cdot (1 / \varepsilon)^{O(d)}.
$$
Now we show that $\mathcal{M}_{\varepsilon}^{\alpha, \beta}$ satisfies the desired properties. 
For any $x \in \R^d$ such that $y = Ax - b$ satisfies $\alpha \le \|y\|_p \le \beta \le \tau$,
we must have $|y_i| \le \tau$ for all entries of $y$.
By normalization, there exists $\hat{y}$ such that $\|y - \hat{y}\|_p \le \varepsilon \cdot \|y\|_p$
and $\|\hat{y}\|_p = (1 + \varepsilon)^i \cdot \alpha$ for some $i \in \mathbb{N}$.
Furthermore, by the property of $\mathcal{M}_{(1 + \varepsilon)^i \alpha}$,
there exists $y' \in \mathcal{M}_{(1 + \varepsilon)^i \alpha} \subseteq \mathcal{M}_{\varepsilon}^{\alpha, \beta}$
such that $\|\hat{y} - y'\|_p \le \varepsilon \cdot \|y'\|_p \le 2 \varepsilon \cdot \|y\|_p$.
Thus, by triangle inequality, we have $\|y - y'\|_p \le 3\varepsilon  \|y\|_p$.
For sufficiently small $\varepsilon$, since $\|y\|_p \le \tau$, we also have
$\|y - y'\|_p \le \tau$, which implies $\|y - y'\|_{\infty} \le \tau$.
Thus, using Assumption \ref{as M}.\ref{as M near quad}, we have
$$
\|y - y'\|_M \le U_M \|y - y'\|_p^p \le U_M \cdot (3 \varepsilon)^p  \cdot \|y\|_p^p \le U_M / L_M (3 \varepsilon)^p \|y\|_M.
$$
Adjusting constants implies the desired properties. 
\end{proof}

\subsection{The Net Argument} \label{sec:net_arg}
\begin{theorem}\label{thm:net}
For any $A \in \R^{n \times d}$ and $b \in \R^n$, 
given a matrix $S \in \mathbb{R}^{r \times n}$ and a weight vector $w \in \mathbb{R}^n$ such that $w_i \ge 0$ for all $i \in [n]$.
Let $c = \min_{x} \|Ax - b\|_p$.
If there exist $U_O, U_A, L_A, L_N \le \poly(n)$ such that
\begin{enumerate}
\item $\|S(Ax_{M}^* - b)\|_{M, w} \le U_{O} \|Ax_{M}^* - b\|_M$, where $x_{M}^* = \argmin_x \|Ax - b\|_M$;
\item $L_A \|Ax - b\|_M \le \|S(Ax - b)\|_{M, w} \le U_A \|Ax - b\|_M$ for all $x \in \mathbb{R}^d$;
\item $\|Sy\|_{M, w} \ge L_{N} \|y\|_M$ for all $y \in \mathcal{N}_{\poly(\varepsilon \cdot \tau / n)} \cup \mathcal{M}_{\poly(\varepsilon / n)}^{c, c \cdot \poly(n)}$,
\end{enumerate}
then, any $C$-approximate solution of $\min_x \|S(Ax - b)\|_{M, w}$ with $C \le \poly(n)$
is a $C \cdot (1 + O(\varepsilon)) \cdot U_O / L_N$-approximate solution of $\min_x \|Ax - b\|_M$.
Here $\mathcal{N}_{\poly(\varepsilon \cdot \tau / n)}$ and
$\mathcal{M}_{\poly(\varepsilon / n)}^{c, c \cdot \poly(n)}$ are as defined
in Lemma~\ref{lem:net1} and Lemma \ref{lem:net2}, respectively.  
\end{theorem}
\begin{proof}
We distinguish two cases in the proof.

\paragraph{Case 1: $(C \cdot U_M \cdot U_A / (L_M \cdot L_A)) \cdot c^p \le \tau^p$.} 
In this case, we prove that any $C$-approximate solution $x_{S,M,w}^C$ of $\min_x \|S(Ax - b)\|_{M, w}$
satisfies $c \le \|Ax_{S,M,w}^C - b\|_p \le (C \cdot U_M \cdot U_A / (L_M \cdot L_A))^{1 / p} \cdot c \le \tau$.
Let $x_{p}^* = \argmin_x \|Ax - b\|_p$, we have 
\begin{align*}
&\|A x_{S,M,w}^C - b\|_M\\
 \le & \|S(Ax_{S,M,w}^C - b)\|_{M, w} / L_A \\
 \le & C \cdot \|S(Ax_{p}^* - b)\|_{M, w} / L_A \\
 \le & C \cdot \|Ax_{p}^* - b\|_M \cdot U_A / L_A\\
 \le & C \cdot \|Ax_{p}^* - b\|_p^p \cdot (U_M \cdot U_A) / L_A\\
 =   & C \cdot c^p \cdot (U_M \cdot U_A) / L_A.
\end{align*}
Since $L_M \le 1$, this implies $\|Ax_{S,M,w}^C - b\|_M \le \tau^p$, which implies $\|Ax_{S,M,w}^C - b\|_{\infty} \le \tau$.
Thus, $\|Ax_{S,M,w}^C - b\|_p^p \le \|Ax_{S,M,w}^C - b\|_M / L_M \le (C \cdot U_M \cdot U_A / (L_M \cdot L_A)) \cdot c^p$, which implies
$
\|Ax_{S,M,w}^C - b\|_p \le (C \cdot U_M \cdot U_A / (L_M \cdot L_A))^{1 / p} \cdot c.
$
Moreover, by the definition of  $c$ we have $\|Ax_{S,M,w}^C - b\|_p \ge c$.

Since $(C \cdot U_M \cdot U_A / (L_M \cdot L_A))^{1 / p} \le \poly(n)$,
by Lemma \ref{lem:net2}, there exists $y' \in \mathcal{M}_{\poly(\varepsilon / n)}^{c, c \cdot \poly(n)}$
such that $\|(Ax_{S,M,w}^C - b) - y'\|_{M} \le \poly(\varepsilon / n) \cdot \|Ax_{S,M,w}^C - b\|_M$.
Notice that
$$
\|S(Ax_{S,M,w}^C  - b)\|_{M, w} = \|Sy' + S((Ax_{S,M,w}^C  - b) - y')\|_{M, w}.
$$
For $Sy'$, since $y' \in \mathcal{M}_{\poly(\varepsilon / n)}^{c, c \cdot \poly(n)}$, we have
$$
\|Sy'\|_{M, w} \ge L_N \|y'\|_M = L_N \|Ax_{S,M,w}^C  - b + (y' - (Ax_{S,M,w}^C  - b))\|_M.
$$
Since $\|y' - (Ax_{S,M,w}^C  - b)\|_M \le \poly(\varepsilon / n) \cdot \|Ax_{S,M,w}^C  - b\|_M$,
by Lemma \ref{lem:perturb}, we have $\|Ax_{S,M,w}^C  - b + (y' - (Ax_{S,M,w}^C  - b))\|_M \ge (1 - \varepsilon) \|Ax_{S,M,w}^C  - b\|_M$, which implies
$\|Sy'\|_{M, w} \ge L_N (1 - \varepsilon) \|Ax_{S,M,w}^C  - b\|_M$.
On the other hand, $\|S((Ax_{S,M,w}^C  - b) - y')\|_{M, w} \le U_A \|(Ax_{S,M,w}^C  - b) - y'\|_{M} \le \poly(\varepsilon / n)  \cdot \|Ax_{S,M,w}^C  - b\|_M$.
Again by Lemma \ref{lem:perturb}, we have $\|S(Ax_{S,M,w}^C  - b)\|_{M, w} \ge (1 - \varepsilon) \|Sy'\|_{M, w} \ge L_N (1 - O(\varepsilon))\|Ax_{S,M,w}^C  - b\|_M$.
Furthermore, since $x_{S,M,w}^C $ is a $C$-approximate solution of $\min_x \|S(Ax - b)\|_{M, w}$, we must have
\begin{align*}
\|Ax_{S,M,w}^C -b\|_M& \le (1 + O(\varepsilon)) / L_N \cdot \|S(Ax_{S,M,w}^C  -b)\|_{M, w} \\
&\le C \cdot (1 + O(\varepsilon)) / L_N \cdot \|S(Ax_{M}^* -b)\|_{M, w} \\
& \le C \cdot (1 + O(\varepsilon)) \cdot U_O / L_N \cdot \|Ax_{M}^*-b\|_M.
\end{align*}

\paragraph{Case 2: $(C \cdot U_M \cdot U_A / (L_M \cdot L_A)) \cdot c^p \ge \tau^p$.} 
In this case, we first prove that any $C$-approximate solution $x_{S,M,w}^C$ of
$\min_x \|S(Ax - b)\|_{M, w}$ is a $\poly(n)$-approximate solution of $\min_x \|Ax -
b\|_M$. By Assumption~\ref{as main}.\ref{as main A int}, this implies all
$C$-approximate solution $x_{S,M,w}^C$ of $\min_x \|S(Ax - b)\|_{M, w}$ satisfies ${\|x_{S,M,w}^C\|_2 \le U}$.

Consider any $C$-approximate solution $x_{S,M,w}^C$ of $\min_x \|S(Ax - b)\|_{M, w}$, we have
\begin{align*}
\|Ax_{S,M,w}^C - b\|_M \le \|S(Ax_{S,M,w}^C - b)\|_{M, w} / L_A
		   & \le C \cdot \|S(Ax_{M}^* - b)\|_{M, w} / L_A
		\\ & \le  C \cdot U_A / L_A \cdot \|Ax_{M}^* - b\|_{M}
		\le \poly(n) \cdot \|Ax_{M}^* - b\|_{M}.
\end{align*}

We further show that $\|Ax - b\|_M \ge  \tau^p / \poly(n)$ for all $x \in \R^d$.
If $\|Ax-b\|_{\infty} \ge \tau$, then the statement clearly holds.
Otherwise, $\|Ax-b\|_M \ge L_M \cdot \|Ax-b\|_p^p \ge L_M c^p \ge L_M^2 L_A / (C \cdot U_M \cdot U_A) \cdot \tau^p \ge \tau^p / \poly(n)$. 
Thus, for any $C$-approximate solution $x_{S,M,w}^C$ of $\min_x \|S(Ax - b)\|_{M, w}$,
there exists $y' \in \mathcal{N}_{\poly(\varepsilon \cdot \tau / n)}$ such that
$$
\|y' - (Ax_{S,M,w}^C - b)\|_M \le  \poly(\varepsilon \cdot \tau / n) \le \poly(\varepsilon / n) \cdot \|Ax_{S,M,w}^C - b\|_M.
$$
The rest of the proof is exactly the same as that of Case 1. 
\end{proof}

% !TEX root = main.tex

\section{A Row Sampling Algorithm for Tukey Loss Functions} \label{sec:row_sample}
In this section we present the row sampling algorithm. The row sampling
algorithm proceeds in a recursive manner. We describe a single recursive step in
Section \ref{sec:recursive_one_step} and the overall algorithm in Section
\ref{sec:recur_alg}.

%Suppose we want to solve the Tukey regression problem $\min_x \|Ax - b\|_M$, where $A \in \mathbb{R}^{n \times d}$ and $b \in \R^n$ satisfy Assumption \ref{as main}, 
%and $M(\cdot)$ is a loss function that satisfies Assumption \ref{as M}. 

\subsection{One Recursive Step}\label{sec:recursive_one_step}
The goal of this section is to design one recursive step of the row sampling
algorithm. For a weight vector $w \in \R^n$, the recursive step outputs a
sparser weight vector $w' \in \R^n$ such that for any set $\mathcal{N} \subseteq
\colspan(A)$ with size $|\mathcal{N}|$, with probability at least $1 -
\deltaos$, simultaneously for all $y \in \mathcal{N}$, 
$$
\|y\|_{M, w'} = (1 \pm \varepsilon) \|y\|_{M, w}.
$$

We maintain that if $w_i \neq 0$, then $w_i \ge 1$ and $\|w\|_{\infty} \le n^2$
as an invariant in the recursion. These conditions imply that we can partition
the positive coordinates of $w$ into $2 \log n$ groups $P_j$, for which $P_j =
\{i \mid 2^{j - 1} \le w_i < 2^j\}$.

Now we define one recursive step of our sampling procedure. We split the matrix
$A$ into $A_{P_1, *}, A_{P_2, *}, \ldots, A_{P_{2 \log n}, *}$, and deal with
each of them separately. 
For each $1 \le j \le 2 \log n$, we invoke the algorithm in Theorem
\ref{thm:struct2} to identify a set $I_j$ for the matrix $A_{P_j, *}$, for some
parameter $\alpha$ and $\deltastruct$ to be determined. 
For each $1 \le j \le 2 \log n$, we also use the algorithm in Theorem
\ref{thm:alg_lewis} to calculate $\{\hat{u}_i\}_{i \in P_j}$ such that
$u_i \le \hat{u}_i \le 2 u_i$ 
where  $\{u_i\}_{i \in P_j}$ are the $\ell_p$ Lewis weights of the matrix $A_{P_j, *}$.

Now for each $i \in P_j$, we define its sampling probability $p_i$ to be
$$
p_i =  \begin{cases}
1 & i \in I_j\\
\min\{1, 1 / 2 + \Theta( 
	d^{\max\{0, p / 2 - 1\}}  \hat{u}_i \cdot Y)\} & i \notin I_j
\end{cases},
$$
where $Y\equiv d \log(1 / \varepsilon)+ \log(\log n / \deltaos) + U_M / L_M \log (|\mathcal{N}| \cdot \log n / \deltaos) / \varepsilon^2$.

For each $i \in [n]$, we set $w_i' = 0$ with probability $1 - p_i$, and set
$w_i' = w_i / p_i$ with probabliity $p_i$. The finishes the definition of one
step of the sampling procedure.

Let 
$$
F \equiv \sum_{1 \le j \le 2 \log n} |I_j| + \sum_{1 \le j \le 2 \log n}\sum_{i \in P_j \setminus I_j}
\Theta( 
	d^{\max\{0, p / 2 - 1\}}  \hat{u}_i \cdot Y).
$$

Our first lemma shows that with probability at least $1 - \deltaos$, the number of non-zero entries in $w'$
is at most $\frac23 \|w\|_0$, provided $\|w\|_0$ is large enough. 
\begin{lemma}\label{lem:size_decrease}
When $\|w\|_0 \ge 10F$, with probability at least $1 - \deltaos$, 
$$
\|w'\|_0 \le \frac23 \|w\|_0.
$$
\end{lemma}
\begin{proof}
Notice that
$$
\E[\|w'\|_{0}] \le \|w\|_0 / 2 + F.
$$
By Bernstein's inequality in Lemma \ref{lem:bernstein}, since
$F \ge \Omega(\log (1 / \deltaos))$, with probability at least
$1 - \exp(-\Omega(\|w\|_0)) \ge 1 - \exp(-\Omega(F)) \ge 1 - \deltaos$, we have
$$
\|w'\|_{0} \le \|w\|_0 / 2 + F + \|w\|_0 / 10 \le \frac23 \|w\|_0.
$$
\end{proof}

Our second lemma shows that $\|w'\|_{\infty}$ is upper bounded by $2\|w\|_{\infty}$.
\begin{lemma}\label{lem:weight_increase}
$\|w'\|_{\infty} \le 2\|w\|_{\infty}$.
\end{lemma}
\begin{proof}
Since $p_i \ge 1/2$ for all $i \in [n]$, we have $\|w'\|_{\infty} \le 2\|w\|_{\infty}$.
\end{proof}

%We further show that each recursive step of the sampling algorithm can be implemented in input-sparsity time.
%\begin{lemma}
%Each step of the sampling algorithm runs in $?$ time.
%\end{lemma}
%\begin{proof}
%In each recursive step, for each $1 \le j \le 2 \log n$, we invoke the algorithm in Theorem \ref{thm:alg_lewis} and the algorithm in Theorem \ref{thm:struct2} on $A_{P_1, *}, A_{P_2, *}, \ldots, A_{P_{2 \log n}, *}$
%\end{proof}

We show that for sufficiently large constant $C$,
if we set $$\alpha = C \cdot U_M / L_M \cdot \log( |\mathcal{N}| \cdot \log n / \deltaos) / \varepsilon^2$$
and $\deltastruct = \deltaos / (4 \log n)$, 
then with probability at least $1 - \deltaos$, simultaneously
for all $y \in \mathcal{N}$ we have
$$
\|y\|_{M, w'} = (1 \pm \varepsilon) \|y\|_{M, w}.
$$
%By Theorem \ref{thm:struct2} and Theorem \ref{thm:fa_lewis}, since 
%$$
%\sum_{1 \le j \le 2 \log n}\sum_{i \in P_j \setminus I_j} \hat{u}_i \le O(d \log n),
%$$
%we have
%\begin{align*}
%&\sum_{1 \le j \le 2 \log n} |I_j| + \sum_{1 \le j \le 2 \log n}\sum_{i \in P_j \setminus I_j}\Theta( d^{\max\{0, p / 2 - 1\}}  \hat{u}_i \cdot (d +  \log(1 / \deltaos) + U_M / L_M (\log (|\mathcal{N}| \cdot \log n / \deltaos)) / \varepsilon^2) ) \\
%\le & \widetilde{O}(d^{\max\{1, p / 2\}} \log n \cdot \alpha \cdot \log (1 / \deltastruct) + d^{\max\{1, p / 2\}} \log n \cdot (d +  \log(1 / \deltaos) + U_M / L_M (\log (|\mathcal{N}  | \cdot \log n / \deltaos)) / \varepsilon^2)) \\
%\triangleq & F.
%\end{align*}
By Theorem \ref{thm:struct2} and Theorem \ref{thm:fa_lewis}, since 
$$
\sum_{1 \le j \le 2 \log n}\sum_{i \in P_j \setminus I_j} \hat{u}_i \le O(d \log n),
$$
this also implies 
$$
F =  \widetilde{O}(d^{\max\{1, p / 2\}}  \log n  \cdot (  \log(  |\mathcal{N}|/ \deltaos ) \cdot \log (1/ \deltaos)  + d) / \varepsilon^2).
$$

Furthermore, for each $1 \le j \le 2 \log n$, we invoke the algorithm in
Theorem \ref{thm:alg_lewis} and the algorithm in Theorem \ref{thm:struct2}
on $A_{P_1, *}, A_{P_2, *}, \ldots, A_{P_{2 \log n}, *}$, and thus
the running time of each recursive step is thus upper bounded by 
$$
\widetilde{O}((\nnz(A) + d^{p / 2 + O(1)} \cdot \alpha) \cdot \log(1 / \deltastruct))
	= \widetilde{O}((\nnz(A) + d^{p / 2 + O(1)}  \cdot \log( |\mathcal{N}|  / \deltaos)  \cdot / \varepsilon^2) \cdot \log(1 / \deltaos) ).
$$

Now we consider a fixed vector $y \in \colspan(A)$.
We use the following two lemmas in our analysis.

\begin{lemma}\label{lem:sample_heavy}
With probability $1 - \deltaos / O(|\mathcal{N}| \cdot \log n)$, the following holds:
\begin{itemize}
\item If $\|y_{H_y \cap P_j}\|_{M, w} \ge C \cdot U_M  \cdot \tau^p \cdot 2^{j - 1} \cdot \log(|\mathcal{N}| \cdot \log n / \deltaos) / \varepsilon^2$,
then 
$$
\|y_{H_y \cap P_j}\|_{M, w'} = (1 \pm \varepsilon / 2) \|y_{H_y \cap P_j}\|_{M, w};
$$
\item If $\|y_{H_y \cap P_j}\|_{M, w}  < C \cdot U_M \cdot \tau^p \cdot 2^{j - 1}  \cdot \log(|\mathcal{N}| \cdot \log n / \deltaos) / \varepsilon^2$, then 
$$
|\|y_{H_y \cap P_j}\|_{M, w'} -  \|y_{H_y \cap P_j}\|_{M, w}| \le C \cdot U_M \cdot \tau^p \cdot 2^{j - 2}\cdot  \log(|\mathcal{N}| \cdot \log n / \deltaos) / \varepsilon.
$$
\end{itemize}
\end{lemma}
\begin{proof}
%Notice that $\|y_{H_y \cap P_j}\|_{M, w} = \sum_{i \in H_y \cap P_j} w_i M(y_i) \ge |H_y \cap P_j| \cdot 2^{j - 1} \cdot L_M \cdot \tau^p$. 
For each $i \in H_y \cap P_j$, we use $Z_i$ to denote the random variable 
$$
Z_i = \begin{cases}
w_i M(y_i) / p_i & \text{with probability $p_i$} \\
0 & \text{with probability $1 - p_i$}
\end{cases}.
$$
Since $Z_i = w_i' M(y_i)$, we have
$$
\|y_{H_y \cap P_j}\|_{M, w'} = \sum_{i \in H_y \cap P_j} Z_i.
$$
It is clear that $Z_i \le 2^{j + 1} \cdot U_M \cdot \tau^p$
since $p_i \ge 1 / 2$ and $w_i \le 2^j$, $\E[Z_i] = w_i M(y_i)$ and $\E[Z_i^2] = w_i^2 (M(y_i))^2 / p_i$.
By H\"older's inequality, 
$$
\sum_{i \in H_y \cap P_j} \E[Z_i^2] \le 2^{j + 1} \cdot \|y_{H_y \cap P_j}\|_{M, w}\cdot U_M \cdot \tau^p.
$$
Thus by Bernstein's inequality in Lemma \ref{lem:bernstein}, we have
$$
\Pr\left[\left|\sum_{i \in H_y \cap P_j} Z_i -  \|y_{H_y \cap P_j}\|_{M, w}\right|
	\ge  t\right] \le 2 \exp \left( -\frac{t^2}{2^{j + 2} \cdot U_M \cdot  \tau^p\cdot t / 3 + 2^{j + 2} \cdot \|y_{H_y \cap P_j}\|_{M, w} \cdot  U_M \cdot \tau^p}\right).
$$

When 
$$
\|y_{H_y \cap P_j}\|_{M, w}
	\ge C \cdot U_M  \cdot \tau^p \cdot 2^{j - 1} \cdot \log(|\mathcal{N}| \cdot \log n / \deltaos) / \varepsilon^2,
$$
we take 
$$
t = \varepsilon / 2 \cdot \|y_{H_y \cap P_j}\|_{M, w}
	\ge C \cdot U_M  \cdot \tau^p \cdot 2^{j - 2} \cdot \log( |\mathcal{N}| \cdot \log n / \deltaos) / \varepsilon.
$$
By taking $C$ to be some sufficiently large constant,
with probability at least $1 - \deltaos / O(|\mathcal{N}| \cdot \log n)$, 
$$
\|y_{H_y \cap P_j}\|_{M, w'} = (1 \pm \varepsilon / 2) \|y_{H_y \cap P_j}\|_{M, w}.
$$

When
$$
\|y_{H_y \cap P_j}\|_{M, w} < C \cdot U_M  \cdot \tau^p \cdot 2^{j - 1} \cdot \log(|\mathcal{N}| \cdot \log n / \deltaos) / \varepsilon^2,
$$
we take 
$$
t =  C \cdot U_M \cdot \tau^p \cdot 2^{j - 2}\cdot  \log(|\mathcal{N}| \cdot \log n / \deltaos) / \varepsilon.
$$
By taking $C$ to be some sufficiently large constant, with probability at least $1 - \deltaos / O(|\mathcal{N}| \cdot \log n)$, 
$$
|\|y_{H_y \cap P_j}\|_{M, w'} -  \|y_{H_y \cap P_j}\|_{M, w}|
	\le C \cdot U_M \cdot \tau^p \cdot 2^{j - 2}\cdot  \log(|\mathcal{N}| \cdot \log n / \deltaos) / \varepsilon.
$$
\end{proof}

The proof of the following lemma is exactly the same as Lemma \ref{lem:sample_heavy}.
\begin{lemma}\label{lem:sample_light}
With probability $1 - \deltaos / O(|\mathcal{N}| \cdot \log n)$, the following holds:
\begin{itemize}
\item If $\|y_{L_y \cap P_j}\|_{M, w} \ge C \cdot U_M  \cdot \tau^p \cdot 2^{j - 1} \cdot \log(|\mathcal{N}| \cdot \log n / \deltaos) / \varepsilon^2$, then 
$$
\|y_{L_y \cap P_j}\|_{M, w'} = (1 \pm \varepsilon / 2) \|y_{L_y \cap P_j}\|_{M, w};
$$
\item If $\|y_{L_y \cap P_j}\|_{M, w}  < C \cdot U_M \cdot \tau^p \cdot 2^{j - 1}  \cdot \log(|\mathcal{N}| \cdot \log n / \deltaos) / \varepsilon^2$, then 
$$
|\|y_{L_y \cap P_j}\|_{M, w'} -  \|y_{L_y \cap P_j}\|_{M, w}| \le C \cdot U_M \cdot \tau^p \cdot 2^{j - 2}\cdot  \log(|\mathcal{N}| \cdot \log n / \deltaos) / \varepsilon.
$$
\end{itemize}
\end{lemma}

Now we use Lemma \ref{lem:sample_heavy} and Lemma \ref{lem:sample_light}
to analyze the sampling procedure.
\begin{lemma}\label{lem:os_one}
If we set $\alpha = C  \cdot U_M / L_M \cdot \log(|\mathcal{N}| \cdot \log n / \deltaos) / \varepsilon^2$, $\deltastruct = \deltaos / (4 \log n)$,
then for each $1 \le j \le 2 \log n$, with probability at least $1 - \deltaos / (2 \log n)$, simultaneously for all $y \in \mathcal{N}$,
$$
\|y_{P_j}\|_{M, w'} = (1 \pm \varepsilon) \|y_{P_j}\|_{M, w}.
$$
Applying a union bound over all $1 \le j \le 2 \log n$,
with probability at least $1 - \deltaos$, simultaneously for all $y \in \mathcal{N}$,
$$
\|y\|_{M, w'} = (1 \pm \varepsilon) \|y\|_{M, w}.
$$
\end{lemma}
\begin{proof}
By Theorem \ref{thm:struct2}, 
for each $1 \le j \le 2 \log n$,
with probability $1 - \deltaos / (4 \log n)$, 
simultaneously for all $y \in \mathcal{N} \subseteq \colspan(A)$,
if $y$ satisfies 
	(i) $\|y_{L_y \cap P_j}\|_p^p \le \alpha \cdot \tau^p$ and 
	(ii) $|H_y \cap P_j| \le \alpha$, 
then we have $H_y \cap P_j \subseteq I_j$. 
We condition on this event in the remaining part of the proof. 

Now we consider a fixed $y \in \mathcal{N}$.
We show that $\|y_{P_j}\|_{M, w'} = (1 \pm \varepsilon) \|y_{P_j}\|_{M, w}$
with probability at least $1 - \deltaos / O(|\mathcal{N}| \cdot \log n)$.
The desired bound follows by applying a union bound over all $y \in \mathcal{N}$.

We distinguish four cases in our analysis.
We use $T$ to denote a fixed threshold
$$
T = C \cdot U_M  \cdot \tau^p \cdot 2^{j - 1} \cdot \log(|\mathcal{N}| \cdot \log n / \deltaos) / \varepsilon^2.
$$
\begin{description}
\item[Case (i): $\|y_{H_y \cap P_j}\|_{M, w} < T$ and $\|y_{L_y \cap P_j}\|_{M, w}< T$. ] 
Since $\|y_{H_y \cap P_j}\|_{M, w} < T$, we must have 
$$
|H_y \cap P_j| < C \cdot U_M / L_M \cdot \log(|\mathcal{N}| \cdot \log n / \deltaos) / \varepsilon^2 = \alpha.
$$
Furthermore, we also have
$$
\|y_{L_y \cap P_j}\|_p^p < C \cdot U_M / L_M \cdot \tau^p \cdot \log(|\mathcal{N}| \cdot \log n / \deltaos) / \varepsilon^2= \alpha \cdot \tau^p.
$$
By Lemma \ref{lem:single_lp}, with probability at least $1 - \deltaos / O(|\mathcal{N}| \cdot \log n)$, we have
$$
\|y_{P_j \setminus I_j}\|_{M, w'} = (1 \pm \varepsilon)\|y_{P_j \setminus I_j}\|_{M, w},
$$
since $H_y \cap P_j \subseteq I_j$.
Moreover, $\|y_{I_j}\|_{M, w} = \|y_{I_j}\|_{M, w'}$ since $w_i = w_i'$ for all $i \in I_j$.
Thus, we have $\|y_{P_j}\|_{M, w'} = (1 \pm \varepsilon) \|y_{P_j}\|_{M, w}$.
\item[Case (ii): $\|y_{H_y \cap P_j}\|_{M, w} \ge T$ and $\|y_{L_y \cap P_j}\|_{M, w} \ge T$. ] 
By Lemma \ref{lem:sample_heavy} and Lemma \ref{lem:sample_light}, with probability at least $1 - \deltaos / O(|\mathcal{N}| \cdot \log n)$, 
$$
\|y_{H_y \cap P_j}\|_{M, w'} = (1 \pm \varepsilon / 2) \|y_{H_y \cap P_j}\|_{M, w}
$$
and
$$
\|y_{L_y \cap P_j}\|_{M, w'} = (1 \pm \varepsilon / 2)  \|y_{L_y \cap P_j}\|_{M, w},
$$
which implies
$$
\|y_{P_j}\|_{M, w'} = (1 \pm \varepsilon / 2) \|y_{P_j}\|_{M, w}.
$$
\item[Case (iii): $\|y_{H_y \cap P_j}\|_{M, w} \ge T$ and $\|y_{L_y \cap P_j}\|_{M, w}< T$. ] 
By Lemma \ref{lem:sample_heavy} and Lemma \ref{lem:sample_light}, with probability at least $1 - \deltaos / O(|\mathcal{N}| \cdot \log n)$, 
$$
\|y_{H_y \cap P_j}\|_{M, w'} = (1 \pm \varepsilon / 2) \|y_{H_y \cap P_j}\|_{M, w}
$$
and
$$
\left|\|y_{L_y \cap P_j}\|_{M, w'} - \|y_{L_y \cap P_j}\|_{M, w}\right| \le C \cdot U_M \cdot \tau^p \cdot 2^{j - 2}\cdot  \log(|\mathcal{N}| \cdot \log n / \deltaos) / \varepsilon. 
$$
Since
$$
\|y_{P_j}\|_{M, w} \ge \|y_{H_y \cap P_j}\|_{M, w} \ge T \ge C \cdot U_M  \cdot \tau^p \cdot 2^{j - 1} \cdot \log(|\mathcal{N}| \cdot \log n / \deltaos) / \varepsilon^2,
$$
we have
$$
\left|\|y_{L_y \cap P_j}\|_{M, w'} - \|y_{L_y \cap P_j}\|_{M, w}\right| \le \varepsilon / 2 \cdot \|y_{P_j}\|_{M, w},
$$
which implies
$$
\|y_{P_j}\|_{M, w'} = (1 \pm \varepsilon) \|y_{P_j}\|_{M, w}.
$$
\item[Case (iv): $\|y_{H_y \cap P_j}\|_{M, w} < T$ and $\|y_{L_y \cap P_j}\|_{M, w} \ge T$. ] 
By Lemma \ref{lem:sample_heavy} and Lemma \ref{lem:sample_light}, with probability at least $1 - \deltaos / O(|\mathcal{N}| \cdot \log n)$, 
$$
\|y_{L_y \cap P_j}\|_{M, w'} = (1 \pm \varepsilon / 2) \|y_{L_y \cap P_j}\|_{M, w}
$$
and
$$
\left|\|y_{H_y \cap P_j}\|_{M, w'} - \|y_{H_y \cap P_j}\|_{M, w}\right|
	\le C \cdot U_M \cdot \tau^p \cdot 2^{j - 2}\cdot  \log(|\mathcal{N}| \cdot \log n / \deltaos) / \varepsilon. 
$$
Since
$$
\|y_{P_j}\|_{M, w}
	\ge \|y_{L_y \cap P_j}\|_{M, w}
	\ge T \ge C \cdot U_M  \cdot \tau^p \cdot 2^{j - 1} \cdot \log(|\mathcal{N}| \cdot \log n / \deltaos) / \varepsilon^2,
$$
we have
$$
\left|\|y_{H_y \cap P_j}\|_{M, w'} - \|y_{H_y \cap P_j}\|_{M, w}\right| \le \varepsilon / 2 \cdot \|y_{P_j}\|_{M, w},
$$
which implies
$$
\|y_{P_j}\|_{M, w'} = (1 \pm \varepsilon) \|y_{P_j}\|_{M, w}.
$$
\end{description}
\end{proof}

Now we show that with probability $1 - \deltaos$, simultaneously for all $x \in \mathbb{R}^d$, $\|Ax\|_{p, w'}^p = (1 \pm \varepsilon) \|Ax\|_{p, w}^p$.
\begin{lemma}\label{lem:se_one}
For any $1 \le j \le 2 \log n$, with with probability at least $1 - \deltaos / (2 \log n)$, simultaneously for all $y = Ax$, 
$$
\|y_{P_j}\|_{p, w'}^p = (1 \pm \varepsilon) \|y_{P_j}\|_{p, w}^p.
$$

Applying a union bound over all $1 \le j \le 2 \log n$, this implies with probability at least $1 - \deltaos$, 
$$
\|y\|_{p, w'}^p = (1 \pm \varepsilon) \|y\|_{p, w}^p.
$$
\end{lemma}
\begin{proof}
For any fixed $1 \le j \le 2 \log n$, by Theorem \ref{thm:se},
if we take $\deltase = \deltaos / (2 \log n)$, with probability
at least $1 - \deltaos / (2 \log n)$, simultaneously for all $y = Ax$, we have
$$
\|y_{P_j \setminus I_j}\|_{p, w'}^p = (1 \pm \varepsilon)\|y_{P_j \setminus I_j}\|_{p, w}^p.
$$
Moreover, $\|y_{I_j}\|_{p, w}^p = \|y_{I_j}\|_{p, w'}^p$ since $w_i = w_i'$ for all $i \in I_j$.
Thus, we have $\|y_{P_j}\|_{p, w'}^p = (1 \pm \varepsilon) \|y_{P_j}\|_{p, w}^p$.
\end{proof}
\subsection{The Recursive Algorithm}\label{sec:recur_alg}
We start by setting $w = 1^n$.
In each recursive step, we use the sampling procedure defined in Section
\ref{sec:recursive_one_step} to obtain $w'$, by setting
$\deltaos = \delta / O(\log n)$ and $\varepsilon = \varepsilon' / O(\log n)$
for some $\varepsilon' > 0$. By Lemma \ref{lem:size_decrease}, for each recursive step, with probability
at least $1 - \delta / (10 \log n)$, we have $\|w'\|_{0} \le 2/3 \|w\|_0$. We
repeat the recursive step until $\|w\|_0 \le 10 F$.

By applying a union bound over all recursive steps, with probability $1 - \delta
/ 10$, the recursive depth is at most $\log_{3/2} n$. By Lemma
\ref{lem:weight_increase}, this also implies with probability $1 - \delta / 10$,
during the whole recursive algorithm, the weight vector $w$ always satisfies
$\|w\|_{\infty} \le 2^{\log_{1.5} n} \le n^2$. If we use $w_{\mathsf{final}}$ to
denote the final weight vector, then we have 
$$
\|w_{\mathsf{final}}\|_0 \le 10
F =  \widetilde{O}(d^{\max\{1, p / 2\}}  \log n  \cdot (  \log(|\mathcal{N}|/ \deltaos ) \cdot \log (1 / \deltaos)  + d) / \varepsilon^2).
$$
By Lemma \ref{lem:os_one}, and a union bound over all the $\log_{1.5} n$
recursive depths, with probability $1 - \delta$, simultaneously for all $y \in \mathcal{N}$, we have
$$
\|Ax\|_{M, w_{\mathsf{final}}} = (1 \pm O(\varepsilon \cdot \log n)) \|Ax\|_{M} = (1 \pm O(\varepsilon')) \|Ax\|_{M}.
$$

Moreover, by Lemma \ref{lem:se_one} and a union bound over all
the  $\log_{1.5} n$ recursive depths, with probability $1 - \delta / 10$, simultaneously for all $y = Ax$ we have
$$
\|Ax\|_{p, w_{\mathsf{final}}}^p = (1 \pm O(\varepsilon \cdot \log n)) \|Ax\|_{p, w}^p = (1 \pm O(\varepsilon')) \|Ax\|_{p, w}^p.
$$
We further show that conditioned on this event, simultaneously for all $x \in \mathbb{R}^d$, 
$$
\|Ax\|_{M, w_{\mathsf{final}}} \ge \frac{L_M}{U_M \cdot n} \cdot \|Ax\|_{M} .
$$
Consider a fixed vector $x \in \mathbb{R}^d$, if there exists a
coordinate $i \in H_{Ax}$ such that $w_i > 0$, since $w_i \ge 1$ if $w_i > 0$, we must have
$$
\|Ax\|_{M, w_{\mathsf{final}}}  \ge w_i M((Ax)_i) \ge M((Ax)_i) \ge L_M \cdot \tau^p.
$$
On the other hand, 
$$
\|Ax\|_M \le n \cdot U_M \cdot \tau^p, 
$$
which implies
$$
\|Ax\|_{M, w_{\mathsf{final}}}  \ge  \frac{L_M}{U_M \cdot n} \cdot \|Ax\|_{M} .
$$
Otherwise, $i \in L_{Ax}$ for all $i \in [n]$, which implies
$$
\|Ax\|_{M, w_{\mathsf{final}}}  \ge L_M \cdot \|Ax\|_{p, w_{\mathsf{final}}}^p \ge (1 - O(\varepsilon')) L_M \|Ax\|_{p, w}^p \ge \frac{(1 - O(\varepsilon')) L_M}{U_M} \|Ax\|_M.
$$

Finally, since each recursive step runs in
$\widetilde{O}((\nnz(A) + d^{p / 2 + O(1)}  \cdot \log( |\mathcal{N}|  / \delta)  \cdot / \varepsilon^2) \cdot \log(1 / \delta) )$ time,
and the number of recursive steps is upper bounded by $\log_{1.5}n$
with probability $1 - \delta / 10$, the total running time is
also upper bounded $\widetilde{O}((\nnz(A) + d^{p / 2 + O(1)}  \cdot \log( |\mathcal{N}|  / \delta)  \cdot / \varepsilon^2) \cdot \log(1 / \delta) )$
with probability $1 - \delta / 10$.

The following lemma can be proved by applying a union bound
over all observations above, 
changing $\varepsilon'$ to $\varepsilon$
and changing $A$ to $[A~b]$.
\begin{lemma}\label{lem:sample}
The algorithm outputs a vector $w_{\mathsf{final}} \in \R^n$, such that for any set $\mathcal{N} \subseteq \colspan([A~b])$ with size $|\mathcal{N}|$, with probability $1 - \delta$, the algorithm runs in $\widetilde{O}((\nnz(A) + d^{p / 2 + O(1)}  \cdot \log( |\mathcal{N}|  / \delta)  \cdot / \varepsilon^2) \cdot \log(1 / \delta) )$ time and the following holds:
\begin{enumerate}
\item $\|w_{\mathsf{final}}\|_0 \le \widetilde{O}(d^{\max\{1, p / 2\}}  \log^3 n  \cdot (  \log( |\mathcal{N}| / \delta ) \cdot \log (1 / \delta)  + d)  / \varepsilon^2 )$;
\item $\|w_{\mathsf{final}}\|_{\infty} \le n^2$;
\item For all $x \in \R^d$, $\|Ax - b\|_{M, w_{\mathsf{final}}} \ge \frac{L_M}{U_M \cdot n} \cdot \|Ax - b\|_{M}$.
\item For all $x \in \mathcal{N}$, $\|Ax- b\|_{M, w_{\mathsf{final}}} = (1 \pm \varepsilon) \|Ax - b\|_{M}$.
\end{enumerate}
\end{lemma}

Combining Lemma \ref{lem:sample} with the net argument in Theorem \ref{thm:net}, we have the following theorem.
\begin{theorem}\label{thm:sample}
By setting $|\mathcal{N}| = n^{O(d^3)} \cdot (1 / \varepsilon)^{O(d)}$, 
the algorithm outputs a vector $w_{\mathsf{final}} \in \R^n$, such that with probability $1 - \delta$, 
the algorithm runs in $\widetilde{O}((\nnz(A) + d^{p / 2 + O(1)} / \varepsilon^2  \cdot \log( 1  / \delta)) \cdot \log(1 / \delta) )$ time, 
$\|w_{\mathsf{final}}\|_0 \le \widetilde{O}(d^{p / 2 + O(1)}  \log ^4 n   \cdot \log^2(1 / \delta)/ \varepsilon^2)$
and any $C$-approximate solution of $\min_x \|Ax - b\|_{M, w_{\mathsf{final}}}$ with
$C \le \poly(n)$ is a $C \cdot (1 + \varepsilon)$-approximate solution of $\min_x \|Ax - b\|_M$. 
\end{theorem}
\begin{proof}
Lemma \ref{lem:sample} implies that $U_O = 1 + \varepsilon$,
$L_N = 1 - \varepsilon$, $L_A = \frac{L_M}{U_M \cdot n}$ and $U_A \le \|w_{\mathsf{final}}\|_{\infty} \le n^2$.
Adjusting constants and applying Theorem \ref{thm:net} imply the desired result. 
\end{proof}

% !TEX root = main.tex

\section{The $M$-sketch}\label{sec M sketch}
In this section we give an oblivious sketch for Tukey loss functions.
Throughout this section we assume $1 \le p \le 2$ in Assumption \ref{as M}.

For convenience and to set up notation, we first describe the construction.
\paragraph*{The sketch.}
Each coordinate $z_p$ of a vector $z$ to be sketched
is mapped to a \emph{level} $h_p$, and the number of coordinates
mapped to level $h$ is exponentially
small in $h$: for an integer branching factor $b>1$, we expect the number of coordinates
at level $h$ to be about a $b^{-h}$ fraction of the coordinates. The number of
buckets at a given level is $N=bcm$, where integers $m,c>1$ are parameters
to be determined later.

Our sketching matrix is $S\in\R^{N\hm\times n}$, where $\hm \equiv \floor{\log_b (n/m)}$.
Our weight vector $w\in\R^{N\hm}$ has entries $w_{i+1}\gets \beta b^h$, for $i\in [Nh,N(h+1))$ and integer $h= 0,1, \ldots, \hm$, and
$\beta\equiv (b-b^{-\hm})/(b-1)$.
Our sketch is reminiscent of sketches in the data stream literature, 
where we hash into buckets at multiple levels of subsampling 
\cite{iw05,VZ12}. However, the estimation performed in the sketch space needs to
be the same as in the original space, 
which necessitates a new analysis. 

The entries of $S$ are $S_{j,p} \gets \Lambda_p$, where $p\in [n]$ and $j \gets g_p + N h_p$ and
\begin{equation}\label{eq S}
\begin{split}
  \Lambda_p  & \gets \pm 1 \mathrm{\ with\ equal\ probability}
\\ g_p & \in [N]  \mathrm{\ chosen\ with\ equal\ probability}	
\\ h_p & \gets h \mathrm{\ with\ probability\ } 1/\beta b^{h} \mathrm{\ for\ integer\ }h\in [0,\hm],\end{split}
\end{equation}
all independently. Let $L_h$ be the multiset $\{z_p\mid h_p = h\}$, and $L_{h,i}$ the multiset
$\{z_p\mid h_p = h, g_p = i\}$; that is, $L_h$ is multiset of values at a given level,
$L_{h,i}$ is the multiset of values in a bucket.  We can write
$\norm{S z}_{M,w}$ as $\sum_{h\in [0,\hm], i\in [N]} \beta b^h M(\norm{L_{h,i}}_\LL)$,
where $\norm{L}_\Lambda$ denotes $ | \sum_{z_p\in L} \Lambda_p z_p|$.

\subsection{Accuracy Bounds for Sketching One Vector}\label{sec one vec}

We will show that our sketching construction has the property that for a given vector $z\in\R^n$,
with high
probability, $\norm{Sz}_{M,w}$ is not too much smaller than $\norm{z}_M$.
%After some preliminary notational setup, we give Lemma~\ref{lem Q< good}, showing good behavior of the sketch.
%Our analysis follows \cite{cw15} in structure,
%David: I think we are already saying this, worried if we repeat it too
%often a cursory reviewer will just think it's the same analysis, but
%it's not
%We consider the estimate $\norm{Sz}_{M,w}$ of $\norm{z}_M$, for $z\in\R^n$.
We assume that $\norm{z}_M=1$, for notational convenience.

Define $y\in\R^n$ by $y_p = M(z_p)$, so that $\norm{y}_1 = \norm{z}_M = 1$.
Let $Z$ denote the multiset comprising the coordinates of $z$, and
let $Y$ denote the multiset comprising the coordinates of $y$. 
For $\hat Z\subset Z$, let $M(\hat Z)\subset Y$
denote $\{M(z_p)\mid z_p\in \hat Z\}$.
Let $\norm{Y}_k$ denote $\left(\sum_{y\in Y} |y|^k\right)^{1/k}$, so $\norm{Y}_1 = \norm{y}_1$.
Hereafter multisets will just be called ``sets''.

\paragraph*{Weight classes.}
Fix a value $\gamma > 1$, and for integer $q\ge 1$,  let $W_q$ denote the multiset comprising
\emph{weight class} $\{y_p\in Y \mid \gamma^{-q}\le y_p \le \gamma^{1-q}\}$.
We have $\beta b^h \E[\norm{M(L_h)\cap W_q}_1]  = \norm{W_q}_1$.
For a set of integers $Q$, let $W_Q$ denote $\cup_{q\in Q} W_q$.

\paragraph*{Defining $\qm$ and $h(q)$.}
For given $\eps>0$,
consider $y'\in\R^n$ with $y'_i\gets y_i$ when $y_i > \eps/n$, and $y'_i\gets 0$ otherwise.
Then $\norm{y'}_1 \ge 1- n(\eps/n) = 1-\eps$. 
We can neglect $W_q$ for $q>\qm \equiv \log_\gamma (n/\eps)$,
up to error $\eps$. 
Moreover,
we can assume that
$\norm{W_q}_1\ge \eps/\qm$, since the contribution to $\norm{y}_1$ of weight classes $W_q$
of smaller total weight, added up for $q\le\qm$, is at most $\eps$.

Let $h(q)$ denote $ \floor{\log_b (|W_q|/\beta m)}$ for $|W_q|\ge \beta m$,
and zero otherwise, so that
\[m\le \E[|M(L_{h(q)})\cap W_q|] \le bm\]
for all $W_q$ except those with $|W_q| < \beta m$,
for which the lower bound does not hold.

Since $|W_q|\le n$ for all $q$, we have $h(q)\le \floor{\log_b (n/\beta m)} \le \hm$.

\subsection{Contraction Bounds}\label{subsec contract G}
Here we will show that $\norm{Sz}_{M,w}$ is not too much smaller than $\norm{z}_M$.
We will need some weak conditions among the parameters.
Recall that $N=bcm$.

\begin{assumption}\label{as params}
We will assume $b\ge m$, $b>c$,
$m=\Omega(\log\log (n/\eps))$, $\log b = \Omega(\log \log(n/\eps))$,
$\gamma\ge 2\ge \beta$, an error parameter $\eps\in [1/10, 1/3]$, and $\log N\le \eps^2m$.
We will consider $\gamma$ to be fixed throughout, that is, not dependent on the other parameters.
\end{assumption}

We need lemmas that allow lower bounds on the contributions of
the weight classes. First, some notation. For $h=0, 1, \ldots, \hm$, let
\begin{equation}\label{eq Q< defs}
\begin{split}
M_< & \equiv  \log_\gamma(m/\eps) = O(\log_\gamma(b / \eps)) \\
Q_< & \equiv \{q\mid |W_q| < \beta m, q\le M_<\} \\
\hat{Q}_h & \equiv \{q\mid h(q)=h, |W_q|\ge \beta m\} \\
M_\ge & \equiv \log_\gamma(2(1+3\eps) b/\eps) \\
Q_h & \equiv \{q\in \hat{Q}_h \mid q\le M_\ge + \min_{q\in \hat Q_h} q\} \\
Q^*  & \equiv Q_< \cup [\cup_h Q_h].
\end{split}
\end{equation}
Here $Q_<$ is the set of indices of weight classes that have relatively few members, but contain relatively large
weights.
$\hat Q_h$ gives the indices of $W_q$ that are ``large'' and have
$h$ as the level at which between $m$ and $bm$ members of $W_q$ are
expected in $L_h$. The set $Q_h$ cuts out the weight classes that can
be regarded as negligible at level $h$.

\begin{lemma}\label{lem Q< good}
If $N\ge \max\{ O(|M_<|  dm^3  \varepsilon), \widetilde{O}(d^2 m ^2 / \varepsilon^2)\}$, then
with constant probability, 
for all $z \in \colspan(A)$ and all $q \in Q_<$, the following event $\cE_v$ holds: there are sets
$W_q^*\subset W_q$, with $|W_q^*|\ge (1-\eps)|W_q|$, such that for all $y\in W_q^*$,
\begin{enumerate}
\item they are isolated: they are the sole members of $W_{Q_<}$ in their bucket;
\item their buckets are low-weight: the set $L$ of other entries in bucket containing $y\in W_q^*$ has $\norm{L}_1\le 1/\eps^2 m^3$.
\end{enumerate}
\end{lemma}
\begin{proof}
Without loss of generality we assume $h(q)$ are the same for all $q \in M_<$, since otherwise we can deal with each $h(q)$ separately. 

Let $\alpha = m / (L_M \cdot \varepsilon)$.
By Lemma \ref{cor:existence}, there exists a set $I \subseteq [n]$ with size $|I| = \widetilde{O}(d \cdot \alpha) = \widetilde{O}(d \cdot m / \varepsilon)$ such that for any $z \in \colspan(A)$, if $z$ satisfies 
	(i) $\|z_{L_z}\|_p^p \le \alpha \cdot \tau^p$ and 
	(ii) $|H_z| \le \alpha$, 
then $H_z \subseteq I$.
Let $\{u\}_{i \in [n] \setminus I}$ be the $\ell_p$ Lewis weights of $A_{[n] \setminus I, *}$
and let $J \subseteq [n] \setminus I$ be the set of indices of the $d \cdot m / \varepsilon \cdot U_M / L_M$ largest coordinates of $u$.
Thus, $|J| \le O(d \cdot m / \varepsilon)$.
Since $J$ contains the $d \cdot m / \varepsilon \cdot U_M / L_M$ largest coordinates of $u$ and $$ 
\sum_{i \in [n] \setminus I}  u_i =  \sum_{i \in [n] \setminus I} \overline{u}_i^p \le d
$$ by Theorem \ref{thm:fa_lewis}, for each $i \in [n] \setminus \left(I \cup J\right)$ , we have $u_i \le d / (d \cdot m / \varepsilon \cdot U_M / L_M) \le \varepsilon /  m \cdot L_M / U_M $.

If $\tau^p < \|z\|_M \cdot \varepsilon / m$, by Assumption \ref{as M}.\ref{as M inc}, we have $M(z_i) \le \tau^p < \|z\|_M \cdot \varepsilon / m$ for all $i \in [n]$. 
In this case, we have $W_{Q <} = \emptyset$.
Thus we assume $\tau^p \ge \|z\|_M \cdot \varepsilon / m$ in the remaining part of the analysis.

Since $\|z\|_M \ge |H_z| \cdot \tau^p$, we have $|H_z| \le m / \varepsilon$. 
Furthermore, by Assumption \ref{as M}.\ref{as M near quad}, $\|z_{L_z}\|_p^p \le \|z_{L_z}\|_M / L_M \le \|z\|_M  / L_M\le \tau^p \cdot m / (L_M \cdot \varepsilon)$.
Thus by setting $\alpha = m / (L_M \cdot \varepsilon)$ we have $H_z \subseteq I$.
For each $i \in [n] \setminus I$, we have $|z_i|\le \tau$.
By Lemma \ref{lem:bound_entry} and Assumption \ref{as M}.\ref{as M near quad}, for each $i \in [n] \setminus I$, $M(z_i) \le |z_i|^p / L_M \le u_i \cdot \|z_{[n] \setminus I}\|_p^p / L_M \le u_i \cdot \|z_{[n] \setminus I}\|_M \cdot U_M / L_M < u_i \cdot \|z\|_M \cdot U_M / L_M$.
Thus for each entry $i \in [n] \setminus \left(I \cup J\right)$, we have $M(z_i) < \varepsilon / m \cdot \|z\|_M$.

Thus, the indices of all members of $W_{Q_<}$ are in $I \cup J$. 
By setting $N \ge  |I \cup J|^2 / \kappa = \widetilde{O}(d^2 m^2 / \varepsilon^2) / \kappa$, 
the expected number of total collisions in $I \cup J$ is $|I \cup J|^2 / N \le \kappa$.
Thus, by Markov's inequality, with probability $1 - 2 \kappa$, the total number of collisions is upper bounded by $1 / 2$, i.e., there is no collision. 
This implies the first condition.

For the second condition, we use $\{u_i\}_{i \in [n] \setminus (I \cup J)}$ to denote the $\ell_p$ Lewis weights of $A_{i \in [n] \setminus (I \cup J), *}$. 
Consider a fixed $q \in M_<$.
By the first condition, all elements in $W_q$ are the sole members of $W_{Q_<}$ in their buckets. 
For each bucket we define $B_{h,i}$ to be the multiset $\{u_p\mid h_p = h, g_p = i, p \in [n] \setminus (I \cup J)\}$.
By setting $N \ge \frac{U_M \cdot |M_< | \cdot dm^3 \varepsilon}{L_M \cdot \kappa} $, for each $y \in W_q$, $\E[\|B_{h, i}\|_1] \le d / N \le  \frac{L_M}{U_M} \cdot \frac{1}{\eps^2 m^3} \cdot \frac{\varepsilon \cdot \kappa}{|M_<|}$ where $L_{h, i}$ is the bucket that contains $y$.
This is simply because $\sum_{i \in N} B_{h, i} \le \sum_{i \in [n] \setminus (I \cup J)} u_i \le d$ by Theorem \ref{thm:fa_lewis}.
We say a bucket is {\em good} if $\|B_{h, i}\|_1 \le \frac{L_M}{U_M} \cdot \frac{1}{\eps^2 m^3}$.
Notice that for $y \in W_q$, if $y$ is in a good bucket $B_{h, i}$, then the set $L$ of other entries in that bucket satisfies
\begin{align*}
&\|L\|_1=  \sum_{y \in L} y  \\
= &\sum_{p \in [n] \setminus (I \cup J) \mid h_p = h, g_p = i} M(z_p) \\
\le  &\sum_{p \in [n] \setminus (I \cup J) \mid h_p = h, g_p = i} U_M \cdot |z_p|^p \tag{Assumption \ref{as M}.\ref{as M near quad}}\\
\le  &\sum_{p \in [n] \setminus (I \cup J) \mid h_p = h, g_p = i} U_M \cdot u_p \cdot \|z_{[n] \setminus (I \cup J)}\|_p^p \tag{Lemma \ref{lem:bound_entry}}\\
\le  &\sum_{p \in [n] \setminus (I \cup J) \mid h_p = h, g_p = i} U_M / L_M \cdot u_p \cdot \|z_{[n] \setminus (I \cup J)}\|_M  \tag{Assumption \ref{as M}.\ref{as M near quad}}\\
\le  & \|B_{h, i}\|_1 \cdot U_M/L_M \cdot \|z\|_M \\
\le &  \frac{1}{\eps^2 m^3} \cdot \|z\|_M.
\end{align*}
Thus, it suffices to show that at least $(1 - \varepsilon)|W_q|$ buckets associated with $y \in W_q$ are good.

By Markov's inequality, for each $y \in W_q$, with probability $1 - \varepsilon \cdot \kappa / |M_<|$, the bucket that contains $y$ is good.
Thus, for the $|W_q|$ buckets associated with $y \in W_q$, the expected number of good buckets is at least $(1 - \varepsilon \cdot \kappa / M_<) |W_q|$.
Again, by Markov's inequality, with probability at least $1 - \kappa / |M_<|$, at least $(1 - \varepsilon)|W_q|$ buckets associated with $y \in W_q$ are good, and we just take these $(1 - \varepsilon)|W_q|$ good buckets to be $W_{q}^*$.
By applying a union bound over all $q \in M_<$, the second condition holds with probability at least $1 - \kappa$.
The lemma follows by applying a union bound over the two conditions and setting $\kappa$ to be a small constant. 

\end{proof}

%\begin{lemma}[Lemma 3.1 of \cite{cw15}]\label{lem W_q good} %Lemma 9 in non-SODA style
%For $\eps\le 1$,
%with failure probability at most \[
%4\qm \hm \exp(-\eps^2m/3)\le C^{-\eps^2 m}\]
% for a constant
%$C>1$,  the event $\cE$ holds, that for all $q\le \qm$ with
%$|W_q|\ge \beta m$,  and all $h\le h(q)$, that
%\[
% |M(L_h) \cap W_q| = \beta^{-1} b^{-h} |W_q|(1\pm\eps),
%\]
%and
%\[
% \norm{M(L_h) \cap W_q}_1 = \beta^{-1} b^{-h} \norm{W_q}_1(1\pm\eps).
%\]
%\end{lemma}
%
%Here $a=b(1\pm \eps)$ means that $|a-b| \le \eps |b|$.
%
%We will hereafter generally assume that $\cE$ holds.

%\begin{lemma}[Lemma 3.7 of \cite{cw15} ] \label{lem ignore small} %Lemma 15 in non-SODA style
%Using Assumption~\ref{as params}, $\sum_{q\in Q^*}  \norm{W_q}_1 \ge 1- 5\eps$.
%\end{lemma}

%\begin{lemma}[Lemma 3.2 of \cite{cw15} ]\label{lem W_q isolates}%Lemma 10 in non-SODA style
%For $h\in [\hm]$, suppose $Q\subset \{q\mid h(q)=h, |W_q|\ge \beta m\}$,
%and $\hW\subset Y$ contains $W_Q \equiv \cup_{q\in Q} W_q$.
%If $|M(L_h)\cap \hW| \le \eps N$, then with failure probability at most $2|Q|\exp(-\eps^2m/3)$,
%each $W_q$ has $W_q^*\subset M(L_h)\cap W_q$ with $|W_q^*|\ge (1-\eps)\beta^{-1} b^{-h}|W_q|$,
%and where each entry of $W_q^*$ is in a bucket with no other element
%of $\hW$. Also if condition $\cE$ of Lemma~\ref{lem W_q good} holds, then
%$$\norm{W_q^*}_1 \ge (1- 4\gamma\eps) \beta^{-1} b^{-h}\norm{W_q}_1.$$
%\end{lemma}

\def\prooflemWqisolates{
\begin{proof}
We will show that for $q\in Q$, with high probability it will hold that
$a_q\ge (1-\eps))\beta^{-1} b^{-h} |W_q|$, where $a_q$ is the number of
buckets $M(L_{h,i})$, over $i\in [N]$,
containing a member of $W_q$, and no other members of $\hW$.

Consider each $q\in Q$ in turn, and the members of $W_q$ in turn,
for $k=1,2,\ldots s\equiv |W_q|$, and let $Z_k$ denote the number of bins occupied by the
first $k$ members of $W_q$.
The probability that $Z_{k+1}>Z_k$ is at least
$\beta^{-1} b^{-h}(1 - |M(L_h)\cap \hW|/N) \ge \beta^{-1} b^{-h}(1-\eps) $.
We have $a_q\ge (1-\eps)\beta^{-1} b^{-h} |W_q|$
in expectation.

To show that this holds with high probability,
let $\hat Z_k \equiv \E[Z_s\mid Z_k]$. Then $\hat Z_1, \hat Z_2, \ldots$ is a Martingale
with increments bounded by 1, and with the second moment
of each increment at most $\beta^{-1} b^{-h}$.
Applying Freedman's inequality gives a concentration for $a_q$ similar to the above
application of Bernstein's inequality, yielding a failure probability $2\exp(-\eps^2 m/3)$,

Applying a union bound over all $|Q|$ yields that with
probability at least $1- 2|Q| \exp(-\eps^2 m/3)$, for each $W_q$ there is $W_q^*$
of size at least
$(1-\eps)\beta^{-1} b^{-h}|W_q|$ such that each member of $W_q^*$
is in a bucket containing no other member of
$\hW$.

For the last claim, we compare the at least $(1-\eps)X$ entries of $W_q^*$, 
where $X \equiv \beta^{-1} b^{-h}|W_q|$, with the at most $(1+\eps)X - |W_q^*|$ entries
of $M(L_h)\cap W_q$ not in $W_q^*$, using condition $\cE$; we have
\begin{align*}
\frac{\norm{W_q^*}_1}{\norm{M(L_h)\cap W_q}_1}
	   & \ge \frac{(1-\eps)X\gamma^{-q}}{(1-\eps)X\gamma^{-q} + 2\eps X \gamma^{1-q}}
	\\ & \ge 1 - 2\gamma\eps/(1-\eps).
\end{align*}
Using condition $\cE$ again to make the comparison with $\norm{W_q}_1$,
the claim follows.
\end{proof}
}
%\begin{proof}
%For convenience, we include a proof in \S\ref{prf lem W_q isolates} of the appendix.
%\end{proof}

%
%\begin{lemma}[Lemma~3.3 of \cite{cw15}]\label{lem bs}%Lemma 11 in non-SODA style
%For $h\in [\hm]$, $\bW\subset M(L_h)$, $T\ge \norm{\bW}_\infty$,
%and $\delta \in (0,1)$, if 
%\[
%N \ge \frac{6\norm{\bW}_1}{T \log(N/\delta)},
%\]
%then with failure probability $\delta$,
%\[
%\max_{i\in[N]} \norm{M(L_{h,i})\cap\bW}_1
%	\le \frac{7}{6}T\log(N/\delta).
%\]
%\end{lemma}
%

%\def\lemWqG{

\begin{lemma}[Lemma 3.8 of \cite{cw15}]\label{lem W_q^* G}%Lemma 16 in non-SODA style
%Assume that condition $\cE$  of Lemma~\ref{lem W_q good} holds.
Let $Q'_h\equiv \{q\mid q\le M'_h\}$, where $M'_h \equiv \log_\gamma(\beta b^{h+1}m^2\qm)$.
Then for large enough $N= O(m^2 b \eps^{-1}\qm)$,
with probability at least $1-C^{-\eps^2 m}$ for a constant $C>1$,
for each $q\in \cup_h Q_h$, there is $W_q^*\subset L_{h(q)}\cap W_q$ such that:
\begin{enumerate}
\item $|W_q^*| \ge (1-\eps) \beta^{-1}b^{-h(q)}|W_q|$.
\item each $x\in W_q^*$ is in a bucket with no other member of $W_{Q^*}$.
\item \label{it |W*|} $\norm{W_q^*}_1 \ge (1-4\gamma\eps)\beta^{-1}b^{-h}\norm{W_q}_1$.
\item \label{it Q'} each $x\in W_q^*$ is in a bucket with no member of $W_{Q'_h}$.
\end{enumerate}
\end{lemma}
%
%\begin{proof}
%For each $h$, apply Lemma~\ref{lem W_q isolates}
%to $Q_h$ and with $\hW \gets W_{Q^*} \equiv W_{Q_<} \cup_{q\in Q_h} W_q$,
%so that, using condition $\cE$,
%\begin{align*}
%|M(L_h)\cap \hW|
%	& \le M_<\beta m+ M_\ge(1+\eps)bm
%	\\ & = O(mb\log_\gamma(b/\eps)).
%\end{align*}
%To apply Lemma~\ref{lem W_q isolates}, we need 
%$N> \eps^{-1}|M(L_h)\cap \hW|$,
%and large enough $N$ in $O(mb\eps^{-1}\log_\gamma(b/\eps))$ suffices for this.
%We have (1) and (2),
%with failure probability $2M_\ge \exp(-\eps^2m)$.
%
%Condition (\ref{it |W*|}) follows from Lemma~\ref{lem W_q isolates}.
%
%For (\ref{it Q'}),
%let $\hW \gets W_{Q_h} \cup W_{Q'_h}$.
%Since $|W_{Q'_h}|\gamma^{-M'_h}\le \norm{y}_1\le 1$, so that
%$|W_{Q'_h}| \le \beta b^{h+1}m^2\qm$,
%we have 
%\begin{align*}
%|M(L_h)\cap \hW |
%	& \le |M(L_h)\cap W_{Q_h}| + |M(L_h) \cap W_{Q'_h}|
%	\\ & \le (1+\eps)bmM_\ge + (1+\eps) \beta^{-1}b^{-h}|W_{Q'_h}|
%	\\ & \le O(bm^2\qm),
%\end{align*}
%using condition $\cE$. Since $| M(L_h)\cap \hW| \le \eps N$
%for large enough $N= O(m^2 b \eps^{-1}\qm)$, we can apply Lemma~\ref{lem W_q isolates}
%to obtain (\ref{it Q'}).
%\end{proof}
%}

For $v\in T\subset Z$, let $T-v$ denote $T\setminus\{v\}$.
\begin{lemma}[Lemma~3.6 of \cite{cw15}]\label{lem G bucket}%Lemma 14 in non-SODA style
For $v\in T\subset Z$,
\[
M(\norm{T}_\Lambda) \ge \left(1-\frac{\norm{T-v}_\Lambda}{ |v|}\right)^2 M(v),
\]
and if $M(v)\ge \eps^{-1}\norm{T-v}_M$, then
\begin{equation}\label{eq T-v}
\frac{\norm{T-v}_2}{|v|}\le \eps^{1/\alphap},
\end{equation}
and for a constant $C$,
%\[
$\E_\Lambda[M(\norm{T}_\Lambda)]
	\ge (1-C\eps^{1/\alphap})M(v).$
%\]
\end{lemma}

\begin{lemma}[Lemma~3.9 of \cite{cw15}] \label{lem Q_h G}%Lemma 17 in non-SODA style
Assume 
%that condition $\cE$  of Lemma~\ref{lem W_q good} holds, and 
Assumption~\ref{as params}.
There is $N=O(\eps^{-2}m^2b\qm)$,
so that for all $0 \le h \le \hm$ and $q\in Q_h$ with $\norm{W_q}_1\ge \eps/\qm$,
we have
\[
\sum_{y_p\in W_q^*} M(\norm{L(y_p)}_\LL)
	\ge (1-\eps^{1/\alphap}) \norm{W_q}_1
\]
with failure probability at most $C^{-\eps^2m}$ for fixed $C>1$.
\end{lemma}

%This follows as for Lemma 3.10 of \cite{cw15}, but we include the proof for completeness
\def\prooflemQhG{
\begin{proof}
For any $q\in Q_h$ we have
\begin{align*}
|W_q|
	& \le (1+\eps) \beta b^h \E[|M(L_h)\cap W_q|]
	\\ & \le (1+\eps) \beta b^h bm
\end{align*}
by condition $\cE$ and the definition of $h(q)=h$; since
\[
|W_q|\gamma^{1-q} \ge \norm{W_q}_1\ge \eps/\qm,
\]
using $\norm{W_q}_1\ge \eps/\qm$ from the lemma statement, we have for any $y_p\in W_q$,
\begin{equation}\label{eq y big}
y_p\ge \gamma^{-q} \ge (\eps/\qm)/\gamma |W_q| \ge \eps/b^{h+1}\gamma \beta m (1+\eps)\qm.
\end{equation}

From condition \ref{it Q'} of Lemma~\ref{lem W_q^* G},
we have that no bucket containing $y_p\in W_q^*$ contains an
entry larger than $ \gamma/\beta b^{h+1}m^2\qm$, so if
$\bW$ comprises $M(L_h)\cap (Y\setminus W_{Q'_h}) $,
we have $\norm{\bW}_\infty \le  \gamma/\beta b^{h+1}m^2\qm$.
Using condition $\cE$,  $\norm{\bW}_1\le (1+\eps)b^{-h}$,
using just the condition $\norm{Y}_1 = 1$.
Therefore the given $N$ is larger than the $O(bm\eps^{-2}\qm)$
needed for Lemma~\ref{lem bs} to apply,
with $\delta = \exp(-\eps^2m)$.
This with \eqref{eq y big} yields
that for each $y_p \in W_q^*$, the remaining entries
in its bucket $L$ have $\norm{L -y_p}_1\le 2\gamma^2 \eps |y_p|$,
with failure probability $\exp(-\eps^2m)$.

For each such isolated $y_p$ we consider the corresponding $z_p$
(denoted by $v$ hereafter),
and let $L(v)$ denote the set of $z$ values in the bucket containing $v$.
We apply Lemma~\ref{lem G bucket} to $v$ with $L(v)$ taking the role
of $T$, and $2\gamma^2\eps$ taking the role of $\eps$, obtaining
$\E_\Lambda[M(\norm{L(v)}_\Lambda)] \ge  (1-C'\eps^{1/\alphap})M(v)$.
(Here we fold a factor of $(2\gamma^2)^{1/\alphap}$ into $C'$, recalling that
we consider $\gamma$ to be fixed.)
Using this relation and condition $\cE$, we have
\begin{align*}
\norm{W_q}_1
	   & \le \beta b^h \norm{W_q^*}_1/(1-4\gamma\eps) \qquad \mathrm{from\ Lem~\ref{lem W_q^* G}.\ref{it |W*|}}
	\\ & \le \beta b^h\sum_{M(v)\in W_q^*}
			\frac{\E_\Lambda[M(\norm{L(v)}_\Lambda)]}{(1-4\gamma\eps)(1-C'\eps^{1/\alphap})},
\end{align*}
so the claim of the lemma follows, in expectation, after adjusting constants,
and conditioned on events of failure probability $C^{-\eps^2 m}$ for constant $C$.

To show the tail estimate, we relate each $M(\norm{L(v)}_\LL)$
to $M(v)$ via the first claim of Lemma~\ref{lem G bucket}, which implies
$M(\norm{L(v)}_\LL) \ge (1-2\norm{L(v)-v}_\LL/|v|)M(v)$. Writing 
$V\equiv M^{-1}(W_q^*)$, we have
\begin{align*}
	\sum_{v\in V} & M(\norm{L(v)}_\LL) 
	  \\ &  \ge 
	  \vssum{v\in V\\ \norm{L(v)}_\LL > |v|}   M(v)
	  \quad +\quad  \vssum{v\in V\\  \norm{L(v)}_\LL\le |v|}
	  	\left( 1-2\frac{\norm{L(v) -v}_\LL}{|v|}\right) M(v)
 	\\ & \ge \norm{W_q^*}_1 -  2 \vssum{v\in V\\  \norm{L(v)}_\LL\le |v|}
					\,\frac{\norm{L(v) -v}_\LL}{|v|}\gamma^{1-q}.
\end{align*}

It remains to upper bound the sum. 
Since $\norm{L(v)}_\Lambda = | |v|\pm t |$,
where $t\equiv \norm{L(v) -v}_\Lambda$,
if $\norm{L(v)}_\LL \le |v|$, then $t\le 2|v|$.

Since 
\[
\E[t\mid t\le 2|v|] \le \E[t] \le C_1\norm{L(v)-v}_2 \le C_1 C' \eps^{1/\alphap}|v|,
\]
using Khintchine's inequality and \eqref{eq T-v}, and similarly
$\E[t^2 | t\le 2|v| ] \le C_2(C' \eps^{1/\alphap})^2v^2$, we can use Bernstein's inequality
to bound
\begin{align*}
\vssum{v\in V\\  \norm{L(v)}_\LL\le |v|}
					\frac{\norm{L(v) -v}_\LL}{|v|}\gamma^{1-q} 
	   & \le\vssum{v\in V\\  \norm{L(v)}_\LL\le |v|}
					C_3\eps^{1/\alphap} \gamma^{1-q}
	\\ & \le C_4 \eps^{1/\alphap} \norm{W_q^*}_1,
\end{align*}
with failure probability $\exp(-\eps^2 m)$.
Hence
\begin{align*}
\sum_{v\in V} & M(\norm{L(v)}_\LL)
	\\ & \ge \norm{W_q^*}_1 -  2C_4 \eps^{1/\alphap} \norm{W_q^*}_1
	\\ & = \norm{W_q^*}_1(1-2C_4\eps^{1/\alphap})
	\\ & \ge  \beta^{-1}b^{-h}\norm{W_q}_1 (1-4\gamma\eps)(1-2C_4\eps^{1/\alphap}),
\end{align*}
using condition $\cE$. Adjusting constants, the result follows.
\end{proof}
} % def prooflemQhG
%\begin{proof}
%This follows as for Lemma 3.10 of \cite{cw15}, but we include the proof for completeness
%in \S\ref{prf lem Q_h G} of the Appendix.
%\end{proof}

\begin{lemma}\label{lem Q_< G}
Assume that $\cE_v$ of Lemma~\ref{lem Q< good} holds, and Assumption~\ref{as params}.
Then for $q\in Q_<$, 
\[
\sum_{y_p\in W_q^*} M(\norm{L(y_p)}_\LL)  \ge (1 - \eps^{1 / \alphap})\norm{W_q}_1
\]
with failure probability at most $C^{-\eps^2 m}$
for a constant $C>1$.
\end{lemma}

\begin{proof}
%For $y_p\in W^*_q$, let $L$ denote the bucket $L_{h,i}$ containing $y_p$.
Let $v\equiv z_p$ where $y_p = M(z_p)$, 
let $L(v)$ denote the $\{ z_{p'} \mid M(z_{p'}) \in L\}$.
Condition $\cE_v$ and  $M(v)\ge \eps/m$ imply that
\[
\norm{L(v)-v}_2^2 \le \norm{L}_1 \le 1/\eps^2 m^3 < M(v)/\eps m,
\]
so that using \eqref{eq T-v} we have
\begin{equation}\label{eq bucket norm}
\frac{\norm{L(v)-v}_2^2}{|v|^2}
	\le \frac{\norm{L(v) - v}_M}{M(v)}
	\le \frac{1}{\eps m}.
\end{equation}
Since $\norm{L}_\infty \le \norm{L}_1$, we also have,
for all $v'\in L(v) - v$, and using again $M(v)\ge \eps/m$,
\begin{equation}\label{eq bucket max}
\left|\frac{v'}{v}\right|
	\le \left( \frac{M(v')}{M(v)}\right)^{1/2}
	\le \frac{1}{m\eps^{3/2}}.
\end{equation}

From \eqref{eq bucket max}, we have that the summands determining
$\norm{L(v)-v}_\LL$ have magnitude at most $|v|\eps^{1 / \alphap}/\eps^2 m$.
From \eqref{eq bucket norm}, we have $\norm{L(v)-v}_2^2$
is at most $v^2\eps /\eps^2 m$. It follows from
Bernstein's inequality that with failure probability $\exp(-\eps^2m)$,
$\norm{L(v)-v}_\LL\le \eps^{1 / \alphap} |v|$. Applying the first claim
of Lemma~\ref{lem G bucket}, we have
$M(\norm{L(v)}_\LL)\ge (1-2\eps^{1 / \alphap}) M(v)$, for all $v\in M^{-1}(W_q^*)$
with failure probability $\beta m M_< \exp(-\eps^2 m)$.
Summing over $W_q^*$, we have
\[
\sum_{v\in M^{-1}(W_q*)} M(\norm{L(v)}_\LL) \ge (1 - \eps^{1 / \alphap})\norm{W_q^*}_1 \ge  (1 - 2\eps\gamma)(1 - \eps^{1 / \alphap})\norm{W_q}_1.
\]
This implies the bound,
using Assumption~\ref{as params},
after adjusting constants.
\end{proof}
%} %def prooflemQltG
%\begin{proof}
%Please see \S\ref{prf lem Q_< G} of the Appendix.
%\end{proof}

%That is, with high probability, the members of $W_{Q_<}$ will be ``well represented'' in the loss function evaluation.
The above lemmas imply that overall,
with high probability, the sketching-based estimate of $\|z\|_M$ of a single given vector $z$ is very likely
to not much smaller than $\|z\|_M$, as stated next.

\begin{theorem}[Theorem~3.2 of \cite{cw15}]\label{thm contract G}
Assume 
%condition $\cE$ of Lemma~\ref{lem W_q good} holds, and 
Assumption~\ref{as params},
and condition $\cE_v$ of Lemma~\ref{lem Q< good}.
Then
$\norm{Sz}_{M,w} \ge \norm{z}_M(1-\eps^{1 / \alphap})$,
with failure probability no more than $C^{-\eps^2 m}$, for an absolute constant $C>1$.

\end{theorem}

%The proof follows much as for Theorem 3.2 of \cite{cw15}, where condition $\cE_v$ takes the role of $\cE_c$ there;
%for convenience we include a proof
\def\proofthmcontractG{
\begin{proof}
(We note that $c$, $b$, and $m$ can be chosen such that the relations among these
quantities and also $N=cbm$ satisfy Assumption~\ref{as params}, up to
the weak relations among $m$, $b$, and $n/\eps$, which ultimately will require that $n$
is not extremely large relative to $d$.)

Recalling $Q^*$ from \eqref{eq Q* defs},
let $Q^{**} \equiv \{q\mid q\in Q^*, \norm{W_q}_1\ge \eps/\qm\}$.
Assuming conditions $\cE$ and $\cE_c$, we have, with probability $1-C^{-\eps^2m}$,
\begin{align*}
\frac1\kappa \norm{Sz}_{M,w}
	   & = \sum_{h,i} \beta b^hM(\norm{L_{ h,i}}_\LL) & \mathrm{Def.}
	\\ & \ge \sum_{q\in Q^{**}, v\in W_q^*} \beta b^{h(q)} M(\norm{L(v)}_\LL) & \mathrm{Lem~\ref{lem W_q^* G}}
	\\ &\ge \sum_{q\in Q^{**}} \beta b^{h(q)} (1-\eps^{1 / \alphap}) \norm{W_q^*}_1 & \mathrm{Lems~\ref{lem Q_h G},\ref{lem Q_< G}}
	\\ &\ge \sum_{q\in Q^{**}} (1-\eps^{1 / \alphap}) (1-4\gamma\eps) \norm{W_q}_1& \mathrm{Lem~\ref{lem W_q^* G}}.
\end{align*}
Using Lemma~\ref{lem ignore small},
\[
\sum_{q\in Q^{**}} \kappa \norm{W_q}_1
	\ge -\kappa qm(\eps/\qm) + \kappa \sum_{q\in Q^*} \norm{W_q}_1
	\ge 1-6\eps.
\]
Adjusting constants gives the result.
\end{proof}
}
%\begin{proof}
%The proof follows much as for Theorem 3.2 of \cite{cw15}, where condition $\cE_v$ takes the role of $\cE_c$ there;
%for convenience we include a proof in \S\ref{prf thm contract G} of the Appendix.
%\end{proof}

\subsection{A ``Clipped'' Version}
For a vector $z$, we use $\|Sz\|_{Mc, w}$ to denote a ``clipped'' version of $\|Sz\|_{M, w}$, in which we ignore small buckets and use a subset of the coordinates of $Sz$ as follows: $\|Sz\|_{Mc, w}$ is obtained by adding in only those buckets in level $h$ that are among the top 
$$
M^* \equiv bmM_{\ge} + \beta m M_{<}
$$
in $\|L_{h, i}\|_{\LL}$, recalling $M_{\ge}$ and $M_<$ defined in \eqref{eq Q< defs}.
Formally, we define $\|Sz\|_{Mc, w}$ to be
$$
\|Sz\|_{Mc, w} = \sum_{h\in [0,\hm], i\in [M^*]} \beta b^h M(\norm{L_{h,(i)}}_{\LL}),
$$
where $L_{h, (i)}$ denotes the level $h$ bucket with the $i$-th largest $\|L_{h, i}\|_{\LL}$ among all the level $h$ buckets.

The proof of the contraction bound of $\|Sz\|_{M, w}$ in Theorem \ref{thm contract G} requires only lower bounds on $M(\|L_{h, i}\|_{\LL})$ for those at most $M^*$ buckets on level $h$.
Thus, the proven contraction bounds continue to hold for $\|Sz\|_{Mc, w}$, and in particular $\|Sz\|_{Mc, w} \ge (1 - \varepsilon)\|Sz\|_{M, w}$.
\subsection{Dilation Bounds}\label{ssec dilation}

%We note that to bound the dilation, we use a standard argument in this area.
We use two prior bounds of \cite{cw15} on dilation; the first shows that the dilation is at most $O(\log n)$ 
in expectation, while the second shows that the ``clipped'' version gives $O(1)$ dilation with constant probability. 
%The variant reduction requires using a slightly different norm for evaluating residuals in the sketched version of the problems.
Note that we need only expectations, since we need the dilation bound to hold only for the optimal solution as in Theorem \ref{thm:net}.

\begin{theorem}[Theorem 3.3 of \cite{cw15}] \label{thm dilate G}%Theorem 22 in non-SODA style
$\E[\norm{Sz}_{M,w}] = O(\hm)\norm{z}_M$.
\end{theorem}

Better dilation is achieved by using the ``clipped'' version $\|Sz\|_{Mc, w}$, as described in \cite{cw15}.
%This implies a different norm, denoted $\norm{Sz}_{Mc, w}$ for estimating $\norm{z}_M$.

\begin{theorem}[Theorem 3.4 of \cite{cw15}] \label{thm dilation clipped}%Lemma 23 in non-SODA style
There is $c = O(\log_\gamma(b/\eps) (\log_b(n/m)))$ and $b\ge c$,
recalling $N=mbc$,
such that
\[
\E[\norm{Sz}_{Mc,w}] \le C \norm{z}_M
\]
for a constant $C$.
\end{theorem}

\subsection{Regression Theorem}\label{sec G regress proof}

\begin{lemma}\label{lem msketch for all}
There is $N = O(d^2 \hm)$, so that with constant probability, simultaneously for all $x \in \R^d$, 
$$
0.9 / (n \cdot U_M / L_M)\|Ax-b\|_M \le \|S(Ax - b)\|_{M, w} \le U_M / L_M \cdot n^2 \cdot  \|Ax-b\|_M.
$$
\end{lemma}
\begin{proof}
For the upper bound,
\[
\norm{S z}_{M, w} = \sum_{h\in [0,\hm], i\in [N]} \beta b^h M(\norm{L_{h,i}}_\LL).
\]
The weights $\beta b^h$ are less than $n$, and
\begin{align*}
& M(\norm{L_{h,i}}_\LL)\\
\le & M(\norm{L_{h,i}}_1)\\
 \le&  M(n^{1 - 1 / p}\norm{L_{h,i}}_p) \tag{Assumption \ref{as M}.\ref{as M inc}} \\
\le & U_M \cdot n^{p - 1}  \norm{L_{h,i}}_p ^p \tag{Assumption \ref{as M}.\ref{as M near quad}} \\ 
\le & U_M  / L_M \cdot n \cdot \sum_{z_p \in L_{h, i}} M(z_p). \tag{Assumption \ref{as M}.\ref{as M near quad}}
\end{align*}
Since any given $z_p$ contributes once to $\norm{S z}_{M,w}$,
$\norm{S z}_{M,w}\le U_M / L_M \cdot n^2 \cdot \norm{z}_M$. 

For the lower bound, notice that
\[
\norm{S z}_{2, w}^2 = \sum_{h\in [0,\hm], i\in [N]} \beta b^h \norm{L_{h,i}}_\LL^2.
\]
For each $h \in [0, \hm]$, since $N = O(d^2 \hm)$, with probability at least $1 - 1 / (10\hm)$, simultaneously for all $z \in \colspan(A)$ we have
$$
\sum_{i \in [N]}  \norm{L_{h,i}}_\LL^2 = (1 \pm 0.1) \sum_{z_p \in L_h} z_p^2,
$$
since the summation on the left-hand side can be equivalently viewed as applying \textsf{CountSketch} \cite{CW13, NN13, MengMahoney}  on $L_h$.
Thus, by applying union bound over all $h \in [0, \hm]$, we have
\begin{equation}\label{equ:l2ose}
\norm{S z}_{2, w}^2 = \sum_{h\in [0,\hm], i\in [N]} \beta b^h \norm{L_{h,i}}_\LL^2 \ge 0.9 \|z\|_2^2.
\end{equation}

If there exists some $i \in H_{Sz}$, since $w_i \ge 1$ for all $i$, we have 
$$
\|Sz\|_{M, w} \ge w_i M((Sz)_i) \ge M((Sz)_i) \ge \tau^p.
$$
On the other hand,
$$
\|z\|_{M} \le n \cdot U_M \cdot \tau^p,
$$
which implies
$$
\|Sz\|_{M, w} \ge \|z\|_M / (n \cdot U_M).
$$

If $H_{Sz} = \emptyset$, then 
\begin{align*}
& \|Sz\|_{M, w}\\
 \ge & \sum_{i} w_i |(Sz)_i|^p \cdot L_M \tag{Assumption \ref{as M}.\ref{as M near quad}} \\
 = & \|Sz\|_{p, w}^p \cdot L_M\\
 \ge & \|Sz\|_{2, w}^p \cdot L_M \tag{$p \le 2$}\\ 
 \ge & 0.9 \|z\|_{2}^p \cdot L_M \tag{\eqref{equ:l2ose}}\\
 \ge & 0.9 \|z\|_{p}^p \cdot L_M / n\\
 \ge & 0.9 \|z\|_M / (n \cdot U_M / L_M).\tag{Assumption \ref{as M}.\ref{as M near quad}} \\
 \end{align*}

\end{proof}
The following theorem states that $M$-sketches can be used for Tukey regression, under the conditions described above.

\begin{theorem}\label{thm G regress}
Under Assumption~\ref{as M} and Assumption~\ref{as main},
there is an algorithm running in $O(\nnz(A))$ time,
that with constant probability
creates a sketched regression problem $\min_x \norm{S(Ax-b)}_{M,w}$
where $SA$ and $Sb$ have $\poly(d\log n)$ rows, and any $C$-approximate solution $\tilde{x}$ of $\min_x \|S(Ax - b)\|_{M, w}$ with $C \le \poly(n)$ satisfies
\[
\norm{A\tilde{x} - b}_M\le O(C \cdot \log_d n)\min_{x\in\R^d}\norm{Ax-b}_M.
\]
Moreover, any $C$-approximate solution $\hat{x}$ of $\min_x \|S(Ax - b)\|_{Mc, w}$ with $C \le \poly(n)$ satisfies
\[
\norm{A\hat{x} - b}_M\le O(C)\min_{x\in\R^d}\norm{Ax-b}_M.
\]

\end{theorem}

\begin{proof}
We set $S$ to be an $M$-sketch matrix
with large enough
$N= \poly(d\log n)$.
We note that, up to the trivial scaling by $\beta$, $SA$ satisfies
Assumption~\ref{as main} if $A$ does.
We also set $m= O(d^3\log n)$, and $\eps=1/10$.
We apply Theorem \ref{thm:net} to prove the desired result.

The given $N$ is large enough for Theorem~\ref{thm contract G} and Lemma \ref{lem msketch for all} to apply,
obtaining a contraction bound with failure probability $C_1^{-m}$.
By Theorem~\ref{thm contract G}, since the needed contraction bound
holds for all members of $\mathcal{N}_{\poly(\varepsilon \cdot \tau / n)} \cup \mathcal{M}_{\poly(\varepsilon / n)}^{c, c \cdot \poly(n)}$, with
failure probability $n^{O(d^3)} C_1^{-m} < 1$, for $m=O(d^3\log n)$, assuming
the condition $\cE_v$.

Thus, by Theorem~\ref{thm dilate G}, we have $U_O \le O(\log_dn)$.
By Lemma \ref{lem msketch for all}, $L_A = 0.9 / (n \cdot U_M / L_M)$ and $U_A = U_M / L_M \cdot n^2 $.
By Theorem~\ref{thm contract G}, $L_N = 1 - \varepsilon^{1 / 2} = \Omega(1)$.
Thus, by Theorem \ref{thm:net} we have
$$
\norm{A\tilde{x} - b}_M\le O(C \cdot \log_d n)\min_{x\in\R^d}\norm{Ax-b}_M.
$$

A similar argument holds for $C$-approximate solution $\hat{x}$ of $\min_x \|S(Ax - b)\|_{Mc, w}$.
\end{proof}
%}

% !TEX root = arxiv.tex
\newcommand{\tsat}{$3$-$\mathsf{SAT}$}

\section{Hardness Results and Provable Algorithms for Tukey Regression}
\subsection{Hardness Results}\label{sec:hardness}
In this section, we prove hardness results for Tukey regression based on the {\em Exponential Time Hypothesis} \cite{impagliazzo2001complexity}.
We first state the hypothesis.

\begin{conjecture}[Exponential Time Hypothesis \cite{impagliazzo2001complexity}]\label{conj:eth}
For some constant $\delta > 0$, no algorithm can solve \tsat~on $n$ variables and $m = O(n)$ clauses correctly with probability at least $2 / 3$ in $O(2^{\delta n})$ time.
\end{conjecture}
Using Dinur's PCP Theorem~\cite{dinur2007pcp}, Hypothesis~\ref{conj:eth} implies a hardness result for $\mathsf{MAX}$-$\mathsf{3SAT}$.
\begin{theorem}[\cite{dinur2007pcp}]\label{thm:pcp}
Under Hypothesis~\ref{conj:eth}, for some constant $\varepsilon > 0$ and $c > 0$, no algorithm can, given a \tsat~formula on $n$ variables and $m = O(n)$ clauses, distinguish between the following cases correctly with probability at least $2 / 3$ in $2^{n / \log ^c n }$ time:
\begin{itemize}
\item There is an assignment that satisfies all clauses in $\phi$;
\item Any assignment can satisfy at most $(1 - \varepsilon)m$ clauses in $\phi$.
\end{itemize}
\end{theorem}

We make the following assumptions on the loss function $M: \R \rightarrow \R^+$.
Notice that the following assumptions are more general than those in Assumption \ref{as M}.
\begin{assumption}\label{assump:loss_hardness}
There exist real numbers $\tau \ge 0$ and $C > 0$ such that 
\begin{enumerate}
\item $M(x) = C$ for all $|x| \ge \tau$.
\item $0 \le M(x) \le C$ for all $|x| \le \tau$.
\item $M(0) = 0$.
\end{enumerate}
\end{assumption}

Now we give an reduction that transforms a \tsat~formula $\phi$ with $d$
variables and $m = O(d)$ clauses to a Tukey regression instance 
$$
\min_{x} \|Ax - b\|_M,
$$
such that $A \in \mathbb{R}^{n \times d}$ and $b \in \mathbb{R}^{n}$ with $n =
O(d)$, and all entries in $A$ are in $\{0, +1, -1\}$ and all entries in $b$ are in
$\{\pm k \tau \mid k \in \mathbb{N}, k \le O(1)\}$. Furthermore, there are at
most three non-zero entries in each row of $A$.

For each variable $v_i$ in the formula $\phi$, there is a variable $x_i$ in the
Tukey regression that corresponds to $v_i$. For each variable $v_i$, if $v_i$
appears in $\Gamma_i$ clauses in $\phi$, we add $2 \Gamma_i$ rows into $[A~b]$.
These $2 \Gamma_i$ rows are chosen such that when calculating $\|Ax-b\|_M$,
there are $\Gamma_i$ terms of the form $M(x_i)$, and another $\Gamma_i$ terms of
the form $M(x_i - 10\tau)$. This can be done by taking the $i$-th entry of the
corresponding row of $A$ to be $1$ and taking the corresponding entry of $b$ to
be either $0$ or $10\tau$. Since $\sum_{i=1}^d \Gamma_i = 3m$ in a \tsat~formula
$\phi$, we have added $6m = O(d)$ rows into $[A~b]$. We call these rows Part I
of $[A~b]$.

 Now for each clause $\mathcal{C} \in \phi$, we add three rows into $[A~b]$.
 Suppose the three variables in $\mathcal{C}$ are $v_i$, $v_j$ and $v_k$. The
 first row is chosen such that when calculating $\|Ax-b\|_M$, there is a term of
 the form $M(a + b + c - 10\tau)$, where $a = x_i$ if there is a positive
 literal that corresponds to $v_i$ in $\mathcal{C}$ and $a = 10\tau - x_i$ if
 there is a negative literal that corresponds to $v_i$ in $\mathcal{C}$.
Similarly, $b = x_j$ if there is a positive literal that corresponds to $v_j$ in
$\mathcal{C}$ and $b = 10\tau - x_j$ if there is a negative literal that
corresponds to $v_j$ in $\mathcal{C}$. The same holds for $c$, $x_k$, and $v_k$.
 The second and the third row are designed such that when calculating
 $\|Ax-b\|_M$, there is a term of the form $M(a + b + c - 20\tau)$ and another
 term of the form $M(a + b + c - 30\tau)$.
Clearly, this can also be done while satisfying the constraint that all entries
in $A$ are in $\{0, +1, -1\}$ and all entries in $b$ are in $\{\pm k \tau \mid k
\in \mathbb{N}, k \le O(1)\}$. We have added $3m$ rows into $[A~b]$. We call
these rows Part II of $[A~b]$.

This finishes our construction, with $6m+ 3m = O(d)$ rows in total. It also
satisfies all the restrictions mentioned above.

Now we show that when $\phi$ is satisfiable, if we are given any solution
$\overline{x}$ such that
$$
\|A\overline{x}-b\|_M \le (1 + \eta)\min_x\| Ax - b\|_M,
$$
then we can find an assignment to $\phi$ that satisfies at least $(1 - 5\eta) m$ clauses.

We first show that when $\phi$ is satisfiable, the regression instance we constructed satisfies
$$
\min_x\|Ax-b\|_M \le 5C \cdot m.
$$
We show this by explicitly constructing a vector $x$. For each variable $v_i$ in
$\phi$, if $v_i = 0$ in the satisfiable assignment, then we set $x_i$ to be $0$.
Otherwise, we set $x_i$ to be $10 \tau$. For each variable $v_i$, since $x_i \in
\{0, 10\tau\}$, for all the $2\Gamma_i$ rows added for it, there will be
$\Gamma_i$ rows contributing $0$ when calculating $\|Ax-b\|_M$, and another
$\Gamma_i$ rows contributing $C$ when calculating $\|Ax-b\|_M$. The total
contribution from this part will be $3C \cdot m$. For each clause
$\mathcal{C} \in \phi$, for the three rows added for it, there will be one row contributing
$0$ when calculating $\|Ax-b\|_M$, and another two rows contributing $C$ when
calculating $\|Ax-b\|_M$. This is by construction of $[A~b]$ and by the fact
that $\mathcal{C}$ is satisfied. Notice that $M(a + b + c - 10\tau) = 0$ if only one
literal in $\mathcal{C}$ is satisfied, $M(a + b + c - 20\tau) = 0$ if two
literals are satisfied, and $M(a + b + c - 30\tau) = 0$ if all three literals in
$\mathcal{C}$ are satisfied. Thus, we must have
$\min_x\|Ax-b\|_M \le 5C \cdot m$, which implies
$\|A \overline{x} - b\|_M \le (1 + \eta) 5C \cdot m$.

We first show that we can assume each $\overline{x}_i$ satisfies
$\overline{x}_i \in [-\tau, \tau]$ or $\overline{x}_i \in [9\tau, 11\tau]$. This is because we
can set $\overline{x}_i = 0$ otherwise without increasing
$\|A \overline{x} - b\|_M$, as we will show immediately. For any $\overline{x}_i$ that is not in the
two ranges mentioned above, its contribution to $\|A\overline{x}-b\|_M$ in Part
I is at least $C \cdot 2\Gamma_i$. However, by setting $\overline{x}_i = 0$, its
contribution to $\|A\overline{x}-b\|_M$ in Part I will be at most $C \cdot
\Gamma_i$. Thus, by setting $\overline{x}_i = 0$ the total contribution to
$\|A\overline{x}-b\|_M$ in Part I has been decreased by at least $C \cdot
\Gamma_i$. Now we consider Part II of the rows in $[A~b]$. The contribution to
$\|A\overline{x}-b\|_M$ of all rows in $[A~b]$ created for clauses that do not
contain $v_i$ will not be affected after changing $\overline{x}_i$ to be $0$.
For the $3\Gamma_i$ rows in $[A~b]$ created for clauses that contain $v_i$,
their contribution to $\|A\overline{x}-b\|_M$ is lower bounded by $C \cdot
2\Gamma_i$ and upper bounded by $C \cdot 3\Gamma_i$. The lower bound follows
since for any three real numbers $a$, $b$ and $c$, at least two elements in
$\{a+b+c - 10 \tau, a+b+c - 20 \tau,a+b+c - 30 \tau\}$ have absolute value at
least $\tau$, and $M(x) = C$ for all $|x| \ge \tau$. Thus, by setting
$\overline{x}_i = 0$ the total contribution to $\|A\overline{x}-b\|_M$ in Part
II will be increased by at most $C \cdot \Gamma_i$, which implies we can set
$\overline{x}_i = 0$ without increasing $\|A \overline{x} - b\|_M$.

Now we show how to construct an assignment to the \tsat~formula $\phi$ which
satisfies at least $(1 - 5\eta)m$ clauses, using a vector
$\overline{x} \in \mathbb{R}^d$ which satisfies
(i) $\|A\overline{x} - b\|_M \le (1 + \eta) 5C \cdot m$ and
(ii) $\overline{x}_i \in [-\tau, \tau]$ or $\overline{x}_i \in
[9\tau, 11\tau]$ for all $\overline{x}_i$. We set $v_i = 0$ if
$\overline{x}_i \in [-\tau, \tau]$ and set $v_i = 1$ if $\overline{x}_i \in [9\tau, 11\tau]$.
To count the number of clauses satisfied by the assignment, we show that for each
clause $\mathcal{C} \in \phi$, $\mathcal{C}$ is satisfied whenever
$a + b + c \ge 7\tau$. Recall that $a = x_i$ if there is a positive literal that
corresponds to $v_i$ in $\mathcal{C}$ and $a = 10\tau - x_i$ if there is a
negative literal that corresponds to $v_i$ in $\mathcal{C}$. Similarly,
$b = x_j$ if there is a positive literal that corresponds to $v_j$ in $\mathcal{C}$
and $b = 10\tau - x_j$ if there is a negative literal that corresponds to $v_j$
in $\mathcal{C}$. The same holds for $c$, $x_k$, and $v_k$. Since $a$, $b$ and
$c$ are all in the range $[-\tau, \tau]$ or in the range $[9\tau, 11\tau]$,
whenever $a + b + c \ge 7\tau$, we must have $a \ge 9\tau$, $b \ge 9\tau$ or
$c \ge 9\tau$, in which case clause $\mathcal{C}$ will be satisfied. Thus, at least
$(1 - 5\eta) m$ clauses will be satisfied, since otherwise
$\|A\overline{x} - b\|_M$ will be larger than
$3C \cdot m + 2C \cdot m+ 5\eta C \cdot m = (1 + \eta)5C \cdot m$.
Here the first term $3C \cdot m$ corresponds to the
contribution from Part I, since any $\overline{x}_i$ must satisfy
$|\overline{x}_i| \ge \tau$ or $|\overline{x}_i - 10\tau| \ge \tau$. The second
and the third term $2C \cdot m + 5\eta C \cdot m$ corresponds to the
contribution from Part II when at least $5\eta m$ clauses are not satisfied.

Our reduction implies the following theorem.
\begin{theorem}\label{thm:hardness}
Suppose there is an algorithm that runs in $T(d)$ time and succeeds with
probability $2 / 3$ for Tukey regression with approximation ratio $1 + \eta$
when the loss function $M$ satisfies Assumption \ref{assump:loss_hardness} and
the input data satisfies the following restrictions:
\begin{enumerate}
\item $A \in \mathbb{R}^{n \times d}$ and $b \in \mathbb{R}^{n}$ with $n = O(d)$.
\item All entries in $A$ are in $\{0, +1, -1\}$ and all entries in $b$ are in $\{\pm k \tau \mid k \in \mathbb{N}, k \le O(1)\}$.
\item There are at most three non-zero entries in each row of $A$.
\end{enumerate}
Then, there exists an algorithm that runs in $T(d)$ time for a \tsat~formula on
$d$ variables and $m = O(d)$ clauses which distinguishes between the following
cases correctly with probability at least $2 / 3$:
\begin{itemize}
\item There is an assignment that satisfies all clauses in $\phi$.
\item Any assignment can satisfy at most $(1 - 5\eta)m$ clauses in $\phi$.
\end{itemize}
\end{theorem}
Combining Theorem~\ref{thm:pcp} and Theorem~\ref{thm:hardness} with the Hypothesis~\ref{conj:eth}, we have the following corollary.
\begin{corollary}
Under Hypothesis~\ref{conj:eth}, for some constant $\eta > 0$ and $C> 0$,
no algorithm can solve Tukey regression with approximation ratio $1 + \eta$ and
success probability $2/ 3$, and runs in $2^{d / \log^C d}$ time, when the loss
function $M$ satisfies Assumption \ref{assump:loss_hardness} and the input data
satisfies the following restrictions:
\begin{enumerate}
\item $A \in \mathbb{R}^{n \times d}$ and $b \in \mathbb{R}^{n}$ with $n = O(d)$.
\item All entries in $A$ are in $\{0, +1, -1\}$ and all entries in $b$ are in $\{\pm k \tau \mid k \in \mathbb{N}, k \le O(1)\}$.
\item There are at most three non-zero entries in each row of $A$.
\end{enumerate}
\end{corollary}
\subsection{Provable Algorithms}\label{sec:provable}
In this section, we use the polynomial system verifier to develop provable algorithms for Tukey regression.

\begin{theorem}[\cite{renegar1992computational,basu1996combinatorial}]\label{thm:poly_sys}
Given a real polynomial system $P(x_1, x_2, \cdots, x_d)$ with $d$ variables and $n$ polynomial constraints $\{f_i (x_1, x_2, \cdots, x_d) \Delta_i 0\}_{i = 1}^n$, where $\Delta_i$ is any of the ``standard relations'': $\{ >, \geq, =, \neq, \leq, < \}$, let $D$ denote the maximum degree of all the polynomial constraints and let $H$ denote the maximum bitsize of the coefficients of all the polynomial constraints. Then there exists an algorithm that runs in
\begin{equation*}
(Dn)^{O(d)} \poly(H)
\end{equation*}
time that can determine if there exists a solution to the polynomial system $P$.
\end{theorem}

Besides Assumption \ref{as M}, we further assume that the loss function $M(x)$ can be approximated by a polynomial $P(x)$ with degree $D$, when $|x| \le \tau$.
Formally, we assume there exist two constants $L_P \le 1 \le U_P$ such that when $|x| \le \tau$, we have
$$
L_P P(|x|) \le M(|x|) \le U_P P(|x|).
$$
Indeed, Assumption \ref{as M} already implies we can take $P(x) = x^p$, with $L_P = L_M$ and $U_P = U_M$ when $p$ is an integer. 
However, for some loss function (e.g., the one defined in \eqref{eq:tukey_loss}), one can find a better polynomial to approximate the loss function.
Since the approximation ratio of our algorithm depends on $U_P / L_P$, for those loss functions we can get an algorithm with better approximation ratio. 
We also assume Assumption \ref{as main} and all entries in $A$ and $b$ are integers. 

We first show that under Assumption \ref{as main} and the assumption that all entries in $A$ and $b$ are integers, either $\|Ax - b\|_M = 0$ for some $x \in \mathbb{R}^d$, or $\|Ax - b\|_M \ge 1 / 2^{\poly(nd)}$ for all $x \in \mathbb{R}^d$.

\begin{lemma}\label{lem:reg_lb}
Suppose all entries in $A$ and $b$ are integers, under Assumption \ref{as M} and Assumption \ref{as main}, either $\|Ax - b\|_M = 0$ for some $x \in \mathbb{R}^d$, or $\|Ax - b\|_M \ge 1 / 2^{\poly(nd)}$ for all $x \in \mathbb{R}^d$.
\end{lemma}
\begin{proof}
We show that either there exists $x \in \mathbb{R}^d$ such that $Ax = b$, or $\|Ax - b\|_2 \ge 1 / 2^{\poly(nd)}$ for all $x \in \mathbb{R}^d$.
Notice that $\|Ax - b\|_2 \ge 1 / 2^{\poly(nd)}$ implies $\|Ax - b\|_{\infty} \ge 1 / 2^{\poly(nd)} / \sqrt{n}$, and thus the claimed bound follows from Assumption \ref{as M}.

Without loss of generality we assume $A$ is non-singular. 
By the normal equation, we know $x^* = (A^T A)^{-1} (A^T b)$ is an optimal solution to $\min_x\|Ax - b\|_2$.
By Cramer's rule, all entries in $x^*$ are either $0$ or have absolute value at least $1 / 2^{\poly(nd)}$.
This directly implies either $Ax^* - b = 0$ or $\|Ax^* - b\|_2 \ge1 / 2^{\poly(nd)}$.
\end{proof}

Lemma \ref{lem:reg_lb} implies that either $\|Ax -b\|_M = 0$ for some $x \in \mathbb{R}^d$, or $\|Ax - b\|_M \ge  1 / 2^{\poly(nd)}$ for all $x \in \mathbb{R}^d$.
The former case can be solved by simply solving the linear system $Ax = b$.
Thus we assume $\|Ax - b\|_M \ge  1 / 2^{\poly(nd)}$ for all $x \in \mathbb{R}^d$ in the rest part of this section. 

To solve the Tukey regression problem $\min_x \|Ax - b\|_M$, we apply a binary search to find the optimal solution value $\mathrm{OPT}$. 
Since $1 / 2^{\poly(nd)} \le \mathrm{OPT} \le n \cdot \tau^p \le 2^{\poly(nd)}$ by Assumption \ref{as M} and Assumption \ref{as main}, the binary search makes at most $\log(2^{\poly(nd)} / \varepsilon) = \poly(nd) + \log(1 / \varepsilon)$ guesses to the value of $\mathrm{OPT}$ to find a $(1 + \varepsilon)$-approximate solution. 

For each guess $\lambda$, we need to decide whether there exists $x \in \mathbb{R}^d$ such that $\|Ax - b\|_M \le \lambda$ or not.
We use the polynomial system verifier in Theorem \ref{thm:poly_sys} to solve this problem. 
We first enumerate a set of coordinates $S \subseteq [n]$, which are the coordinates with $|(Ax^* - b)_i| \ge \tau$, 
where $x^* = \argmin_x \|Ax - b\|_M$, 
and then solve the following decision problem:
\begin{align*}
 & ~\sum_{i \in [n] \setminus S} P(\sigma_i (Ax - b)_i) + |S| \cdot \tau^p \le \lambda\\
\text{s.t}& ~\sigma_{i}^2 =1, \forall {i \in [n] \setminus S}  \\
& ~0\le  \sigma_i (Ax - b)_i \le \tau, \forall {i \in [n] \setminus S}.
\end{align*}

Clearly, $\sigma_i (Ax - b)_i = |(Ax - b)_i|$, and thus $L_P P(\sigma_i (Ax - b)_i) \le M((Ax - b)_i) \le U_P P(\sigma_i (Ax - b)_i)$. 
Thus by Assumption \ref{as M}, for all $x \in \mathbb{R}^d$ and $S \subseteq [n]$, 
$$L_P \|Ax - b\|_M \le \sum_{i \in [n] \setminus S} P(\sigma_i (Ax - b)_i) + |S| \cdot \tau^p.$$
Moreover, 
$$\sum_{i \in [n] \setminus S} P(\sigma_i (Ax^* - b)_i) + |S| \cdot \tau^p \le U_P \|Ax^* - b\|_M$$ 
when $S = \{i \in [n] \mid |(Ax^* - b)_i| \ge \tau\}$, which implies the binary search will return a $((1 + \varepsilon) \cdot U_P / L_P)$-approximate solution. 

Now we analyze the running time of the algorithm. 
We make at most $ \poly(nd) + \log(1 / \varepsilon)$ guesses to the value of $\mathrm{OPT}$.
For each guess, we enumerate a set of coordinates $S$, which takes $O(2^n)$ time. 
For each set $S \subseteq [n]$, we need to solve the decision problem mentioned above, which has $n + d$ variables and $O(n)$ polynomial constraints with degree at most $D$.
By Theorem \ref{thm:poly_sys} this decision problem can be solved in $(nD)^{O(n)}$ time.  
Thus, the overall time complexity is upper bounded by $(nD)^{O(n)} \cdot \log(1 / \varepsilon)$.

Notice that we can apply the row sampling algorithm in Theorem \ref{thm:sample} to reduce the size of the problem before applying this algorithm.
This reduces the running time from $(nD)^{O(n)}  \cdot \log(1 / \varepsilon)= 2^{O(n \cdot (\log n + \log D))} \cdot \log(1 / \varepsilon)$ to $2^{\widetilde{O}(\log D \cdot d^{p / 2} \poly(d \log n) / \varepsilon^2)}$.
Formally, we have the following theorem.
\begin{theorem}
Under Assumption \ref{as M} and \ref{as main}, and suppose all entries in $A$ and $b$ are integers, and there exists a polynomial $P(x)$ with degree $D$ and two constants $L_P \le 1 \le U_M$ such that when $|x| \le \tau$, we have
$$
L_P P(|x|) \le M(|x|) \le U_P P(|x|).
$$
Then there exists an algorithm that returns a $((1 + \varepsilon) \cdot U_P / L_P)$-approximate solution to $\min_x \|Ax - b\|_M$ and runs in $2^{\widetilde{O}(\log D \cdot d^{p / 2} \poly(d \log n) / \varepsilon^2)}$ time. 
\end{theorem}
\begin{corollary}
Under Assumption \ref{as main}, and suppose all entries in $A$ and $b$ are integers, for the loss function $M$ defined in \eqref{eq:tukey_loss} there exists an algorithm that returns a $(1 + \varepsilon)$-approximate solution to $\min_x \|Ax - b\|_M$ and runs in $2^{\widetilde{O}(\poly(d \log n) / \varepsilon^2)}$ time. 
\end{corollary}

% !TEX root = arxiv.tex

\section{Experiments}\label{sec exper}
In this section we provide experimental results to illustrate the practicality of our dimensionality reduction methods. 
Figure~\ref{fig:exp} shows the approximation ratio of our dimensionality reduction methods, when applied to synthetic and real datasets. 
For all datasets, the number of data points is $n = 10000$.
The dimension $d$ is different for different datasets and is marked in Figure~\ref{fig:exp}.
We adopt the loss function defined in \eqref{eq:tukey_loss} and use different values of $\tau$ for different datasets.
To calculate the approximation ratio of our dimensionality reduction methods, 
we solve the full problems and their sketched counterparts by using the {\sf LinvPy} software \cite{linvpy}.
This software uses iteratively re-weighted least squares (IRLS), and we modify it for $\norm{\cdot}_{M,w}$,
which requires only to include a ``fixed'' weighting from $w$ into the IRLS solver.

The Random Gaussian dataset is a synthetic dataset, whose entries are sampled i.i.d. from the standard Gaussian distribution. 
The remaining datasets are chosen from the UCI Machine Learning Repository.
The $\tau$ values were chosen roughly so that there would be significant clipping of the residuals.
For each dataset, we also randomly select $5\%$ of the entries of the $b$ vector and change them to $10^4$, to model outliers.
Such modified datasets are marked as ``with outliers'' in Figure~\ref{fig:exp}.

We tested both the row sampling algorithm and the oblivious sketch. 
We varied the size of the sketch from $2d$ to $10d$ ($d$ is the dimension of the dataset) and calculated the approximation ratio
$
\|A\hat{x} - b\|_M / \|Ax^* - b\|_M
$
using the modified {\sf LinvPy} software, where $x^*$ is the solution returned by solving the full problem and $\hat{x}$ is the solution returned by solving the sketched version. 
We repeated each experiment ten times and took the best result among all repetitions. 

\begin{figure}
	\subfigure[Random Gaussian. $d = 20, \tau = 10$.]
	{
		\includegraphics[width=0.233\textwidth]{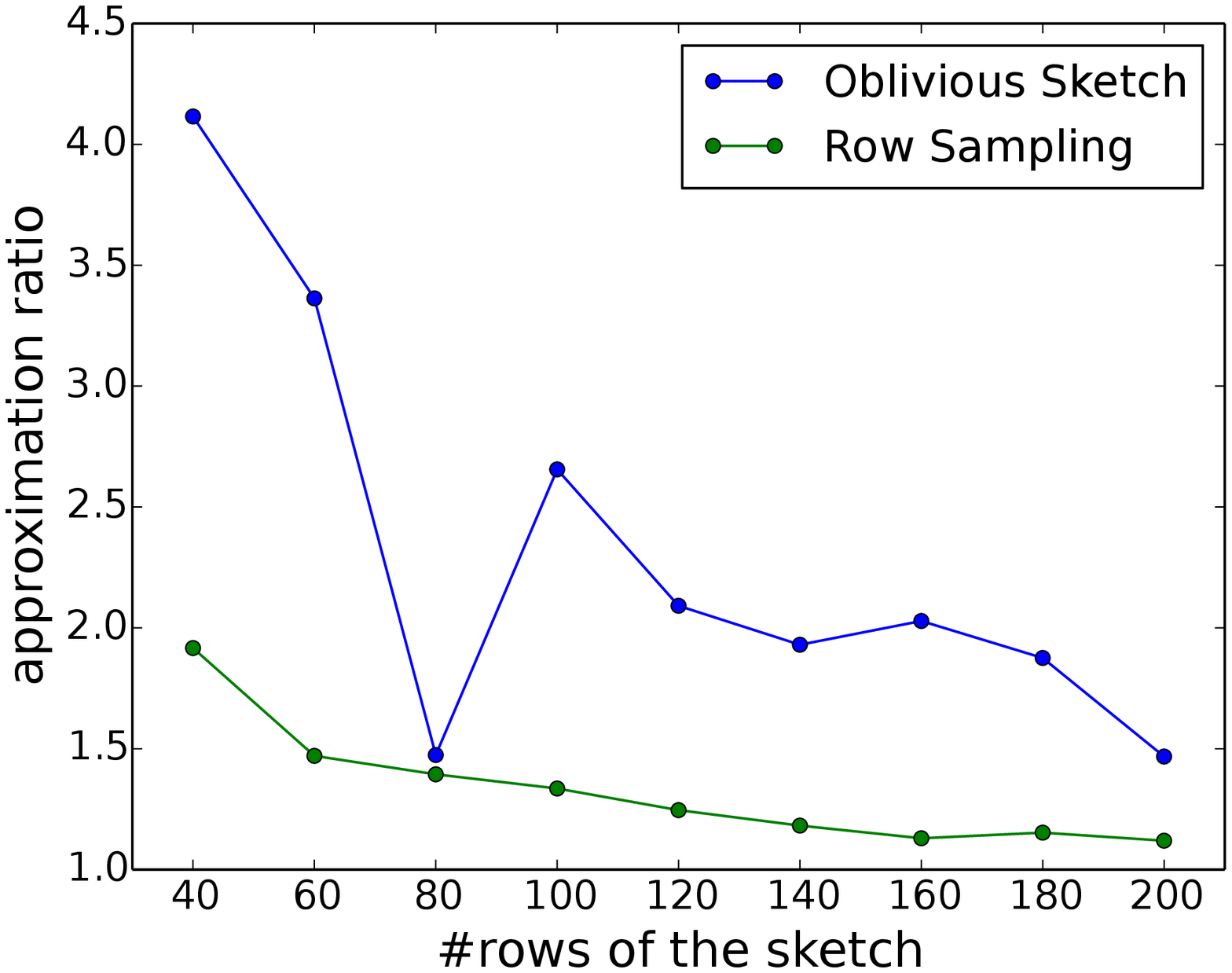}}
	\subfigure[Random Gaussian (with outliers). $d = 20$, $\tau = 10$.]
	{
		\includegraphics[width=0.233\textwidth]{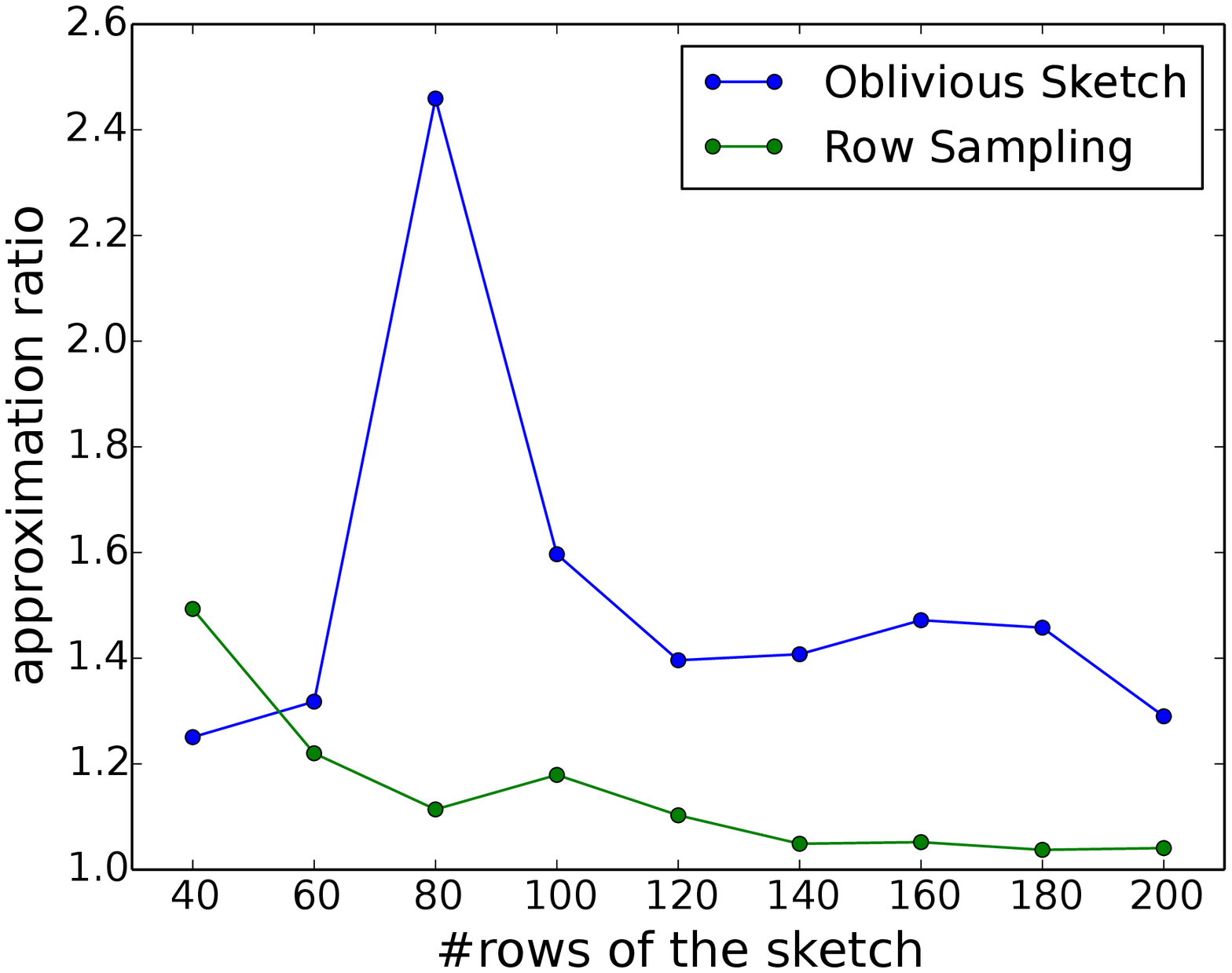}}  
	\subfigure[Facebook Comment Volume. $d = 53$, $\tau = 10$.]
	{
		\includegraphics[width=0.233\textwidth]{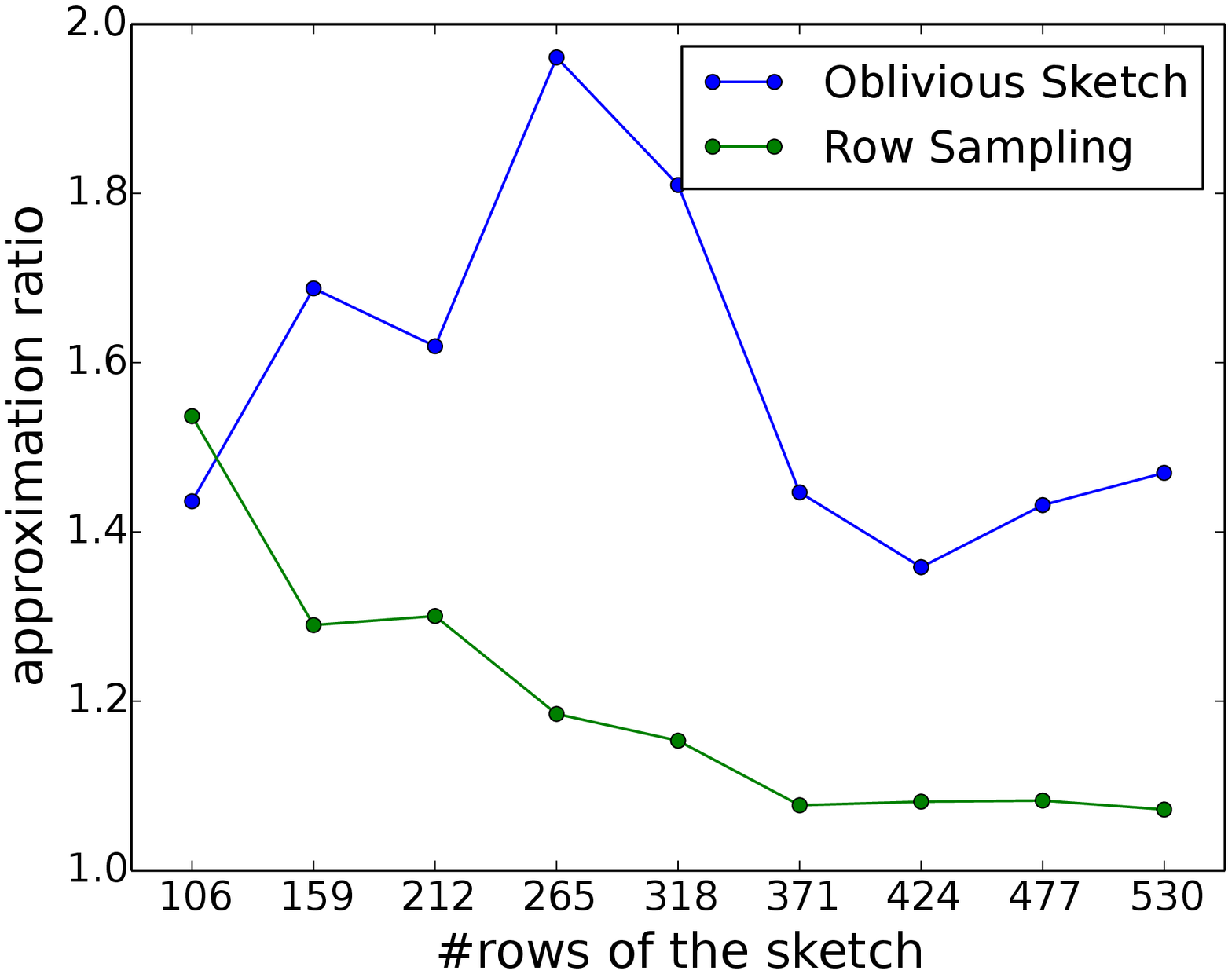}}  
	\subfigure[Facebook Comment Volume (with outliers). $d = 53$, $\tau = 100$.]
	{
		\includegraphics[width=0.233\textwidth]{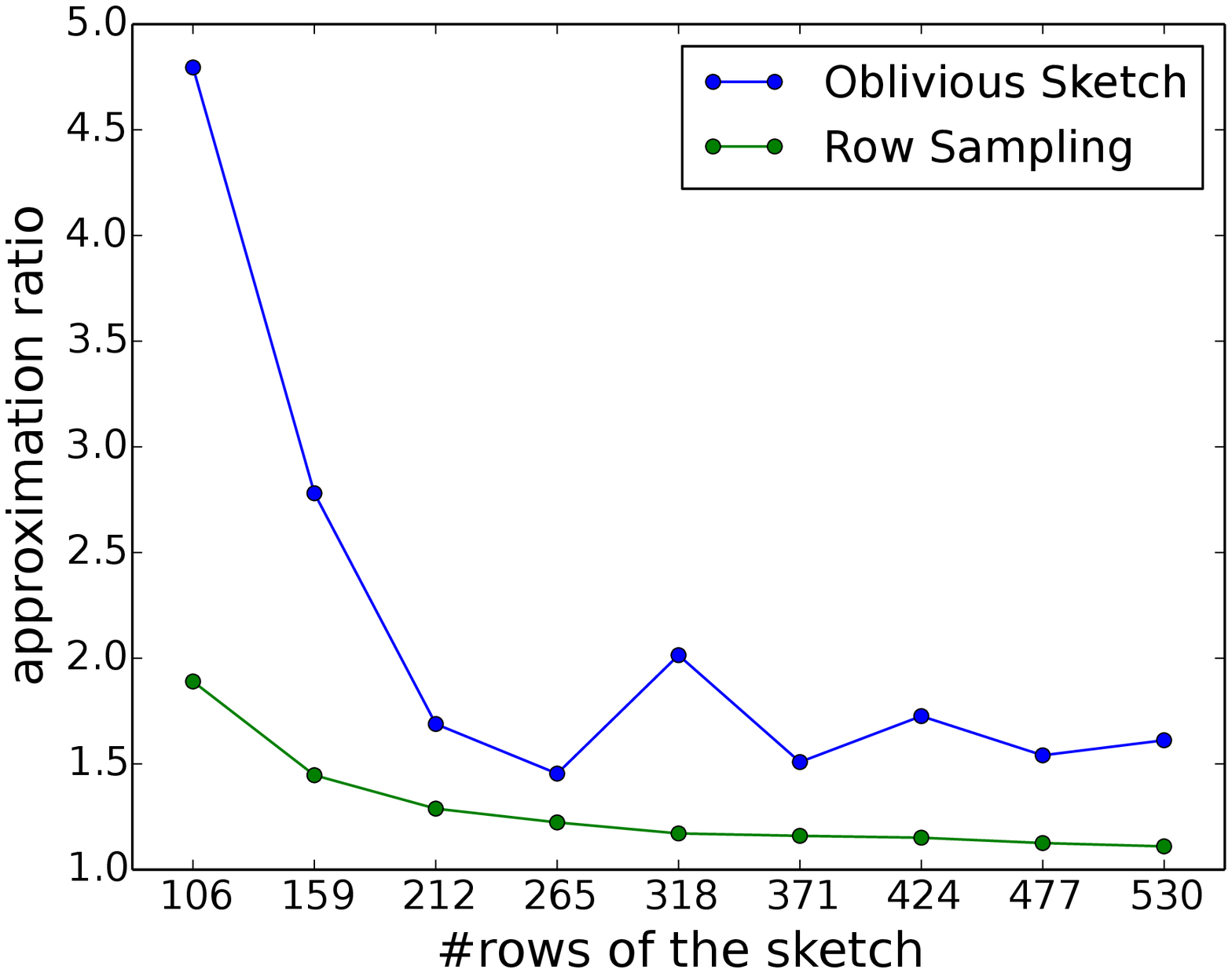}}   
	\subfigure[Appliances Energy Prediction. $d = 25$, $\tau = 1000$.]
	{
		\includegraphics[width=0.233\textwidth]{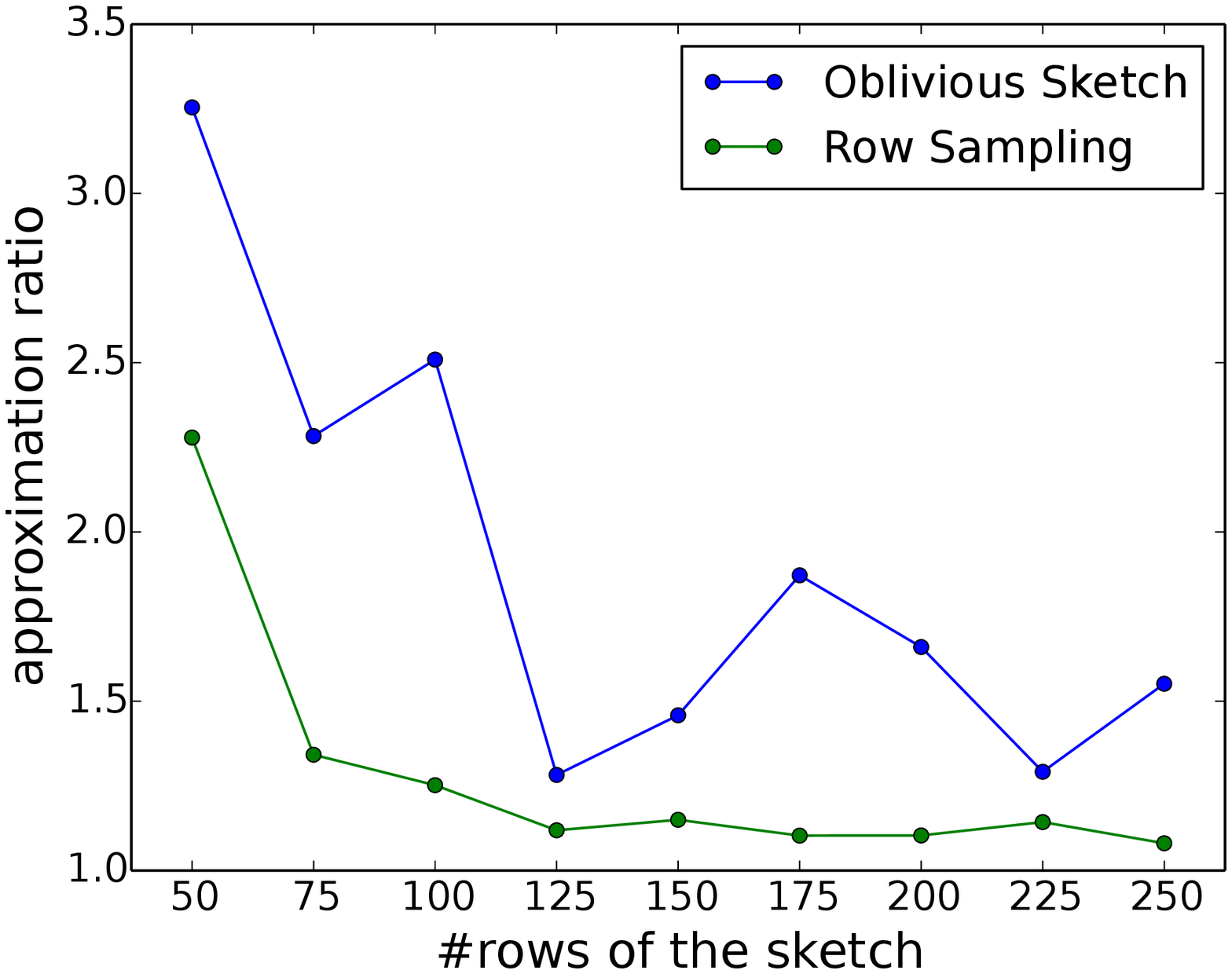}} 
	\subfigure[Appliances Energy Prediction (with outliers). $d = 25$, $\tau = 100$. ]
	{
		\includegraphics[width=0.233\textwidth]{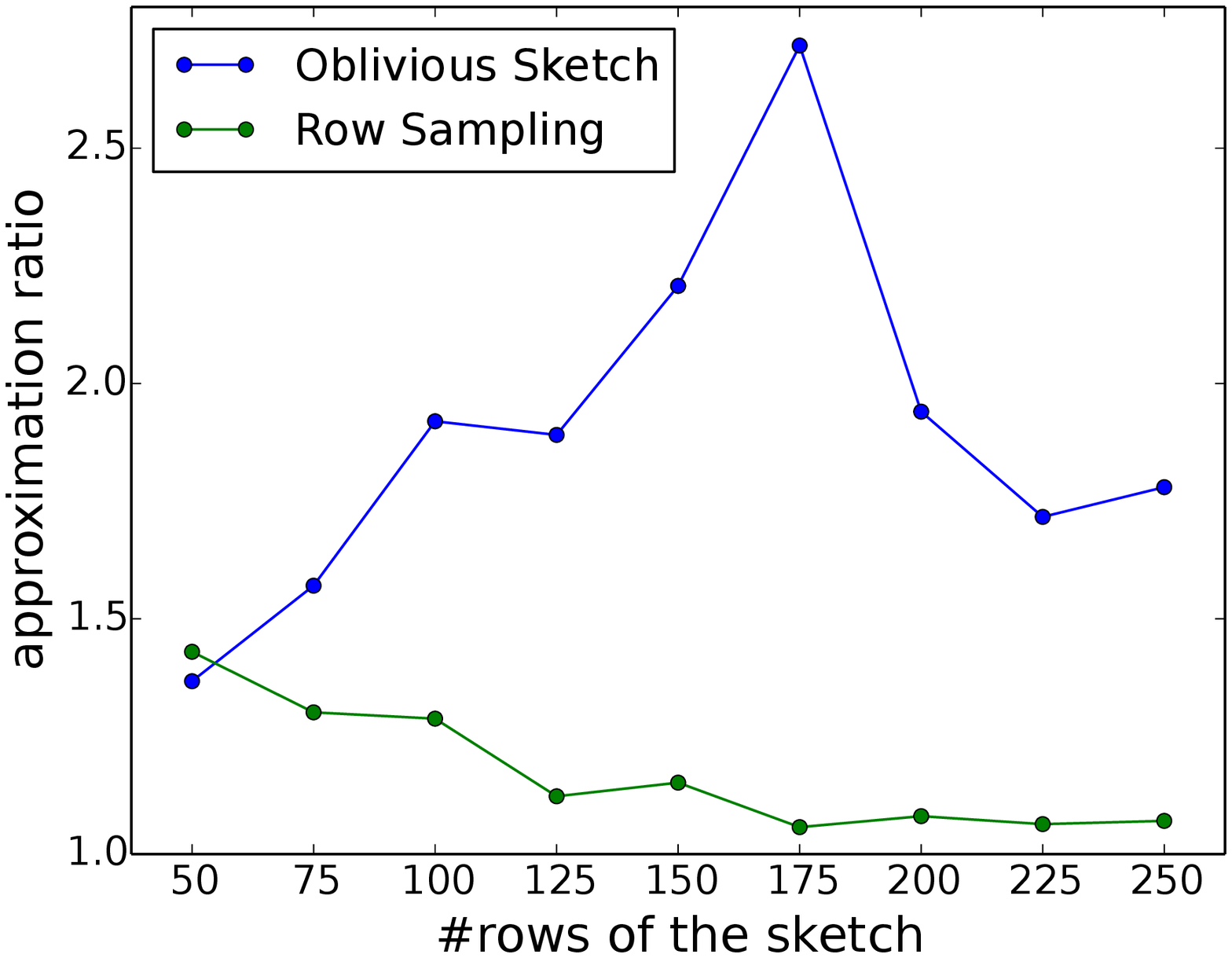}}      
	\subfigure[CT Slice Localization. $d = 384$, $\tau = 100$. ]
	{
		\includegraphics[width=0.233\textwidth]{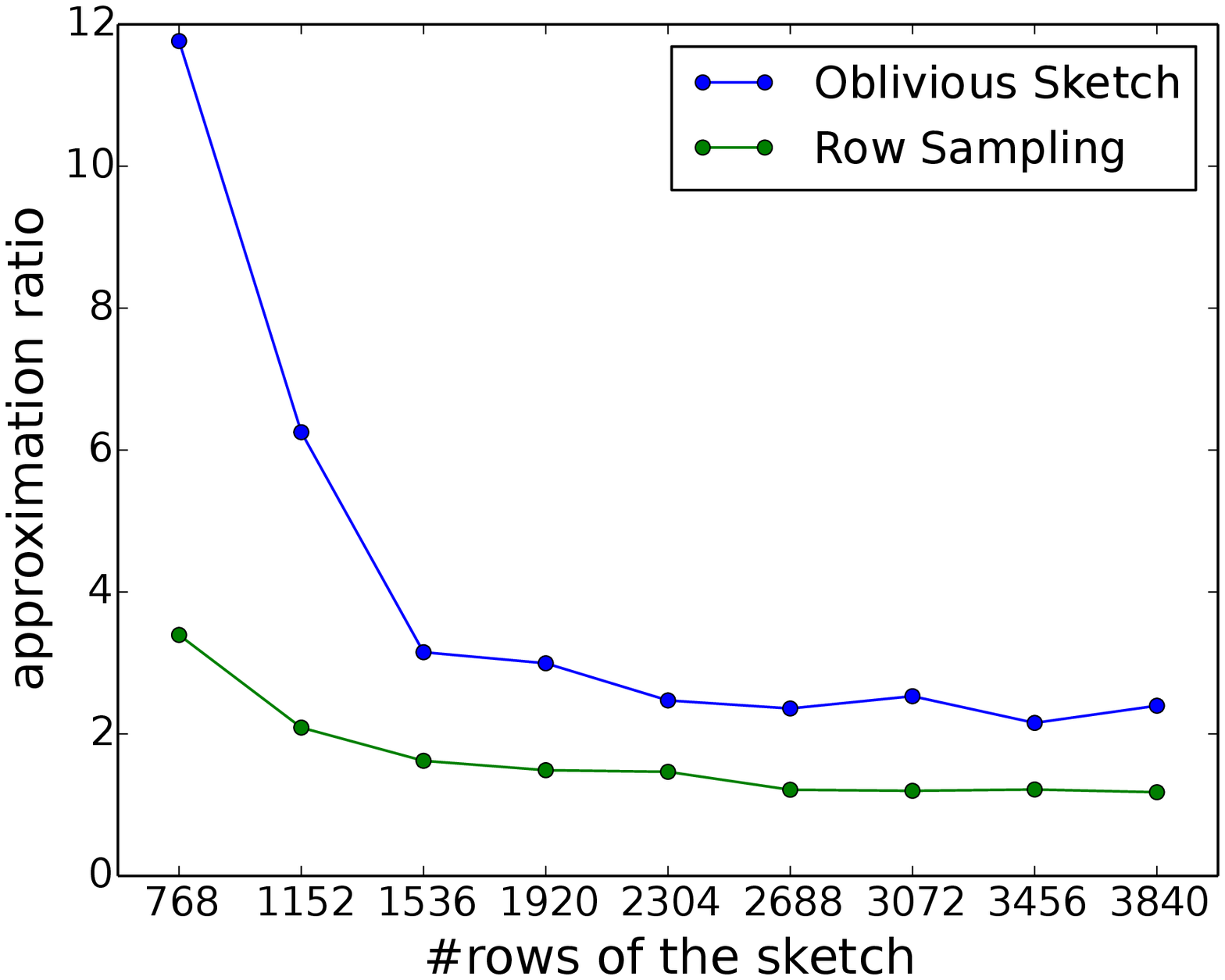}}   
	\subfigure[CT Slice Localization (with outliers). $d = 384$, $\tau = 1000$.]
	{
		\includegraphics[width=0.233\textwidth]{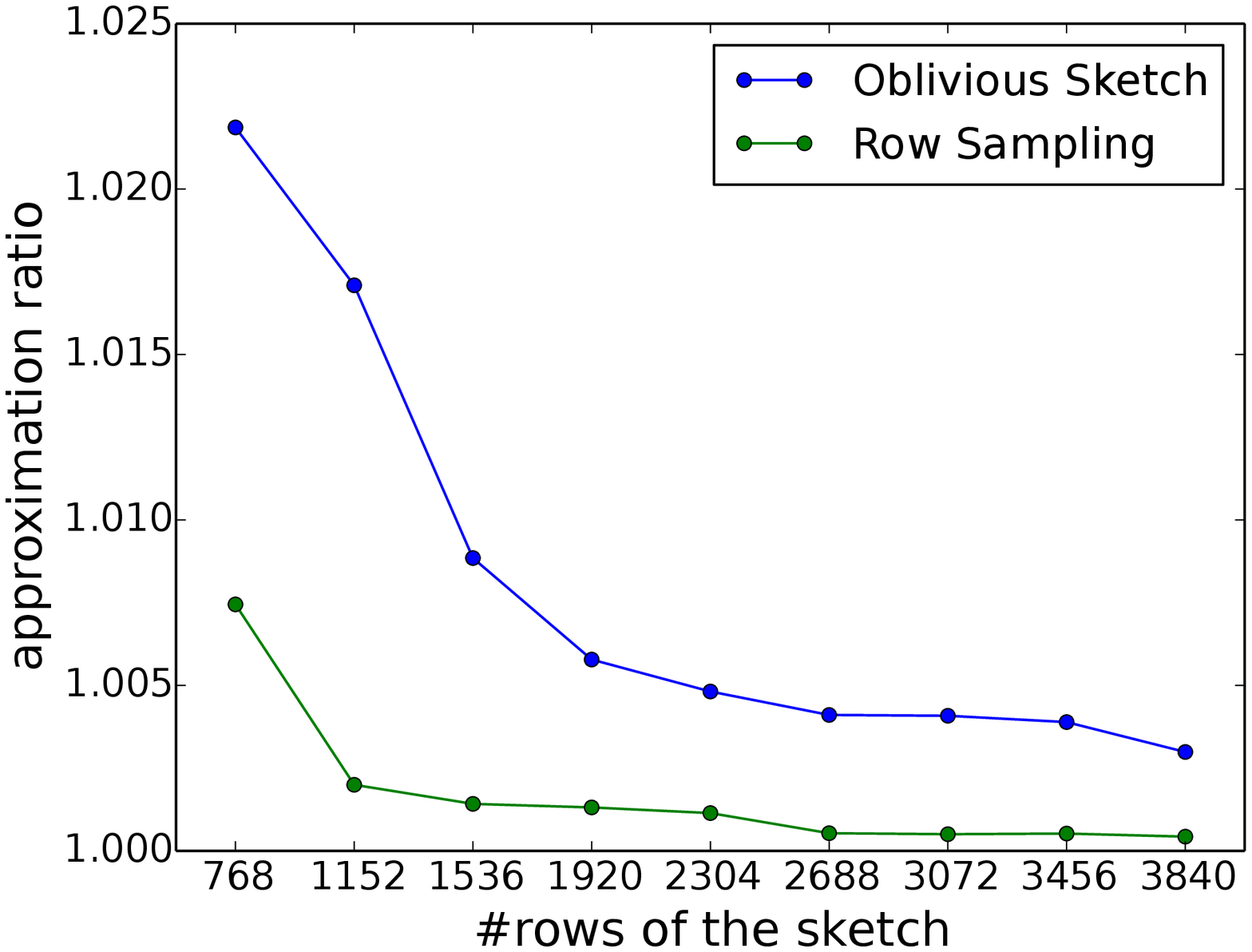}}   
	\caption{Approximation ratio of our dimensionality reduction methods. }
	\label{fig:exp}
\end{figure}

\paragraph{Discussions}
As can be seen from Figure \ref{fig:exp}, the row sampling algorithm has better approximation ratio as we increase the sketch size. 
The same is not always true for the oblivious sketch, since the oblivious sketch only guarantees an $O(\log n)$ approximation instead of a $(1 + \varepsilon)$-approximate solution, as returned by the row sampling algorithm. 
Moreover, the row sampling algorithm consistently outperforms the oblivious sketch in the experiments, except for extremely small sketch sizes (around $2d$). 
However, applying the oblivious sketch requires only one pass over the input, and the distribution of the sketching matrices does not depend on the input.
These advantages make the oblivious sketch preferable in streaming and distributed settings. 
Another advantage of the oblivious sketch is its simplicity.

Our empirical results demonstrate the practicality of our dimensionality reduction methods.
Our methods successfully reduce a Tukey regression instance of size $10000 \times d$ to another instance with $O(d)$ rows, without much sacrifice in the quality of the solution.
In most cases, the row sampling algorithm reduces the size to $3d$ rows while retaining an approximation ratio of at most $2$.

\section{Conclusions}
We give the first dimensionality reduction methods
for the overconstrained Tukey regression problem. 
We first give a row sampling algorithm which takes $\widetilde{O}(\nnz(A) + \poly(d \log n /\varepsilon))$ time
to return a weight vector with $\poly(d \log n /\varepsilon)$ non-zero entries,
such that the solution of the resulting weighted Tukey regression problem gives a 
$(1 + \varepsilon)$-approximation to the original problem. 
We further give another way to reduce Tukey regression problems
to smaller weighted versions, via an
oblivious sketching matrix $S$, applied in a single pass over the data.
Our dimensionality reduction methods are simple and easy to implement,
and we give empirical results demonstrating their practicality.
We also give hardness results showing that the Tukey regression problem cannot be efficiently solved in the worst-case. 

From a technical point of view, our algorithms for finding heavy coordinates and our structural theorem seem to be of independent interest. 
We leave it as an intriguing open problem to find more applications of them.

\section*{Acknowledgements} The authors would like to thank Lijie Chen and Peilin Zhong for helpful discussions, and the anonymous ICML reviewers for their insightful comments.

\bibliography{main}
\end{document}